\documentclass[onecolumn,a4paper,unpublished,11pt]{quantumarticle}
\pdfoutput=1

\usepackage[utf8]{inputenc}
\usepackage{graphicx}
\usepackage{cellspace}
\usepackage{hyperref}
\usepackage{amsmath}
\usepackage{amssymb}
\usepackage{xcolor}
\usepackage{amsthm}
\usepackage{subcaption}
\usepackage{cleveref}
\usepackage{tikz}
\usepackage{quantikz}
\usepackage[numbers,sort&compress]{natbib}
\usepackage{makecell}

\newcommand{\CXSWAP}{\mathrm{CXSWAP}}

\newcommand{\CZSWAP}{\mathrm{CZSWAP}}
\newcommand{\CPSWAP}{\mathrm{CPSWAP}}
\newcommand{\iSWAP}{\mathrm{iSWAP}}
\newcommand{\CX}{\mathrm{CX}}
\newcommand{\SWAP}{\mathrm{SWAP}}   
\newcommand{\Z}{\mathbb{Z}}
\newcommand{\SPAN}{\mathrm{SPAN}}
\newcommand{\modulo}{\mathrm{mod\;}}

\newtheorem{theorem}{Theorem}
\newtheorem{proposition}{Proposition}
\newtheorem{corollary}{Corollary}

\title{Directional Codes for Quantum LDPC Error Correction on Square and Hexagonal Grid Architectures}

\author{Gy\"{o}rgy P. Geh\'{e}r}
\email{george.geher@riverlane.com, gehergyuri@gmail.com}
\affiliation{Riverlane, St.~Andrew's House, 59 St.~Andrew's Street, Cambridge CB2 3BZ, United Kingdom}
\affiliation{All three authors contributed equally.}
\author{David Byfield}
\email{david.s.byfield@gmail.com}
\affiliation{Riverlane, St.~Andrew's House, 59 St.~Andrew's Street, Cambridge CB2 3BZ, United Kingdom}
\affiliation{All three authors contributed equally.}
\author{Archibald Ruban}
\email{archibald.ruban@riverlane.com, archibald.ruban@gmail.com}
\affiliation{Riverlane, St.~Andrew's House, 59 St.~Andrew's Street, Cambridge CB2 3BZ, United Kingdom}
\affiliation{All three authors contributed equally.}

\begin{document}

\maketitle

\begin{abstract}
Utility-scale quantum computing requires quantum error correction (QEC) to protect quantum information against noise. Currently, superconducting hardware is a promising candidate for achieving fault tolerance due to its fast gate times and feasible scalability. However, it is often restricted to two-dimensional nearest-neighbour connectivity, and therefore the variety of quantum low-density parity-check (qLDPC) codes that can be implemented on it without sacrificing QEC performance is believed to be greatly restricted. In this paper we construct a new family of qLDPC codes, which we call ``directional codes'', that outperforms the rotated toric code (RTC) while satisfying the connectivity requirements of the widely adopted square-grid, and some even the sparser hexagonal-grid, on a torus. The key idea is to utilise the iSWAP gate -- a native gate demonstrated on superconducting qubits -- to construct circuits that measure the stabilisers of these qLDPC codes without the need for additional connections. We numerically evaluate the performance of directional codes, encoding four, six, twelve and eighteen logical qubits, using a common superconducting-inspired circuit-level Pauli noise model. We also compare them to the RTC and to the bivariate bicycle (BB) codes, currently the two most popular quantum LDPC code families. As a concrete example, when evaluated with the Tesseract decoder with short beam setting, the best directional code family investigated achieves the same logical error rate as the RTC at physical error rate $p=10^{-3}$ but requires only a quarter to a third of the number of physical qubits to encode a given number of logical qubits. Our discovery opens a novel direction in QEC code design, suggesting that complex high-connectivity hardware may not be necessary for low-overhead fault-tolerant quantum computation.
\end{abstract}

\section{Introduction}

Quantum computers have the potential to solve problems which would be intractable on a classical computer \cite{Feynman1986,Nielsen_Chuang_2010,Preskill2018quantumcomputingin,shors}. However, qubits are inherently very noisy (much noisier than classical bits), and therefore it is largely believed that quantum error correction (QEC) is an essential part of achieving useful quantum computation \cite{earl-review,terhal-review,Roffe03072019}. In QEC we use a number of physical qubits to encode a smaller number of logical qubits -- referred to as a QEC code -- by repeatedly measuring a set of Pauli-product operators (or stabilisers) on the physical qubits. If we gain enough information about the errors on the physical qubits, we can then protect the logical qubits and ultimately implement algorithms with them \cite{ft1,ft2,ft3}.

Superconducting qubits are a strong contender for useful quantum computing due to their typically very fast gates \cite{google2023,qec_below_threshold_google,delft-softinfo,rigetti-paper,alec-dynamic-demonstration,fujitsu-paper} and potential scalability \cite{google-roadmap,IBM_roadmap}. However, they usually come with the draw-back of restricted connectivity, meaning a two-qubit gate can be executed only between a limited number of pairs of qubits. Typically, each qubit is connected to at most four or three others, as is the case for e.g. the ``square-grid'' or ``hexagonal-grid'' connectivities, which have so far been the most commonly realised connectivities for quantum hardware demonstrations \cite{google2023,qec_below_threshold_google,delft-softinfo,rigetti-paper,alec-dynamic-demonstration,ibm-experiment}. This constraint poses a significant challenge for QEC code design, since on most hardware a syndrome extraction circuit \cite{fowler-surface-codes,fowler-surface-code-computation} is needed to measure a Pauli-product operator on multiple qubits. The standard, sometimes called ``bare-ancilla'', syndrome extraction circuit does this via the measurement of an ancilla qubit that is connected to the qubits that support the operator. Therefore the weight of stabilisers of QEC codes implementable on such hardware with this method is limited by the degree of connectivity of the hardware \cite{ldpc2}.

Quantum low-density parity-check (qLDPC) codes \cite{gottesman-ldpc,ldpc2,Breuckmann_2021} are an exciting class of QEC codes that promise to greatly outperform the toric and planar surface codes \cite{pantaleev-good-ldpc,tanner-ldpc}. However, implementing them on a superconducting device is usually challenging as they often require connections beyond nearest-neighbour (i.e. not part of the square-grid) and an increased degree of connectivity (i.e. each qubit needs to be connected to more than four qubits), see e.g. \cite{IBM_roadmap,qldpc-demonstration}. The rotated planar surface code (RPSC), defined on a plane of qubits, and the rotated toric code (RTC), defined on a torus of qubits, are two of the most popular and best performing qLDPC codes that can be implemented on a square-grid connectivity device. Naturally, adding further connections gives the flexibility of implementing qLDPC codes with higher encoding rates. A recent example of this are bivariate bicycle (BB) codes that require two additional and non-local connections at each qubit location, increasing the degree of connectivity from four to six, see \cite{Bravyi2024} and also \cite{mac-terhal-bb-codes}. However, it is worth noting that introducing extra connections has the trade-off of making the hardware noisier, see e.g. \cite{alec-dynamic-demonstration, tesseractdecoder}, which may degrade the QEC performance of the codes that require these additional connections. Therefore, novel ideas for implementing better performing qLDPC codes without the need for extra connections are desirable.

In this paper, we construct ``directional codes'', a new family of qLDPC codes that leverages the iSWAP gate to be implemented on a toric square-grid connectivity device while requiring significantly less physical qubits than the RTC to achieve the same logical error rate. Additionally, the syndrome extraction circuit of directional codes uses as many layers of two-qubit gates per QEC round as the weight of the stabilisers. The key idea in the construction of these codes is that iSWAP exchanges the state of the two qubits it acts on. This is generally considered a hindrance, as it makes circuit construction more complex \cite{MBG,alec-dynamic-demonstration}. However, as we show here, we may exploit this trait of the iSWAP gate to our advantage.

The controlled-phase gate CZ, and its local-Clifford equivalent variants, are currently the more common native entangling gates in superconducting hardware. This is partly because it is well-known how to build syndrome extraction circuits with them, and partly because CZ gates have benefited from more than a decade of sustained hardware optimisation \cite{cz_2014,cz_2019,cz_2019_2,cz_2020,experimental-iswap,cz_2021_2,alec-dynamic-demonstration}. However, the iSWAP gate is a competitive alternative, which has already been demonstrated natively on superconducting hardware \cite{alec-dynamic-demonstration,rigetti-iswap,experimental-iswap,experimental-iswap2}. In particular, recent tunable-coupler experiments have shown that parasitic $ZZ$ interactions and leakage can be strongly suppressed, yielding a $ZZ$-free iSWAP gate with reported two-qubit interaction fidelity of $99.87 \pm 0.23\%$ \cite{experimental-iswap}. Additionally, the iSWAP and CZ gates were compared directly in \cite{alec-dynamic-demonstration} as native entangling gates used in the syndrome extraction circuit of dynamic variants of the RPSC. Even though the device used was optimised for the CZ gate, the iSWAP-variant showed overall competitive error suppression and budget. In particular, on the one hand, the iSWAP gate indicated superior leakage error suppression compared to the baseline CZ gate implementation, without active leakage removal (DQLR). On the other hand, it is also worth noting that the iSWAP gate showed a significantly larger c-phase error, which would need further optimisation. All these recent developments make a strong case for studying iSWAP-based circuits. Furthermore, in \cite[Appendix G]{alec-dynamic-demonstration}, simulations of quantum processors are performed that indicate that reducing the connectivity of qubits on a superconducting chip may suppress one- and two-qubit gate errors. This makes a strong case for further exploration of the key idea of this paper, namely utilising iSWAP gates to relax hardware connectivity requirements, as an alternative to the costly and complex development of high-connectivity hardware, like that proposed in \cite{IBM_roadmap}.

The structure of our paper is as follows. In the next section, we explain our method to construct circuits that measure the stabilisers of directional codes, and provide a more rigorous explanation in \Cref{sec:app-method,sec:app-logicals}. We first explain the circuit construction on an infinite plane, and then wrap around a (possibly twisted) torus to obtain a Calderbank--Shor--Steane (CSS) code. In \Cref{sec:simulations}, we present simulation results for directional codes that have stabilisers of weight-$7$ and encode twelve logical qubits.  In \Cref{sec:app-sim_details_additional_results}, we give more details on the simulation method and present additional numerical results, including for other directional codes with weight-$5$, -$6$ and -$8$ stabilisers that encode four, six and eighteen logical qubits, respectively. Our numerical results indicate that directional codes outperform the RTC under circuit-level Pauli noise. For instance, at physical error rate $p=10^{-3}$, decoded with the Tesseract decoder with short beam setting, the best directional code family achieves the same logical error rate as the RTC while using a quarter to a third of the number of physical qubits. In \Cref{sec:discussion}, we conclude the paper with some discussion of our results and potential future work. This includes adding boundaries to enable implementation on planar rather than toric hardware.


\section{Code construction}
\label{sec:code-construct}
\begin{figure}[!htb]
    \centering
    \hfill
    \begin{subfigure}{0.15\textwidth}
        \centering
        \includegraphics[width=0.4\linewidth]{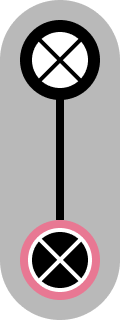}
        \caption{}
        \label{fig:cxswap_qubit_pair}
    \end{subfigure}
    \hfill
    \begin{subfigure}{0.70\textwidth}
        \centering
        \includegraphics[width=1.0\linewidth]{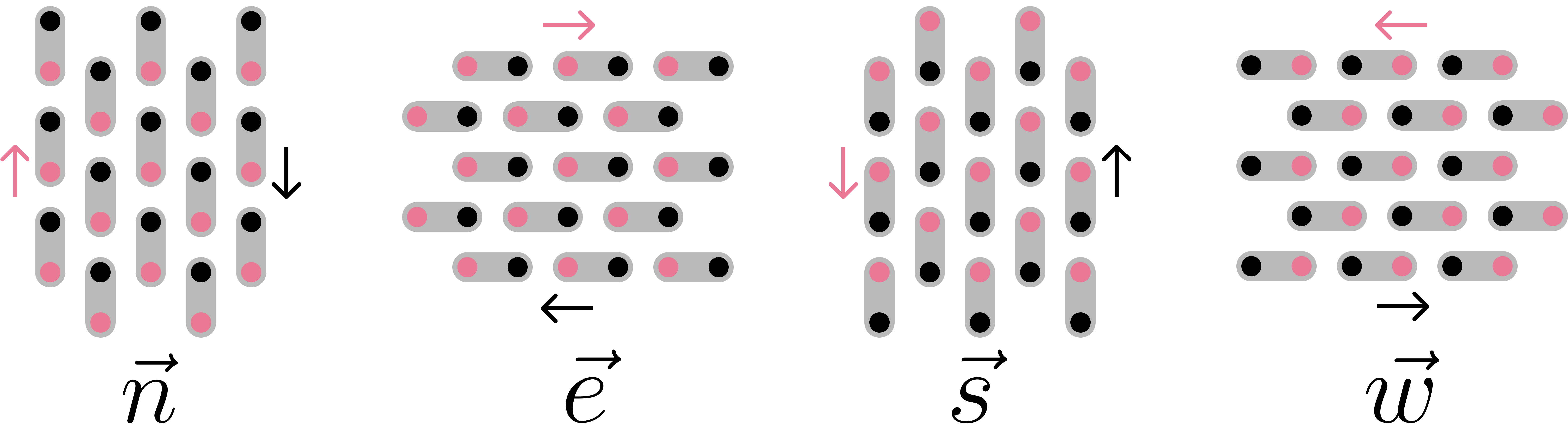}
        \caption{}
        \label{fig:directional_layers_with_arrows}
    \end{subfigure}
    \hfill
    \caption{The four types of CPSWAP layers used in the syndrome extraction of directional codes, each corresponding to a direction. (a) A zoomed-in view of the data-ancilla qubit pairs in (b), with overlaid CXSWAP gate. The bottom qubit is the ancilla qubit where the control of the gate is, and the top qubit is the data qubit where the target is. (b) The four types of CPSWAP layers forming the building blocks of the syndrome extraction circuit construction corresponding to $\vec{n}$ (north), $\vec{e}$ (east), $\vec{s}$ (south), and $\vec{w}$ (west) directions, respectively. Data and ancilla qubits are shown as black and red nodes, respectively. Red arrows indicate the direction in which the ancilla qubits uniformly move by one unit due to the action of the layer. The data qubits uniformly move in the opposite direction also by one unit, indicated by black arrows. Each grey-highlighted pair of qubits are involved in a CPSWAP gate in that layer. The CPSWAP's control is always the ancilla and the target is always the data qubit, as depicted in (a) for the CXSWAP case.}
    \label{fig:directional_layers}
\end{figure}

We construct directional codes by defining their associated syndrome extraction circuits, which by design satisfies the square-grid connectivity constraint, and some of them even the hexagonal-grid. Recall that the CXSWAP gate on control qubit $Q_0$ and target qubit $Q_1$ is defined by $\CXSWAP(Q_0, Q_1)|\psi\rangle = \SWAP(Q_0, Q_1)\CX(Q_0, Q_1)|\psi\rangle$, meaning it is equivalent to applying a CX gate followed by a SWAP gate. The CZSWAP is defined similarly, and for brevity we will refer to these two gates as CPSWAP gates. It is well-known that CPSWAP gates are local-Clifford equivalent to the iSWAP gate \cite[Section A.5]{MBG}; more precisely, we can find one-qubit unitary Clifford gates $A_0$, $A_1$, $B_0$, $B_1$ such that $\CPSWAP(Q_0, Q_1) = A_0\otimes A_1 \cdot \iSWAP(Q_0, Q_1) \cdot B_0 \otimes B_1$. Therefore, if we construct a syndrome extraction circuit in terms of CPSWAP gates, obtaining one that uses iSWAPs instead is a straightforward compilation that does not change the qubit connections required. For benchmarking directional codes we use circuits compiled to iSWAP, see \Cref{sec:app-simulation_details}. However, for describing our code construction, it is more natural to use CPSWAP gates because we can think of a layer of CPSWAPs as being a layer of controlled-Pauli gates up to permutation of the quantum states of the physical qubits. As such, each CPSWAP layer is effectively dynamically altering the connectivity -- a feature central to our construction.

Consider the infinite square-grid lattice $\Z^2$ on which we lay out data and ancilla qubits in a chequerboard pattern. More precisely, define $\Z^2_{\mathrm{data}} := \{(x,y)\in\Z^2 \colon x\equiv y \; (\mathrm{mod} \; 2)\}$ and $\Z^2_{\mathrm{anc}} := \Z^2\setminus\Z^2_{\mathrm{data}}$. Note that any two qubits that are part of the same sublattice are not connected. Each ancilla qubit is connected to four data qubits positioned in the north: $\vec{n}:=(0,1)$, east: $\vec{e}:=(1,0)$, south: $\vec{s}:=(0,-1)$ and west: $\vec{w}:=(-1,0)$ directions. Let us pick one of these directions, say north, and consider the following circuit layer of CPSWAP gates: for each ancilla qubit, apply a CPSWAP between it as control and the data qubit north of it as target, see \Cref{fig:cxswap_qubit_pair}. Due to the SWAP part of the gates, we may think of each qubit as being interchanged with its pair as an effect of this layer. In this way, the sublattice of ancilla qubits is moved north uniformly by one unit, while the sublattice of data qubits is moved uniformly in the opposite direction, south, by one unit -- see \Cref{fig:directional_layers_with_arrows}. Consequently, each ancilla qubit is now connected to three data qubits it was not connected to before, while it remains connected to the data qubit that has just been involved in the same CPSWAP gate. Therefore, even though the relative location of ancilla/data qubits did not change, the connectivity structure between the two sublattices changed considerably. This is a crucial observation that is the core idea for our code construction, and which also holds for the other three directions. We associate to each of these four CPSWAP layers an element of $\{ \vec{n}, \vec{e}, \vec{s}, \vec{w} \}$ (see \Cref{fig:directional_layers_with_arrows}).

\begin{figure*}[!htb]
    \centering
    \includegraphics[width=1.0\linewidth]{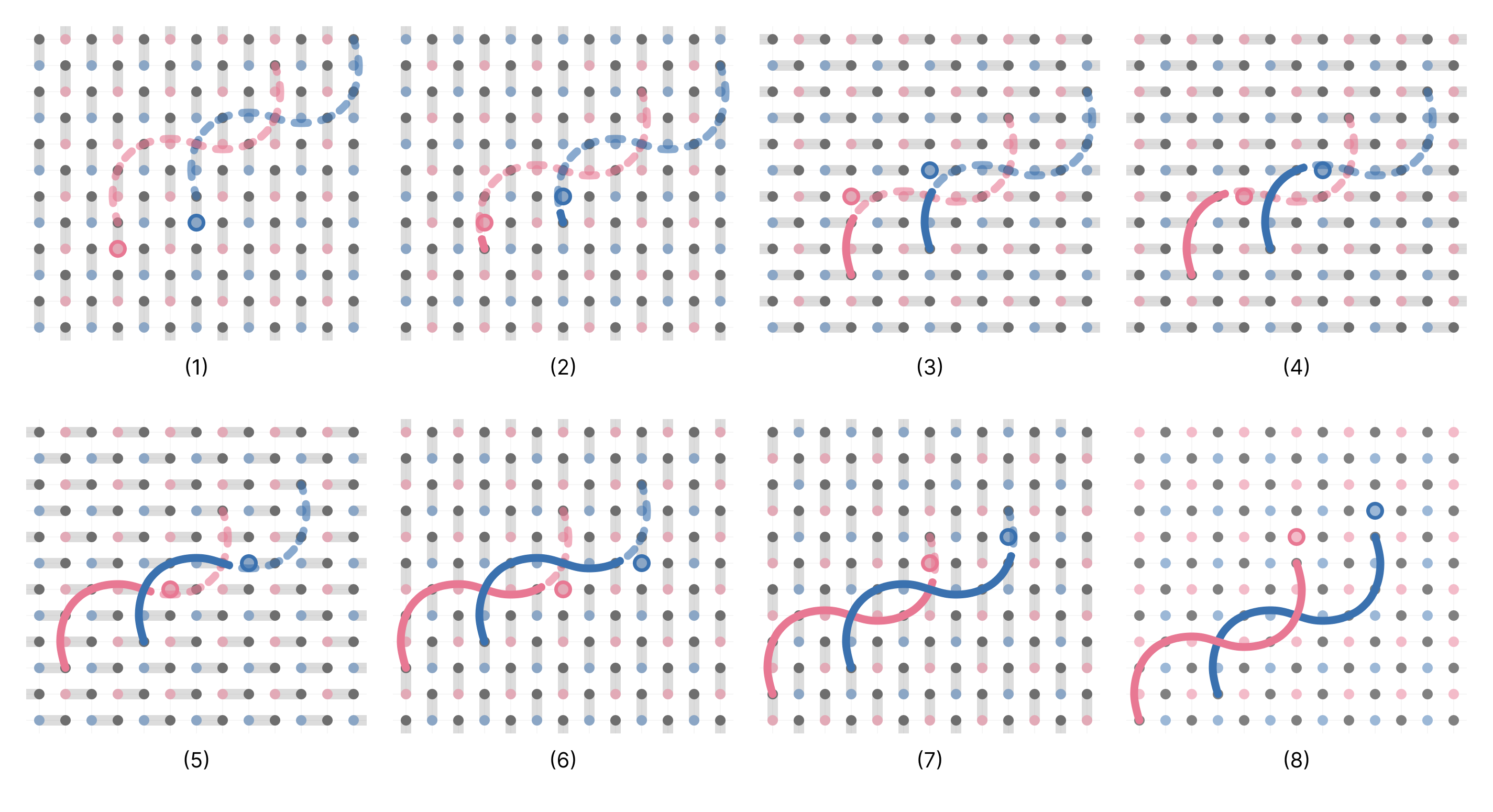}
    \caption{The entangling layers of the syndrome extraction circuit that measures the stabilisers of an $N^2E^3N^2$-code using CPSWAP gates (see also \Cref{thm:Thm1}). The ancilla qubit resets and measurements take place before and after the layers shown, respectively. Data qubits are shown as black dots and ancilla qubits corresponding to $X$-/$Z$-stabilisers are depicted as red/blue dots. Each grey line connecting a data-ancilla qubit pair represents a CPSWAP between the two qubits, with the ancilla/data qubit always being the control/target qubit. A single $X$- and a single $Z$-stabiliser are shown in red and blue, respectively, as snake-like shapes supported on all the data qubits they touch. Their respective ancilla qubit is shown circled in red and blue respectively. The solid part of each snake-like shape represents the part of the stabiliser that has already been entangled with the ancilla qubit, and the dashed part the rest of the stabiliser. All panels are shown in the hardware's frame of reference. The data qubit sub-lattice (and stabilisers) move uniformly in one direction, and the ancilla qubit sub-lattice in the opposite direction, due to the action of the SWAP-part of each CPSWAP gate. Panels (1)--(2) show the first two north layers, panels (3)--(5) the subsequent three east layers, panels (6)--(7) the last two north layers, and panel (8) the final configuration of the qubits and stabilisers.}
    \label{fig:N2E3N2_syndrome_extraction}
\end{figure*}

By applying CPSWAP layers according to a sequence of directions, we are able to measure stabilisers of arbitrary weight that all have the same shape. We do this by adapting the ``bare-ancilla'' syndrome extraction circuit method to CPSWAP gates. As an example, allocate the ancilla qubits initially on $\Z^2_{\mathrm{anc}}$ such that we have alternating rows (and columns) of $X$- and $Z$-type, just as for the surface code (see \Cref{fig:N2E3N2_syndrome_extraction} and \Cref{fig:N2E3N2_on_parallelogram}). For each CPSWAP gate, we always apply CXSWAP if the gate interacts with an $X$-ancilla qubit, and CZSWAP otherwise. Consider the circuit where we first reset all ancilla qubits in the $|+\rangle$ state, then apply a sequence of $\ell$ CPSWAP layers (for $\ell$ a positive integer) each corresponding to an element of $\{\vec{n}, \vec{e}, \vec{s}, \vec{w}\}$ (see \Cref{fig:directional_layers_with_arrows}), and finally measure the ancilla qubits in the $X$ basis. This circuit satisfies the square-grid connectivity and, provided that the sequence of CPSWAP layers fulfils the conditions of \Cref{thm:Thm1} laid out in \Cref{sec:app-method}, measures independently and simultaneously one stabiliser per ancilla qubit. All such stabilisers are of weight $\ell$, have the same shape and scheduling, and are composed either of $X$ or $Z$ Pauli terms only, as determined by the $X$- or $Z$-type of the associated ancilla qubit. We give more details on this in \Cref{sec:app-method} -- in particular, \Cref{thm:Thm1} characterises which allocations of ancilla qubits give rise to a valid circuit, given CPSWAP layers corresponding to a sequence of directions.

Without loss of generality, we may only consider sequences that begin with some number of $\vec{n}$ layers followed by some number of $\vec{e}$ layers. In this paper, all directional codes explicitly constructed and simulated have the following structure. Let $\alpha\geq 1,\beta\geq 2$ be integers, allocate the data and ancilla qubits as described above, and consider the stabilisers obtained with the following choice of \mbox{CPSWAP} layers: $\alpha$ many $\vec{n}$ layers, followed by $\beta$ many $\vec{e}$ layers and finally $\alpha$ many $\vec{n}$ layers. As shown in \Cref{cor:NaEbNa} of \Cref{sec:app-method}, these sequences of directions satisfy the conditions of \Cref{thm:Thm1}. The resulting stabilisers have weight $2\alpha+\beta$, and we call them $N^\alpha E^\beta N^\alpha$-stabilisers. \Cref{fig:N2E3N2_syndrome_extraction} is an illustration of the $(\alpha,\beta)=(2,3)$ case, which has weight-$7$ stabilisers.

\begin{figure*}[!htb]
    \centering
    \begin{subfigure}{0.45\textwidth}
        \centering
        \includegraphics[width=1.0\linewidth]{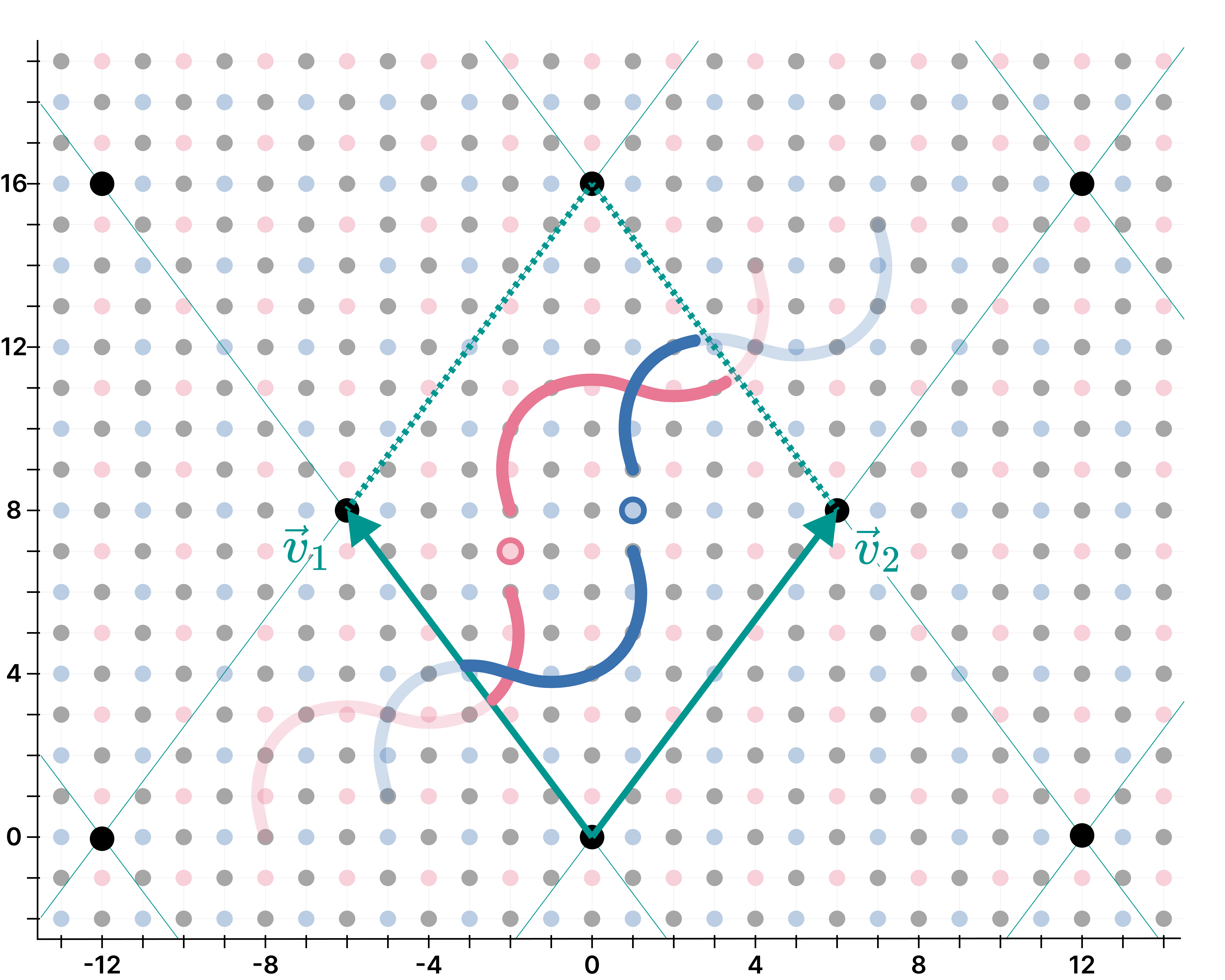}
        \caption{}
        \label{fig:N2E3N2_rotated}
    \end{subfigure}
    \begin{subfigure}{0.45\textwidth}
        \centering
        \includegraphics[width=1.0\linewidth]{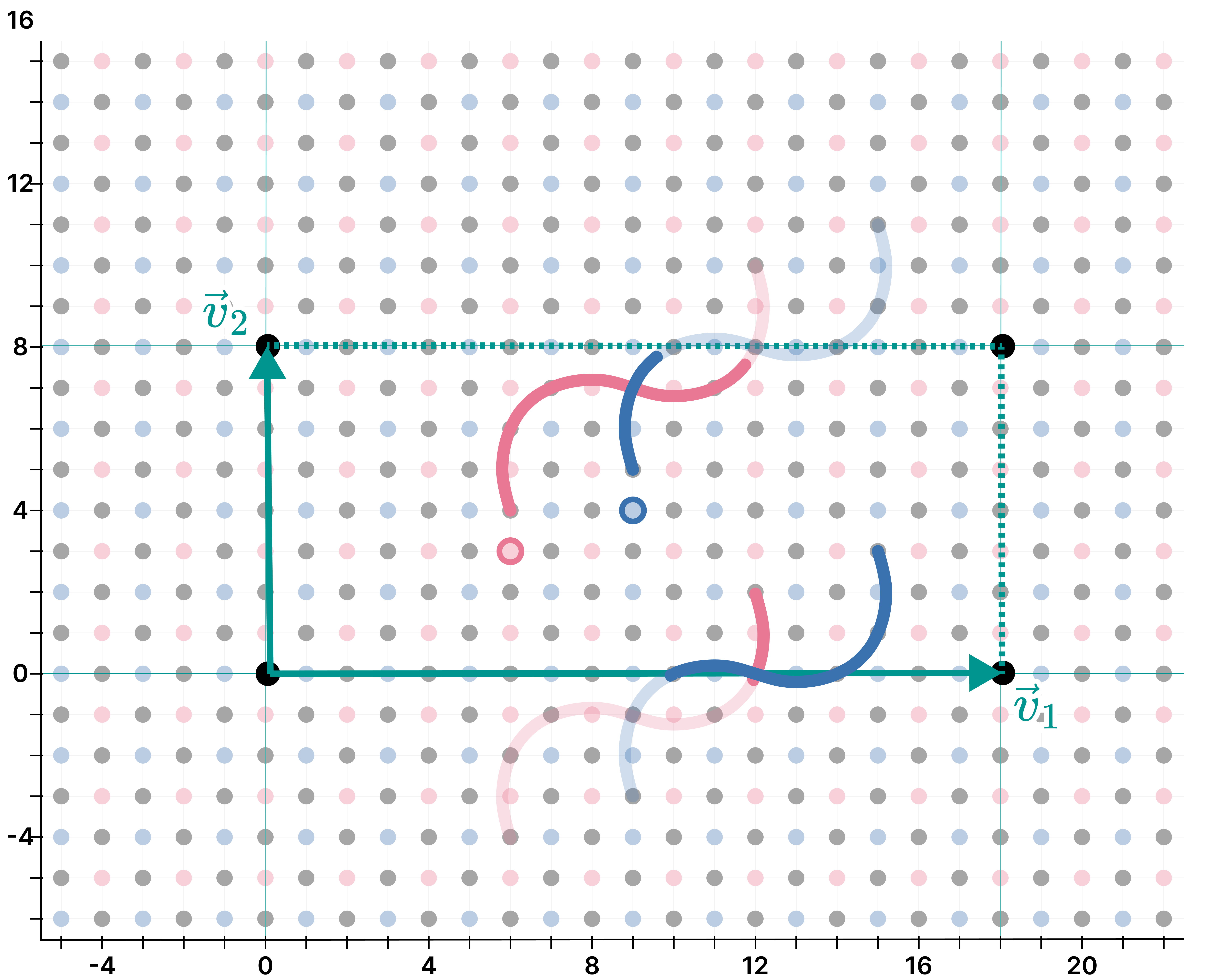}
        \caption{}
        \label{fig:N2E3N2_rectangular}
    \end{subfigure}
    \caption{Illustration of (a) the rotated and (b) the rectangular $N^2E^3N^2$-codes of distance $4$. Two of the stabilisers, one $X$-type in red and one $Z$-type in blue, are shown in both cases, each with two equivalent representations. The transparent part of each snake-like shape indicates where the stabiliser leaves the parallelogram and wraps around $\mathcal{P}$. The associated sublattice $\mathcal{K}$ is also shown, as large black points. (a) The rotated code is defined on the twisted torus corresponding to the diamond-shaped parallelogram $\mathcal{P}(\vec{v}_1=(-6,8),\vec{v}_2=(6,8))$, giving the $[\![48, 12, 4]\!]$ code from \Cref{tab:code_list}. (b) The rectangular code is defined on the regular torus corresponding to $\mathcal{P}(\vec{v}_1=(18,0),\vec{v}_2=(0,8))$, giving the $[\![72, 12, 4]\!]$ code from \Cref{tab:code_list}.} 
    \label{fig:N2E3N2_on_parallelogram}
\end{figure*}

We may define an $N^\alpha E^\beta N^\alpha$-code by wrapping the infinite plane, with the qubits and the $N^\alpha E^\beta N^\alpha$-stabilisers, around a parallelogram, resulting in a CSS code defined on a regular or twisted torus. We note that the same works for general direction sequences, see \Cref{sec:app-parallelograms}. For two linearly independent vectors $\vec{v}_1, \vec{v}_2\in\Z^2$ we define the parallelogram associated with them as 
\begin{equation}\label{eq:parallelogram_definition}
    \mathcal{P} = \mathcal{P}(\vec{v}_1, \vec{v}_2) := \left\{ P\in\Z^2 \colon \; \overrightarrow{OP} = a \vec{v}_1 + b \vec{v}_2 \text{ with some } 0 \leq a < 1,\ 0 \leq b < 1 \right\},
\end{equation}
where $O=(0,0)$. The torus that we obtain by wrapping around $\mathcal{P}$ is defined to be the quotient $\Z^2/\mathcal{K}$, where $\mathcal{K}:=\{\ell_1\vec{v}_1+\ell_2\vec{v}_2\colon \ell_1, \ell_2\in\Z\}$. In the special case where $\vec{v}_1$ and $\vec{v}_2$ are parallel to the $x$- and $y$-axes we call it a regular torus, otherwise a twisted torus. We note that, in the case where $(\alpha,\beta)=(1,2)$ and $\vec{v}_1 = (2d,0)$, $\vec{v}_2 = (0,2d)$, we obtain the standard unrotated (a.k.a. Kitaev) toric code (TC) with uniform scheduling that has code distance $d$ and full circuit-level distance $d_{circ}=d$ (see \Cref{sec:app-motivating_example} for the $d=3$ case). Additionally, for even $d$ and $\vec{v}_1 = (-d,d)$, $\vec{v}_2 = (d,d)$, we obtain the rotated toric code with a uniform schedule that has code distance $d$, but its circuit-level distance is reduced: $d_{circ}= \lceil\frac{3}{4}d\rceil$. However, it is possible to recover the full distance for the RTC by using the N--Z scheduling that is the standard for the RPSC with controlled-Pauli gates (see e.g. \cite[Figure 1]{MBG}). From now on, whenever we use the abbreviation RTC, we implicitly mean that it is implemented with N--Z scheduling using CX and CZ gates, giving full circuit-level distance $d_{circ}=d$. Importantly, this is what we benchmark directional codes against in \Cref{sec:simulations}, see also \Cref{sec:app-additional_results}.

Our numerical exploration found that the $N^\alpha E^\beta N^\alpha$-codes offer two regularly structured code families. We benchmarked the $(\alpha,\beta) \in\{(1,3),(2,2),(2,3),(2,4)\}$ cases for which we also observed strong error-correction capabilities. More precisely, for positive integer $d$ divisible by four, we have:
\begin{itemize}
    \item the ``rotated'' $N^\alpha E^\beta N^\alpha$-codes defined on a diamond-shaped twisted torus specified by $\vec{v}_1=(-\frac{1}{2}\beta d,\alpha d)$, $\vec{v}_2=(\frac{1}{2}\beta d,\alpha d)$. These codes have $n=\frac{1}{2}\alpha\beta d^2$ data qubits and $k\geq 2(2\alpha-1)(\beta-1)$ logical qubits, which we show in \Cref{sec:app-logicals}. In the $N^2E^4N^2$ case, this torus coincides with the torus on which the RTC with distance $2d$ is defined. We conjecture the code distance to be exactly $d$, which we verified for $d=4,8,12$ using integer programming for $(\alpha,\beta) \in\{(1,3),(2,2),(2,3),(2,4)\}$. Furthermore, except for the $NE^3N$ case, we conjecture that the circuit-level distance is $d_{circ} = d$. We verified this using integer programming for $d=4,8$, except for the $d=8$ $N^2E^4N^2$-code, which was too large for an exact distance calculation.
    \item the ``rectangular'' $N^\alpha E^\beta N^\alpha$-codes defined on a regular torus specified by $\vec{v}_1=(\frac{3}{2}\beta d,0)$, $\vec{v}_2=(0,\alpha d)$, with $n=\frac{3}{4}\alpha\beta d^2$ and $k\geq 2(2\alpha-1)(\beta-1)$. We conjecture the code distance to be exactly $d$, which we verified for $d=4,8,12$ using integer programming for $(\alpha,\beta) \in\{(1,3),(2,2),(2,3),(2,4)\}$. Furthermore, we also conjecture the circuit-level distance to be exactly $d$, which we verified using integer programming for $d=4,8$, except for the $d=8$ $N^2E^4N^2$-code, which was too large for an exact distance calculation.
\end{itemize}
The above defined families of $NE^3N$-, $N^2E^2N^2$-, $N^2E^3N^2$- and $N^2E^4N^2$-codes offer exactly four, six, twelve and eighteen logical qubits, respectively. We emphasise that when $d$ is not divisible by four, the above definitions either do not give a valid CSS code, or they provide fewer logical qubits. We also note that if one applies the above definition in the $NE^2N$ case, then they are well-defined for even $d$, but in the rectangular case one does not obtain the (unrotated) TC, but instead a toric code that is defined on a twisted torus when implemented with controlled-Pauli gates. For $NE^3N$-codes we only benchmarked the rectangular family, since as noted above the rotated $NE^3N$-codes do not achieve full circuit-level distance. For $N^2E^2N^2$- and $N^2E^3N^2$-codes we benchmarked both rotated and rectangular families (see \Cref{fig:N2E3N2_on_parallelogram} for an illustration of the $N^2E^3N^2$ families). As for the $N^2E^4N^2$-codes, we do not include the rectangular family because initial simulations suggested that the decoder setting used across all simulations led here to a large overestimation of the obtained logical error probabilities. This overestimation emerges as the circuits simulated get particularly large (see \Cref{sec:app-short_beam_big_circuits} for evidence of this effect). 

Due to the sizes of the circuits being too large to decode in reasonable time with the Tesseract decoder \cite{tesseractdecoder} for $d=12$ and beyond, we define a ``filler'' code as the $d=6$ alternative for both the rotated and rectangular families. This is obtained in a similar way to how the $n=144$ BB (gross) code is obtained from the $n=72$ BB code in \cite{Bravyi2024}. More precisely,
\begin{itemize}
    \item the ``filler'' $N^\alpha E^\beta N^\alpha$-code is defined on a regular torus obtained from the $d=4$ rectangular case by doubling its vertical side, i.e. $\vec{v}_1=(6\beta,0)$, $\vec{v}_2=(0,8\alpha)$. It has $n = 24\alpha\beta = \frac{2}{3} \alpha\beta d^2$ data qubits, and the number of logical qubits is the same as in the rotated and rectangular cases. The code distance is $d=6$ in all four cases, and the circuit-level distance is $d_{circ}=6$ in all but the $NE^3N$ case when it is $5$.
\end{itemize}
Note that the quantity $\frac{n}{d^2}$ is $\frac{1}{2}\alpha\beta$, $\frac{2}{3}\alpha\beta$ and $\frac{3}{4}\alpha\beta$ for the rotated, filler and rectangular codes, respectively. In the Appendix we provide a list of all directional codes that we benchmarked in \Cref{tab:code_list}, and in \Cref{sec:app-logicals} we explain the structure of the logical operators of the rotated family in detail. 


\section{Simulation results}
\label{sec:simulations}
In this section, we focus on comparing the quantum memory performance of the $N^2E^3N^2$-codes, with weight-$7$ stabilisers providing twelve logical qubits, against the RTC, TC, RPSC, and the BB codes. For all codes, we prepared quantum memory circuits with as many QEC rounds as the distance of the code. We added noise according to the SI-$1000$ circuit-level Pauli noise model \cite{si1000-1,si1000-2}, including at the logical state preparation and measurement. We sampled each circuit with stim \cite{stim} and decoded with Tesseract using the short beam setting \cite{tesseractdecoder}. From each sampled logical error probability we then calculated the logical error rate (per round) denoted by $p_L$, which is plotted in all plots. The error bars shown for the data points correspond to the Clopper-Pearson interval \cite{clopper-pearson} with $95\%$ confidence. We give the details in \Cref{sec:app-simulation_details}, and provide the plots for the other directional codes in \Cref{sec:app-additional_results}.

\begin{figure}[!htb]
     \centering
     \begin{subfigure}{0.45\textwidth}
         \centering
         \includegraphics[width=1.0\linewidth]{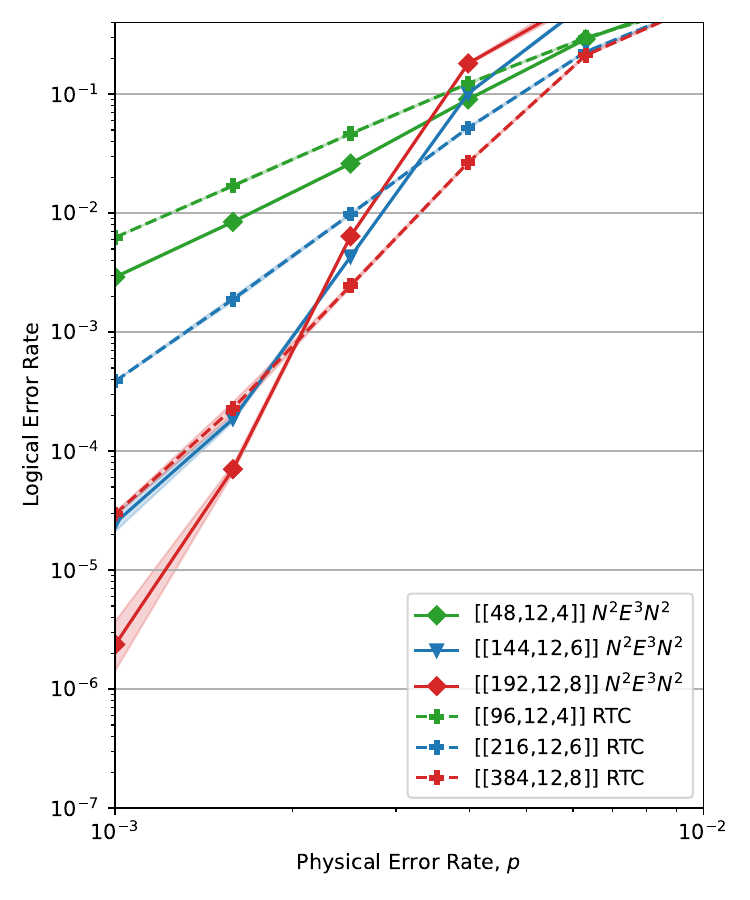}
         \caption{rotated $N^2E^3N^2$ vs RTC}
         \label{fig:N2E3N2_rotated_vs_RTC}
     \end{subfigure}
     \begin{subfigure}{0.45\textwidth}
         \centering
         \includegraphics[width=1.0\linewidth]{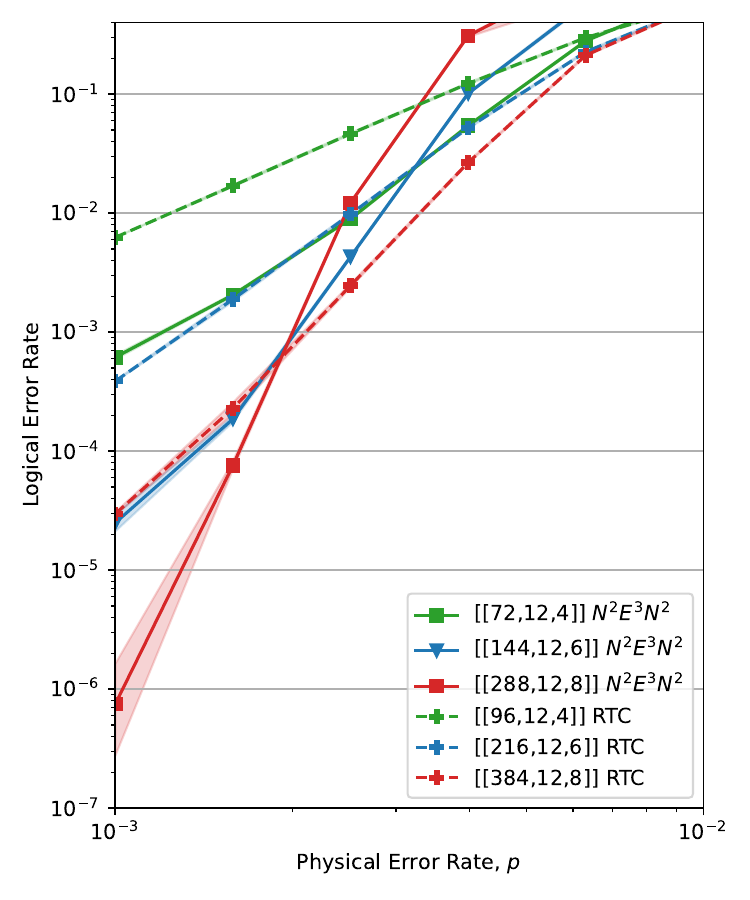}
         \caption{rectangular $N^2E^3N^2$ vs RTC}
         \label{fig:N2E3N2_rectangular_vs_RTC}
     \end{subfigure}
     \caption{Comparison of $N^2E^3N^2$-codes against six copies of the RTC, giving twelve logical qubits for both. The $N^2E^3N^2$-codes correspond to solid lines with diamond (rotated), square (rectangular) and triangle (filler) markers, and the RTCs to dashed lines with cross markers. Each distance is assigned a unique colour, and codes of equal distance are plotted with the same colour. (a) The rotated ($d=4,8$) and the filler ($d=6$) $N^2E^3N^2$-codes compared against the RTC. (b) The rectangular ($d=4,8$) and the filler ($d=6$) $N^2E^3N^2$-codes compared against the RTC.}
     \label{fig:N2E3N2_vs_RTC}
\end{figure}

In \Cref{fig:N2E3N2_rotated_vs_RTC}, we compare the rotated $N^2E^3N^2$-codes, including the filler code, with $d=4,6,8$ against six independent copies of the RTC of the same distances. The number of encoded logical qubits is therefore twelve for both. Recall that the syndrome extraction circuit of both codes satisfies the square-grid connectivity on a torus. The physical error rate varies between $p=10^{-3}$ and $10^{-2}$ on the $x$-axis. The rotated $N^2E^3N^2$-codes use half the number of physical qubits to achieve the same distance as the RTC. However, the logical error rates for directional codes drop more steeply as $p$ decreases (at least for the simulated $10^{-3} \leq p \leq 10^{-2}$) giving more than an order of magnitude lower logical error rate for $d=6,8$ by $p=10^{-3}$. \Cref{fig:N2E3N2_rectangular_vs_RTC} shows a similar tendency for the rectangular case. 

The reader may notice that the gaps between the data of the $d=4,6$ $N^2E^3N^2$-codes in \Cref{fig:N2E3N2_vs_RTC} are different from the gap between the $d=6,8$. We believe there are two reasons for this. Firstly, as we explained earlier, the $d=6$ filler code used for both plots is strictly speaking not a member of the rotated or rectangular family. A similar difference in the gaps between the data of the $n=72,144$ and $n=144,288$ BB codes is also present in \cite[Figure 2a]{Bravyi2024}. Secondly, as was mentioned in \cite{tesseractdecoder}, and is visible in Figure 2 and 6 there, the short beam setting in Tesseract is enough to obtain close to optimal decoding results for the RPSC up to $d\leq 9$, or equivalently $n\leq 81$. In \Cref{sec:app-short_beam_big_circuits} we present some plots related to the TC that demonstrate that as the number of data qubits $n$ gets larger, the short beam setting gives noticeably higher logical error rates than the long beam setting. As such, this may cause the logical error rates of the larger directional codes to be overestimated. Nonetheless, it is clear from \Cref{fig:N2E3N2_vs_RTC} that the $N^2E^3N^2$-codes outperform the RTC, even with the potentially suboptimal short beam setting of the Tesseract decoder.

\begin{figure}[!htb]
    \centering
     \begin{subfigure}{0.6\textwidth}
         \centering
         \includegraphics[width=1.0\linewidth]{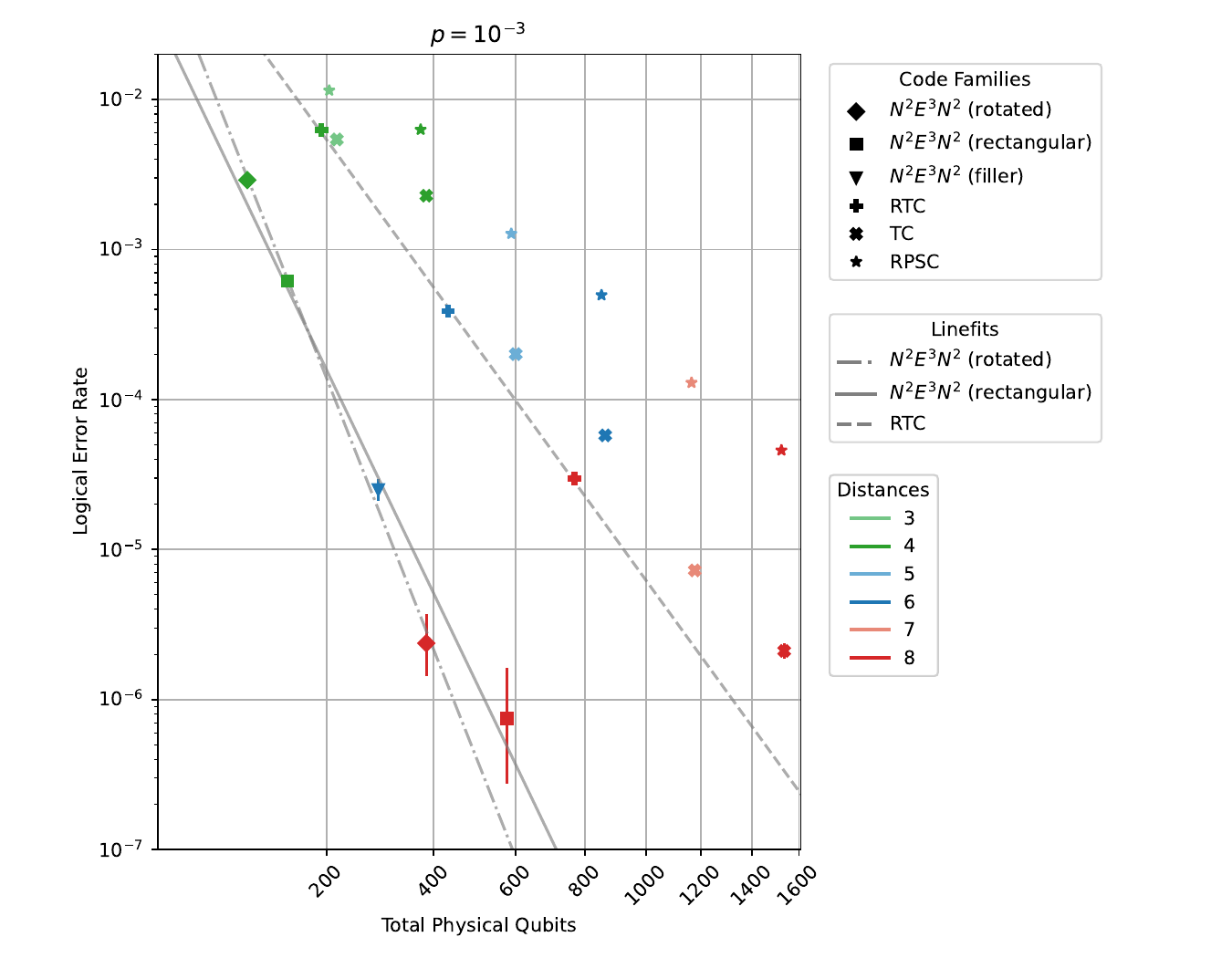}
         \caption{$N^2E^3N^2$ vs three types of surface codes at $p=10^{-3}$}
         \label{fig:N2E3N2_fixed_p}
     \end{subfigure}
     \begin{subfigure}{0.35\textwidth}
         \centering
         \includegraphics[width=1.0\linewidth]{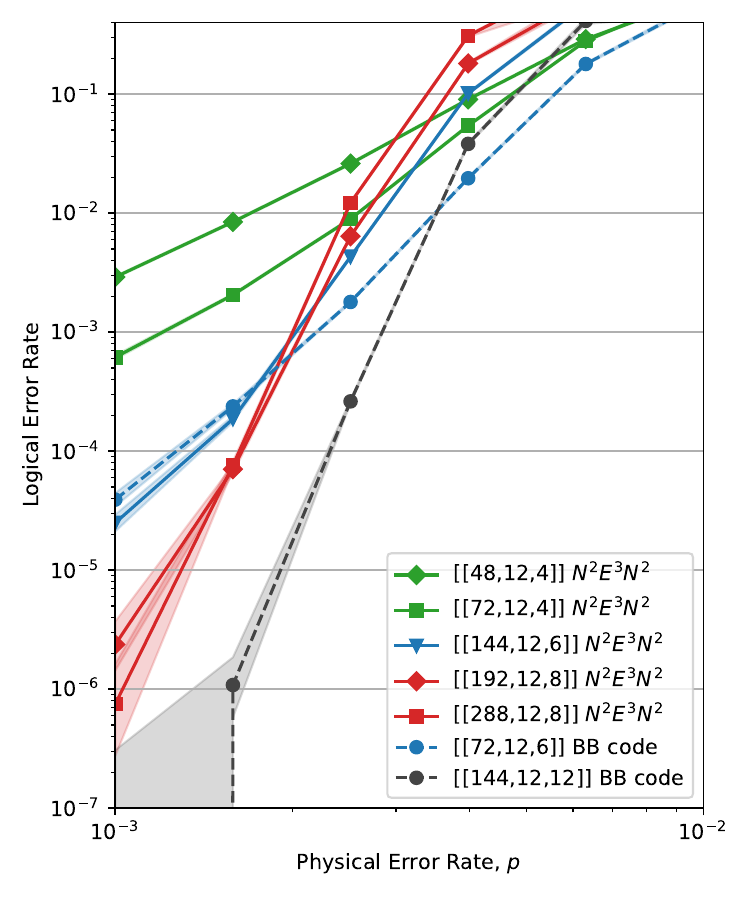}
         \caption{$N^2E^3N^2$ vs BB}
         \label{fig:N2E3N2_vs_BB}
     \end{subfigure}
    \caption{(a) Comparison of all five $N^2E^3N^2$-codes against the RTC, TC and RPSC where we fix $p=10^{-3}$ and plot the total number of physical qubits on the $x$-axis, using square-root scale, against the logical error rate on the $y$-axis, using logarithmic scale. We also plot three lines that are fitted on the data obtained for the rotated (dashed-dotted) and rectangular (solid) $N^2E^3N^2$-codes, including the filler code for both, and the RTC (dashed).
    (b) The $n=72, 144$ BB codes (dashed blue and black lines, respectively, with dot markers) are compared against the five $N^2E^3N^2$-codes. All codes encode twelve logical qubits.}
    \label{fig:N2E3N2_fixed_p_and_vs_BB}
\end{figure}

Next, we wish to characterise quantitatively how much benefit directional codes offer against the RTC. In \Cref{fig:N2E3N2_fixed_p}, we fix the physical error rate to be $p=10^{-3}$ and plot the logical error rates against $n_{ph}$, the total number of physical qubits used, for all five $N^2E^3N^2$-codes from \Cref{fig:N2E3N2_vs_RTC} as well as the RTC, TC and RPSC. We additionally display a line for each of the two families of $N^2E^3N^2$-codes and the RTC, which are fitted to the transformed data points $(\sqrt{n_{ph}},\; \log p_L)$, where $\sqrt{n_{ph}}$ is proportional to the distance of the code. Note that we do not show fitted lines on the TC's and RPSC's data to avoid overloading the plot. If $s_{dir}$ and $s_{rtc}$ denote the slopes of the lines corresponding to one family of $N^2E^3N^2$-codes and the RTC, respectively, then $(\frac{s_{rtc}}{s_{dir}})^2$ expresses the ratio of physical qubits one needs when using that family of $N^2E^3N^2$-codes to achieve the same logical error rate as the RTC at $p=10^{-3}$, see \Cref{sec:app-simulation_details} for more details on this. These comparison numbers are given in the fourth and fifth columns of the first row of \Cref{tab:slope_ratios}. In each cell the first number is $(\frac{s_{rtc}}{s_{dir}})^2$ expressed in percentage, and the sub- and superscripts describe an interval around it that provides an at least $98\%$ confidence for this ratio, see details in \Cref{sec:app-simulation_details}. Based on this, we can conclude that the rotated $N^2E^3N^2$-code family requires only a quarter to a third of the number of physical qubits to achieve the same logical error rate as the RTC at $p=10^{-3}$. The rectangular family requires approximately a third to half the number of physical qubits to achieve the same. The other rows of \Cref{tab:slope_ratios} contain the same comparison against the TC with even and odd distances, and the RPSC with even and odd distances, furthermore, the other columns describe the comparisons of the other directional code families.

\begin{table}[!htb]
    \centering
    \footnotesize
    \begin{tabular}{|Sc||Sc|Sc|Sc|Sc|Sc|Sc|}
        \hline
        \fbox{$p=10^{-3}$} & \makecell{$NE^3N$ \\ rectangular} & \makecell{$N^2E^2N^2$ \\ rotated} & \makecell{$N^2E^2N^2$ \\ rectangular} & \makecell{$N^2E^3N^2$ \\ rotated} & \makecell{$N^2E^3N^2$ \\ rectangular} & \makecell{$N^2E^4N^2$ \\ rotated} \\ [0.5ex]
        \hline\hline
        \makecell{RTC \\ (even $d$)} & \makecell{$73.4\%_{-16.1\%}^{+13.7\%}$} & \makecell{$29.6\%_{-10.8\%}^{+3.4\%}$} & \makecell{$55.7\%_{-9.1\%}^{+4.3\%}$} & \makecell{$29.2\%_{-4.4\%}^{+2.1\%}$} & \makecell{$49.9\%_{-18.3\%}^{+5.7\%}$} & \makecell{$32.8\%_{-2.9\%}^{+1.9\%}$} \\
        \hline
        \makecell{TC \\ even $d$} & \makecell{$63.0\%_{-14.6\%}^{+13.7\%}$} & \makecell{$25.1\%_{-9.3\%}^{+3.9\%}$} & \makecell{$48.0\%_{-8.6\%}^{+4.8\%}$} & \makecell{$25.1\%_{-4.1\%}^{+2.5\%}$} & \makecell{$43.0\%_{-16.3\%}^{+6.0\%}$} & \makecell{$28.2\%_{-2.9\%}^{+2.4\%}$} \\
        \hline
        \makecell{TC \\ odd $d$} & \makecell{$55.9\%_{-12.3\%}^{+11.5\%}$} & \makecell{$22.7\%_{-8.4\%}^{+2.8\%}$} & \makecell{$42.8\%_{-7.3\%}^{+3.6\%}$} & \makecell{$22.4\%_{-3.5\%}^{+1.9\%}$} & \makecell{$38.4\%_{-14.3\%}^{+4.7\%}$} & \makecell{$25.0\%_{-2.2\%}^{+1.9\%}$} \\
        \hline
        \makecell{RPSC \\ even $d$} & \makecell{$30.5\%_{-6.7\%}^{+5.9\%}$} & \makecell{$12.3\%_{-4.5\%}^{+1.5\%}$} & \makecell{$23.3\%_{-3.9\%}^{+1.8\%}$} & \makecell{$12.2\%_{-1.8\%}^{+0.9\%}$} & \makecell{$20.8\%_{-7.6\%}^{+2.5\%}$} & \makecell{$13.7\%_{-1.2\%}^{+0.8\%}$} \\
        \hline
        \makecell{RPSC \\ odd $d$} & \makecell{$25.1\%_{-5.5\%}^{+4.7\%}$} & \makecell{$10.1\%_{-3.7\%}^{+1.2\%}$} & \makecell{$19.2\%_{-3.1\%}^{+1.4\%}$} & \makecell{$10.0\%_{-1.5\%}^{+0.7\%}$} & \makecell{$17.1\%_{-6.3\%}^{+1.9\%}$} & \makecell{$11.2\%_{-1.0\%}^{+0.6\%}$} \\
        \hline
    \end{tabular}
    \caption{Comparison of directional codes against the RTC, TC and RPSC. In each cell we display a number that, based on the sampled logical error rates, expresses what percentage of the physical qubits one needs with directional codes to achieve the same logical error rate as the RTC/TC/RPSC. The sub- and super-scripts describe a confidence interval for this quantity with at least $98\%$.}
    \label{tab:slope_ratios}
\end{table}

Lastly, in \Cref{fig:N2E3N2_vs_BB}, we further compare the rotated, rectangular and filler $N^2E^3N^2$-codes against the $n=72,144$ BB codes. It is clear from the plot that the BB codes outperform the directional codes, although at the expense of introducing additional non-local connections. Since on physical hardware this modification introduces additional noise and also poses significant engineering challenges, directional codes may be a potentially more desirable choice than the BB codes for hardware implementation on a toric device. Additionally, it is important to note that our simulations assume the simplistic SI-$1000$ noise model that does not increase the system's noise with the addition of connections beyond square-grid. See e.g. \cite{tesseractdecoder} for a comparison of BB codes and the RPSC with such a noise model. Overall, this makes directional codes a potentially competitive alternative to the BB codes.


\section{Discussion}
\label{sec:discussion}
In this paper, we presented directional codes, a novel family of qLDPC codes tailored specifically to nearest-neighbour connectivity hardware. We simulated directional codes that encode four, six, twelve and eighteen logical qubits, and found that they outperform the RTC under the standard SI-$1000$ circuit-level Pauli noise model. In particular, we found that the best directional code family achieves the same logical error rate as the RTC at physical error rate $p=10^{-3}$ while requiring only a quarter to a third of the number of physical qubits. As such, directional codes may be a competitive alternative to proposals that necessitate adding connections beyond square-grid, like the BB codes, as they impose significant engineering challenges and make the device noisier.

The key idea is to utilise iSWAP gates instead of CZ gates to implement the syndrome extraction circuits of directional codes. Since iSWAP gates essentially dynamically alter the connectivity, utilising them opens up a new way to discover better QEC codes than the ones which use the more popular CZ gate. In particular, we believe it would be interesting to explore what other QEC codes can be implemented under the same connectivity that the BB codes from \cite{Bravyi2024} require, but utilising iSWAP gates instead of CZ gates. This we leave as an interesting open question.

In this paper we only investigated the quantum memory performance of directional codes; a natural next step would be to explore their ability to support fault-tolerant quantum computation. We see two potential avenues to address this. Since directional codes fit on a toric device, the first avenue is to develop a modular quantum computer similar to IBM's ``Tour de Gross'' architecture \cite{IBM_roadmap} but with substantially fewer qubit connections within each module. More precisely, a directional code could be stored on each toric ``code module'' and the planar ``factory module'' could operate with planar surface codes. A second possible avenue involves first adding boundaries to directional codes so that they may be implemented on a planar hardware with nearest-neighbour connectivity. We note that while we have been working on the revision of this manuscript, three papers addressing this problem have been released \cite{Georgia-boundary,Boren-boundary,Bunny-boundary}. This avenue would then involve the development of lattice surgery for planar directional codes, which respects the square-grid connectivity constraint. We recently became aware that the authors of \cite{tile-codes-ls} intend to release a revised version of their paper soon in which they address this. Finally, it would need to be figured out how to perform logical non-Clifford gates on directional code logical qubits, which may involve teleporting high-quality magic states prepared on RPSCs, or figuring out how to prepare them directly on planar directional codes. These ingredients would together make it possible to perform logical circuits via the Pauli-based computational model \cite{BravyiPBC}. We believe the latter option is more promising, since it requires the scaling up of a more generic-purpose planar square-grid architecture with iSWAP gates rather than the construction of a purpose-built hardware tailored for a specific QEC code.

For decoding directional codes, we used Tesseract which is relatively slow and unlikely to be suitable for real-time decoding. However, recent improvements in BP-based decoders suggest that, for example, the BP-RELAY decoder \cite{bp-relay} could be a natural possible way forward for real-time decoding of directional codes. We also mention a more recent result \cite{Kaavya-paper} in which the authors adapted the minimum-weight perfect matching algorithm to decode various qLDPC codes, including the directional codes. (We note that the directional codes benchmarked in \cite{Kaavya-paper} are from an earlier version of this manuscript.) This could be another way to achieve real-time decoding for directional codes. Alternatively, there may be further scope for development of a bespoke decoder for directional codes, which we leave as an interesting open problem.


\section*{Author contributions}
G.G. conceived the project idea and lead the project and paper writing. All authors contributed to the scientific work afterwards, as well as to software work. D.B. handled all the simulations and numerical data collection. A.R. implemented distance calculators that notably aided the discovery of the rotated and rectangular directional code families. All authors contributed to writing the paper.
\section*{Acknowledgements}
We would like to thank Dan Browne and Ophelia Crawford for reading an early version of the paper and providing valuable feedback. We also would like to thank Earl T. Campbell for some very inspiring discussions at the early stages of the project, and for creating a stimulating scientific environment at Riverlane. We are also grateful to Ophelia Crawford who leads the Logic team at Riverlane, and who was very supportive of the project from the start and gave us space to work on it. We are also extremely grateful to Chidi Nnadi, Gordon Bateman and Mohamed Rashid Hassan, who provided us with a high performance computing cluster with hundreds of cores, without which this research would have been impossible. Finally, we thank Maria Maragkou and Luigi Martiradonna for giving feedback on the abstract, and the Introduction and Discussion sections.

\section*{Corresponding authors}
Correspondence and requests for materials should be addressed to all authors: \linebreak G.G. (george.geher@riverlane.com or gehergyuri@gmail.com), \linebreak D.B. (david.s.byfield@gmail.com) and \linebreak A.R. (archibald.ruban@riverlane.com or archibald.ruban@gmail.com).
\section*{Competing interests}
UK patent application 2503472.9 (naming G.G., D.B. and A.R. as co-inventors) contains technical aspects from this paper.

\section*{Data availability}
The circuits used for simulations and the numerical data obtained are available at \url{https://doi.org/10.5281/zenodo.21242387}.

\clearpage
\bibliography{references}

\newpage

\appendix

\renewcommand{\thefigure}{S\arabic{figure}}\setcounter{figure}{0}
\renewcommand{\thetable}{S\arabic{table}}\setcounter{table}{0}
\renewcommand{\theequation}{S\arabic{equation}}\setcounter{equation}{0}


\section{Additional details on code construction}
\label{sec:app-method}
\begin{figure}[!htb]
    \centering
    \begin{subfigure}{0.45\textwidth}
        \centering
        \includegraphics[width=0.45\linewidth]{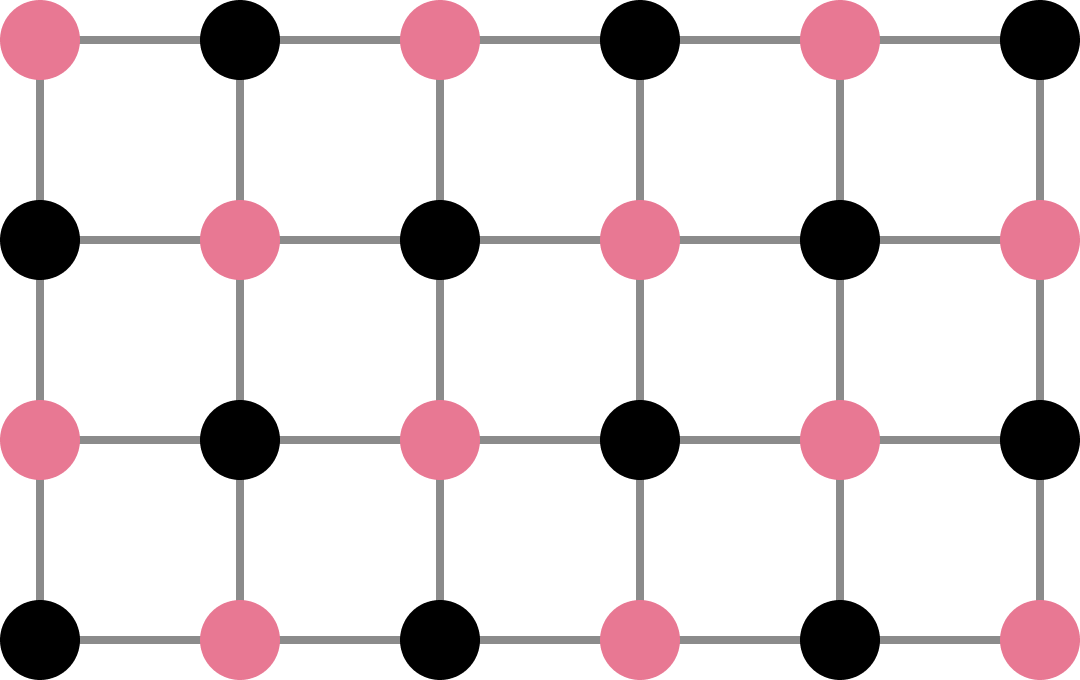}
        \caption{Square-grid connectivity}
        \label{fig:square_grid}
    \end{subfigure}
    \begin{subfigure}{0.45\textwidth}
        \centering
        \includegraphics[width=0.45\linewidth]{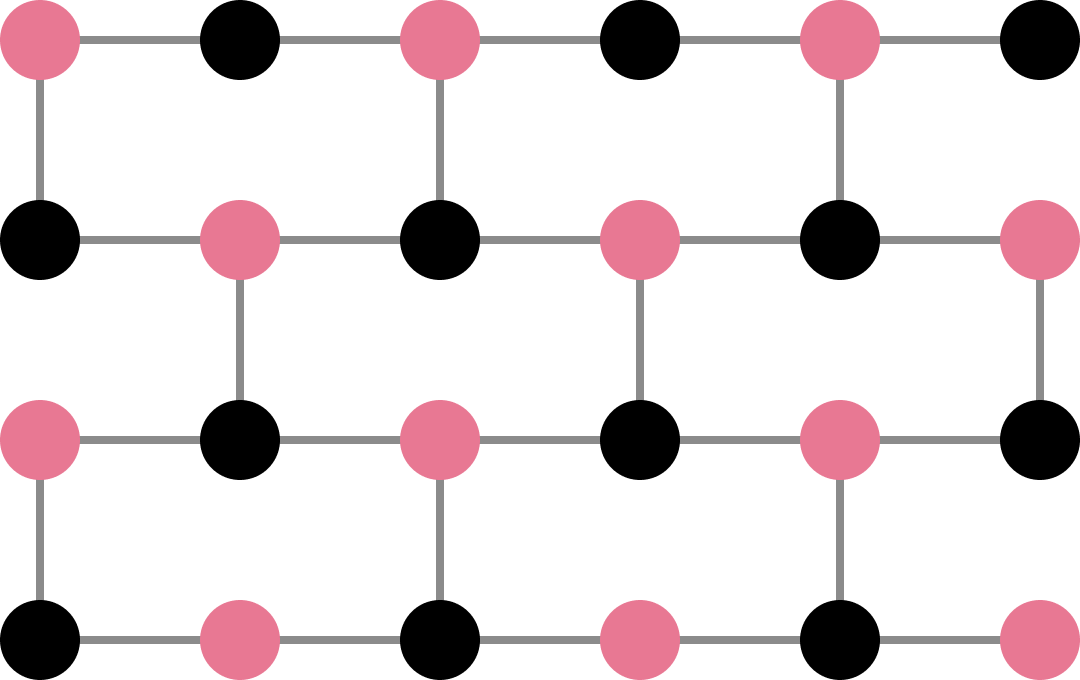}
        \caption{Hexagonal-grid connectivity}
        \label{fig:hex_grid}
    \end{subfigure}
    \caption{Hardware layouts with square- and hexagonal-grid connectivity. All the codes constructed in this paper can be executed on the square-grid (a), and some of them even on the hexagonal-grid (b) (see also \Cref{tab:directions_and_layouts}). The qubits are coloured with red and black, corresponding to ancilla and data qubits.}
    \label{fig:square-and-hex-grids}
\end{figure}

Directional codes are constructed by defining their associated syndrome extraction circuits, which are designed to satisfy the square-grid connectivity constraint (\Cref{fig:square_grid}) on a regular or twisted torus. We will also see that some of these codes do not require the full connectivity of the square-grid architecture, allowing them to be executed on the sparser hexagonal-grid. We achieve this by using iSWAP as the native two-qubit gate, instead of the popular choices like CZ.

After showing a motivating example in \Cref{sec:app-motivating_example}, we explain our code construction method in three steps, detailed in \Cref{sec:app-one_stab_explanation,sec:app-schedule_conflicts,sec:app-parallelograms}. First, we consider the infinite planar square-grid lattice $\Z^2$ with square-grid connectivity, and describe how we can measure a set of $X$-stabilisers whose data qubits are not all connected to a single (ancilla) qubit. Second, we consider what happens if we naively apply this circuit-construction method to arbitrary stabiliser types. In order to obtain a valid circuit that measures these stabilisers simultaneously and independently, we need to avoid entangling the ancilla qubits, see e.g. \cite{TangledSchedules}. We prove a theorem that characterises when this is the case. Finally, in order to get a CSS code that encodes logical qubits, we need to make the lattice finite. We achieve this by wrapping $\Z^2$ around a (possibly twisted) torus. We also collect all the balanced CSS code types (i.e. the ones with an equal number of $X$- and $Z$-stabilisers) that we can obtain in such a way with stabiliser weights $4,5,6$ and $7$, although the construction is valid for arbitrarily high weight. The weight-$4$ case, as shown in \Cref{sec:app-motivating_example}, provides a linearly transformed version of the toric code. However, the higher weight versions provide novel CSS codes that can outperform the RTC, while still being executable under low connectivity constraints. It is important to note that our construction specifies the syndrome extraction circuits for directional codes implicitly.
\begin{figure*}[!htb]
    \centering
    \includegraphics[width=0.95\linewidth]{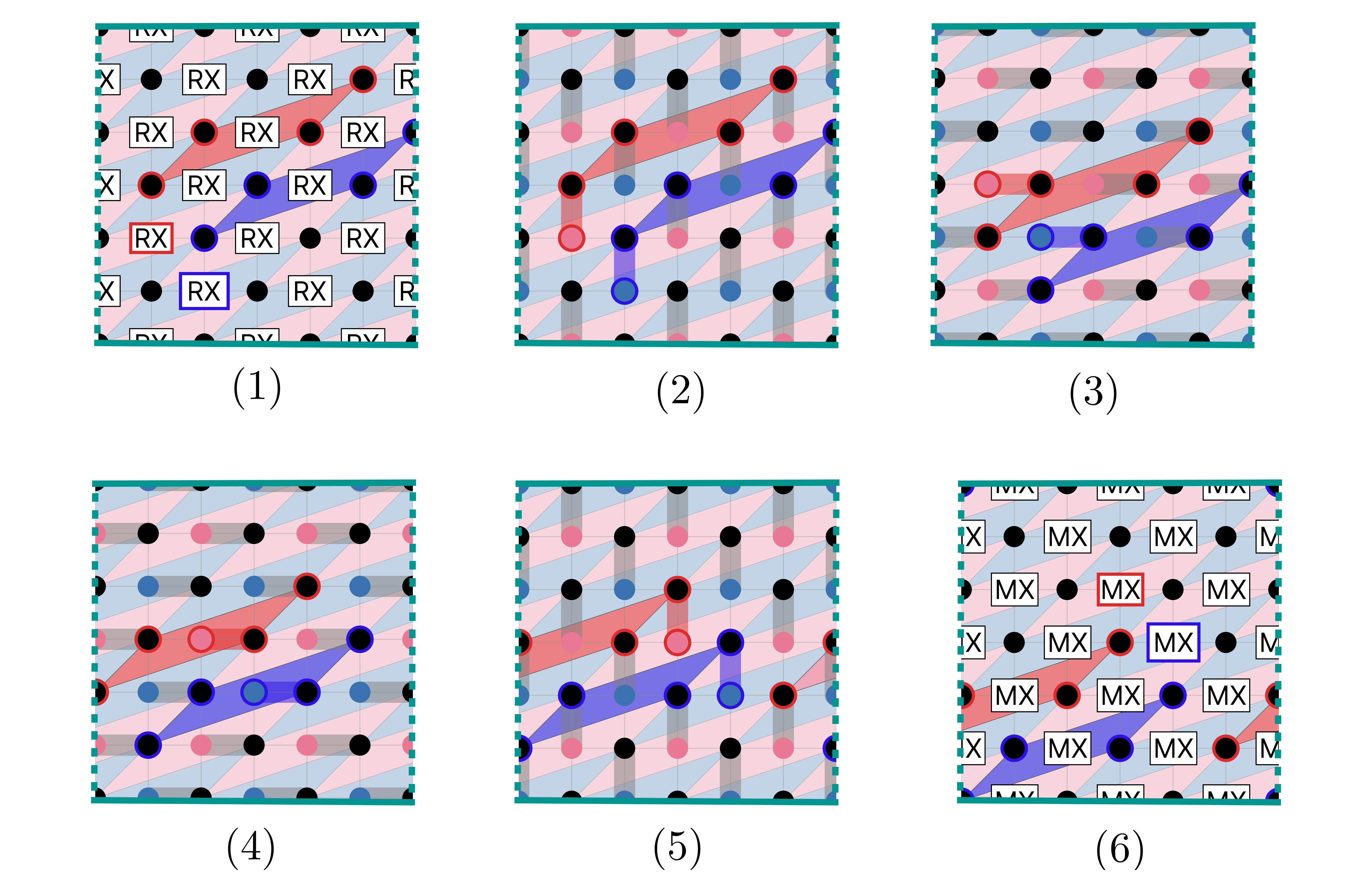}
    \caption{A sequence of layers of ancilla qubit resets (panel (1)), CPSWAP gates (panels (2)--(5)) and ancilla qubit measurements (panel (6)) that measure the stabilisers of a distance-$3$ TC on a $6\times6$ torus. Note that the TC is linearly transformed. Each red (blue) parallelogram represents a weight-$4$ $X$- ($Z$)-stabiliser supported on the four data qubits at its vertices. Data qubits are shown as black dots, and the ancilla qubit corresponding to each $X$- ($Z$)-stabiliser is shown as a red (blue) dot. Each grey line in panels (2)--(5) is a CPSWAP gate between the two qubits it touches, with the ancilla (data) qubit always being the control (target) qubit. The CPSWAP is a CXSWAP (CZSWAP) if the ancilla qubit corresponds to a $X$- ($Z$)-stabiliser. An $X$- ($Z$)-stabiliser has been highlighted in red (blue), together with its corresponding ancilla and data qubits, and the reset (RX), CPSWAP and measurement (MX) gates that participate in measuring it out. These serve to illustrate the movement of the stabilisers and data qubits relative to the ancilla qubits, as caused by the action of the SWAP-part of each CPSWAP layer.}
    \label{fig:neen_syndrome_extraction}
\end{figure*}


\subsection{Motivating example: the toric code as the \texorpdfstring{$NE^2N$}{NE2N}-code}
\label{sec:app-motivating_example}
As was mentioned in the main text, for all integer $d\geq 2$, the $NE^2N$-code on the regular torus specified by $\vec{v}_1 = (2d,0)$, $\vec{v}_2 = (0,2d)$ gives the (unrotated or Kitaev) TC of distance $d$. We explain this now for general $d$ using the illustration depicted in \Cref{fig:neen_syndrome_extraction} for the $d=3$ case. More precisely, consider the $(2d)\times(2d)$ square with toric boundary conditions, on which we place qubits in a square-grid fashion, with the bottom-left-most qubit being a data qubit. As shown in \Cref{fig:neen_syndrome_extraction}, we apply a sequence of four CPSWAP layers (see also \Cref{fig:directional_layers}), along with the necessary resets before and measurements after on the ancilla qubits, to perform syndrome extraction. The first layer, which is an $\vec{n}$ layer, has the effect of moving all the data qubits one unit south and the ancilla qubits one unit north. The next two layers, which are $\vec{e}$ layers, each move all the data qubits west and all the ancilla qubits east. We then perform a final $\vec{n}$ layer, again moving the data qubits south and the ancilla qubits north. The stabilisers measured this way all have parallelogram shapes. More precisely, expressed in the initial configuration (i.e. corresponding to panel (1) or (2) of \Cref{fig:neen_syndrome_extraction}), the stabiliser measured by the ancilla qubit initially at position $A_0$ is supported on the qubits that are initially at the positions $A_0+\vec{n}$, $A_0+2\vec{n}+\vec{e}$, $A_0+2\vec{n}+3\vec{e}$ and $A_0+3\vec{n}+4\vec{e}$. For instance, if we consider $A_0$ to be the highlighted red (blue) ancilla qubit in panel (2) of \Cref{fig:neen_syndrome_extraction}, then the four black qubits that are encircled red (blue) and are at the vertices of the red (blue) highlighted parallelogram are the qubits on which the stabiliser measured by $A_0$ is supported.

Now, consider the vector $\vec{w}_2 = 2\vec{v}_1+\vec{v}_2 = (4d,2d)$, and define $\mathcal{K}(\vec{v}_1,\vec{v}_2)=\{\ell_1\vec{v}_1+\ell_2\vec{v}_2\colon \ell_1, \ell_2\in\Z\}$ and $\mathcal{K}(\vec{v}_1,\vec{w}_2)=\{\ell_1\vec{v}_1+\ell_2\vec{w}_2\colon \ell_1, \ell_2\in\Z\}$. It is straightforward to see that $\mathcal{K}(\vec{v}_1,\vec{v}_2) = \mathcal{K}(\vec{v}_1,\vec{w}_2)$, and as such wrapping around the square $\mathcal{P}(\vec{v}_1, \vec{v}_2)$ (like in \Cref{fig:neen_syndrome_extraction}) gives the same torus and $NE^2N$-code as wrapping around the parallelogram $\mathcal{P}(\vec{v}_1, \vec{w}_2)$ (see \Cref{sec:app-parallelograms} for more details on parallelogram equivalence). The latter is depicted in \Cref{fig:NEEN_lin_transform}, and we see that it now resembles the standard TC, depicted in \Cref{fig:TC}. More precisely, consider the linear transformation $T$ defined by $T((1,0)) = (1,0)$ and $T((0,1))=(-2,1)$. Clearly, $T$ transforms the parallelogram $\mathcal{P}(\vec{v}_1, \vec{w}_2)$ into the $(2d)\times(2d)$ square, and additionally the $NE^2N$-stabilisers into diamond shaped squares. This shows that the $NE^2N$-code defined on the square $\mathcal{P}(\vec{v}_1, \vec{v}_2)$ is indeed a linearly transformed version of the standard TC.

\begin{figure}[!htb]
    \centering
    \begin{subfigure}{0.45\textwidth}
        \centering
        \includegraphics[width=1.25\linewidth]{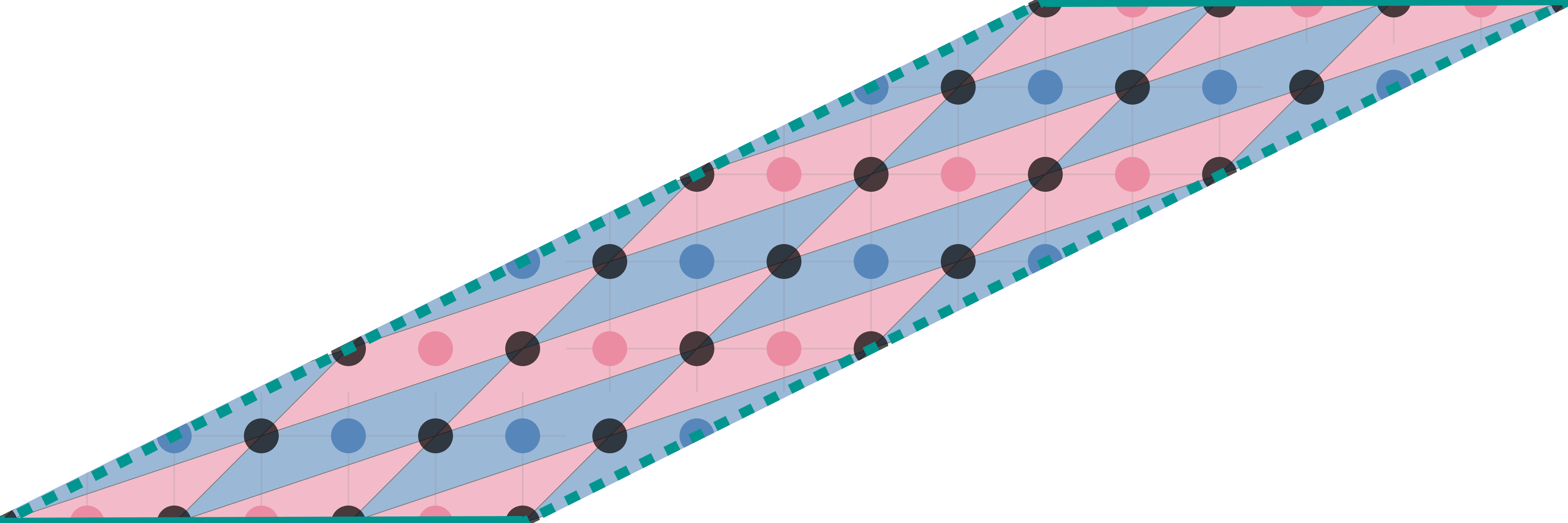}
        \caption{$NE^2N$-code on a parallelogram.}
        \label{fig:NEEN_lin_transform}
    \end{subfigure}
    \begin{subfigure}{0.45\textwidth}
        \centering
        \includegraphics[width=0.45\linewidth]{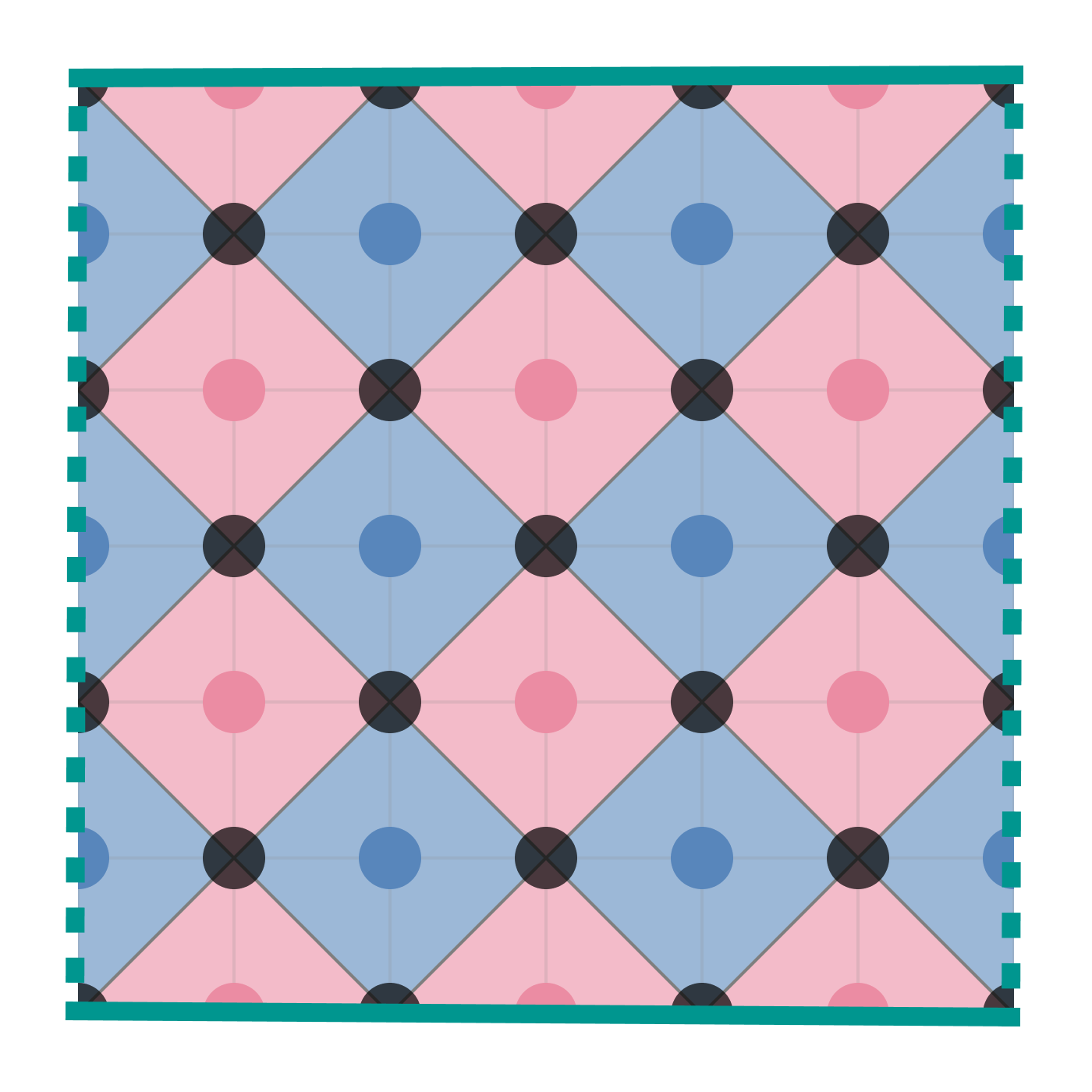}
        \caption{Standard representation of the TC.}
        \label{fig:TC}
    \end{subfigure}
    \caption{Two variants of the TC of distance $3$. (a) The $NE^2N$-code wrapped around the parallelogram $\mathcal{P}(\vec{v}_1=(6,0), \vec{w}_2=(12,6))$. This is equivalent to the $NE^2N$-code wrapped around the square $\mathcal{P}(\vec{v}_1=(6,0), \vec{w}_2=(0,6))$ (\Cref{fig:neen_syndrome_extraction}), and can be implemented using CPSWAP gates on a toric square-grid device. (b) The standard representation of the TC that can be implemented with controlled-Pauli gates under toric square-grid connectivity. Applying the linear transformation $T$, defined in the text, to the qubits and the parallelogram-shaped stabilisers in (a) brings us to exactly this configuration of qubits and diamond-square shaped stabilisers.}
    \label{fig:NEEN_to_TC}
\end{figure}

In the next few subsections we will formalise this process and describe how we may generalise to higher-weight stabiliser measurements, allowing novel CSS codes to be constructed on a torus. 

Lastly, we point out that the above explanation highlights that in the case of $NE^2N$-codes adding boundaries in the parallelogram shaped box $\mathcal{P}(\vec{v}_1, \vec{w}_2)$ is the most natural choice. Indeed, this would give a linearly transformed version of the unrotated planar surface code with distance $d$. However, if boundaries are added in the square shaped box $\mathcal{P}(\vec{v}_1, \vec{v}_2)$, then the code distance would reduce. We also point out that in order to obtain a linearly transformed version of the RPSC with $NE^2N$, one needs to add boundaries along another parallelogram shaped box. However, as the uniform scheduling reduces the circuit-level distance for the RPSC, that construction would not be optimal.

\subsection{Measuring \texorpdfstring{$X$}{X}-stabilisers on an infinite planar square-grid}
\label{sec:app-one_stab_explanation}

The goal of this subsection is to construct circuits on $\Z^2$ that measure an infinite set of $X$-stabilisers using CXSWAP gates. Recall that the following circuit measures an $X$-stabiliser on data qubits $\{Q_1,Q_2,\dots, Q_\ell\}$ by using an ancilla qubit $A$ and CX gates (see e.g. \cite{Nielsen_Chuang_2010,Mac-morphing-2} or \cite[Section 2]{TangledSchedules}): initialise qubit $A$ in the $|+\rangle$ state, apply the gates $\CX(A,Q_1), \dots, \CX(A,Q_\ell)$ in this order, and finally measure qubit $A$ in the $X$ basis. This measurement outcome on qubit $A$ corresponds to the $\prod_{j=1}^\ell X_{Q_j}$ stabiliser measurement. Denote the circuit defined this way by $\mathcal{M}(X; A; Q_1,Q_2,\dots, Q_\ell)$, which is often called the standard, or ``bare-ancilla'', syndrome extraction circuit to measure the stabiliser $\prod_{j=1}^\ell X_{Q_j}$. Note that the ordering of the data qubits in this circuit fixes the schedule of the stabiliser, i.e. the order in which the entangling gates are applied.

Assume we have a collection of such syndrome extraction circuits that each measure an $X$-stabiliser on some data qubits. Suppose further that these circuits use ancilla qubits that are different from each other and also from all data qubits. It is well-known that executing these circuits simultaneously gives a syndrome extraction circuit that measures all the $X$-stabilisers simultaneously and independently, provided that no qubit appears in the same layer twice. A crucial assumption here is that all the stabilisers are of the same type, $X$ (we will consider the more general case in the next subsection).

\begin{figure*}[t]
    \centering
    \includegraphics[height=7cm]{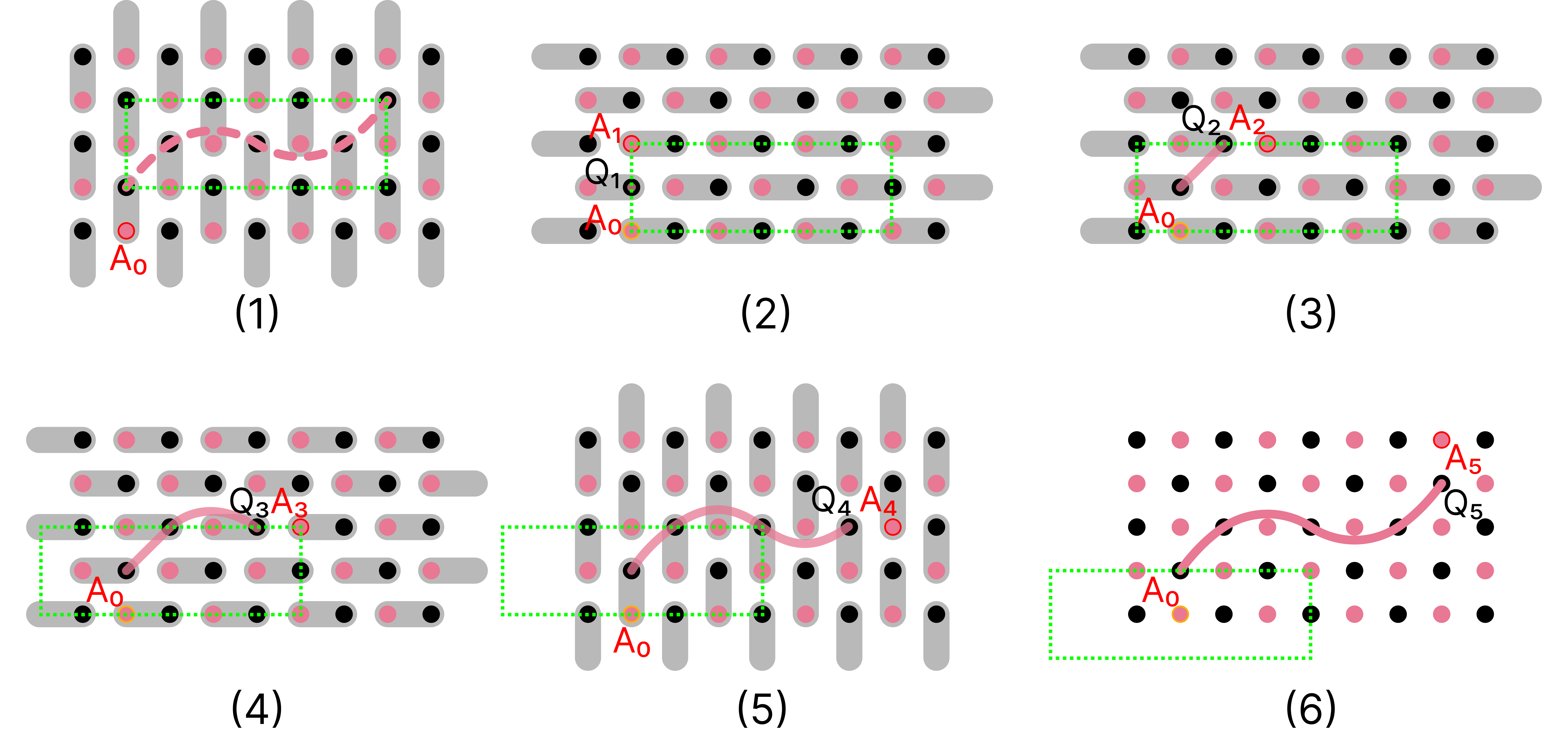}
    \caption{Illustration of CXSWAP layers executed according to the sequence of directions $\vec{n},\vec{e},\vec{e},\vec{e},\vec{n}$, shown in the data qubits' frame of reference. These layers appear in the syndrome extraction circuit of $NE^3N$-stabilisers, all of $X$-type here, although note that varying stabiliser types is also possible with $NE^3N$-stabilisers, see \Cref{fig:NEEEN_Delta_Vectors,fig:example_parallelogram_lattice} and \Cref{tab:directions_and_layouts}. (1)-(6) show how the ancilla qubit of a given stabiliser moves ($A_0, A_1,\dots, A_5$) under the action of each layer to reach all its data qubits ($Q_1,\dots, Q_5$). Data qubits are shown as black nodes, and ancilla qubits as red nodes. Grey highlighted pairs of qubits are involved in a CXSWAP together in that layer, as in \Cref{fig:cxswap_qubit_pair}. The CXSWAP's control is always the ancilla and target always the data qubit. The red dashed snake-like shape shown in (1) is an $NE^3N$-stabiliser whose Pauli terms are collected by the ancilla qubit circled in red. The shape fills up as each of its data qubits are covered by a CXSWAP gate with the ancilla qubit as control. Since each panel is shown in the data qubits' frame of reference, the stabiliser appears static throughout, even though on the hardware it is constantly sliding. This is illustrated by the green dashed box, which represents a fixed section of hardware across which the qubits move. As shown, the stabiliser slides out of that region, as we apply the CXSWAP layers.}
    \label{fig:neeen_syndrome_extraction}
\end{figure*}

Next, consider the infinite square-grid lattice $\Z^2$ on which we lay out data and ancilla qubits in a chequerboard pattern as described in \Cref{sec:code-construct}. Consider an arbitrary sequence of directions: $\vec{d_1},\vec{d_2},\dots, \vec{d_\ell}$, where $\vec{d_j} \in \{\vec{n},\vec{e},\vec{s},\vec{w}\}$. We describe how the ancilla qubits move in the frame of reference of the data qubit sublattice if we were to execute the series of CXSWAP layers according to $\vec{d_1},\vec{d_2},\dots, \vec{d_\ell}$ (\Cref{fig:directional_layers_with_arrows}). Recall that in each layer the two sublattices move in opposite directions on the hardware, as explained in \Cref{sec:code-construct} and illustrated in \Cref{fig:directional_layers_with_arrows}. Therefore, in the frame of reference of the data qubit sublattice (in which data qubits appear as fixed qubits) the ancilla qubits always move twice the distance. Thus in the data qubits' frame of reference, if the ancilla qubit started at position $A_0$ and after the first $j$ layers got shifted to $A_j$, then the vector pointing from the former to the latter position is 
\begin{equation}\label{eq:A_0A_j}
    \overrightarrow{A_0A_j} = 2\sum_{p=1}^j \vec{d_p}.
\end{equation}
Therefore, given that the data qubits just moved in the opposite direction, the vector that points from the original position of the ancilla qubit $A_0$ to the current position of the data qubit $Q_j$ it just interacted with is 
\begin{equation}\label{eq:Q_j}
    \overrightarrow{A_0Q_j} = \overrightarrow{A_0A_j} - \vec{d_j} = 2\sum_{p=1}^{j-1} \vec{d_p} + \vec{d_j},
\end{equation}
see \Cref{fig:neeen_syndrome_extraction}. 

Now, consider the circuit where we start by initialising all ancilla qubits $A\in\Z^2_{\mathrm{anc}}$ in the $|+\rangle$ state, then execute a series of CXSWAP gates according to the sequence of directions $\vec{d_1},\vec{d_2},\dots, \vec{d_\ell}$, and finally measure all ancilla qubits (at their final location) in the $X$ basis. Then, based on the above observations, we measured an infinite set of $X$-stabilisers simultaneously and independently, which all have the same shape and scheduling. More precisely, for each ancilla qubit that in the beginning is placed at $A_0$ in the frame of reference of the data qubit sublattice, the stabiliser which was measured by it is $\prod_{j=1}^\ell X_{Q_j}$, and throughout the circuit the ancilla qubit itself moved to position $A_\ell$ where it was measured, see \Cref{eq:A_0A_j,eq:Q_j} that define $Q_j$ and $A_\ell$. We call stabilisers obtained in such a way $D_1D_2\dots D_\ell$-stabilisers, where $D_j$ stands for the capitalised version of $\vec{d_j}$ without the arrow. 

If we wish to measure these stabilisers for a second time, we could continue by executing a circuit according to the reversed direction sequence $(-\vec{d_\ell}),(-\vec{d_{\ell-1}}),\dots, (-\vec{d_1})$. For instance, in the case of $NE^2N$-codes, the reversed layers would correspond to $\vec{s}, \vec{w}, \vec{w}, \vec{s}$. If we wish to measure them more times, we could apply these two circuits repeatedly, one after the other. In this way the stabilisers on the hardware move back and forth. Note that in principle we could use the same direction sequence all the time, but then the stabilisers would move in the same direction on the hardware, and as a result would eventually leave any finite area of $\Z^2$. Also, as was pointed out in e.g. \cite{Mac-morphing-2}, it is usually beneficial to reverse the scheduling of every other round.


\subsection{Measuring stabilisers of varying Pauli types on an infinite planar square-grid}
\label{sec:app-schedule_conflicts}

The goal of this subsection is to construct circuits on $\Z^2$ that measure an infinite set of $X$- and $Z$-stabilisers simultaneously and independently using CPSWAP gates. Recall that in the previous subsection it was a crucial assumption that all stabilisers had the same Pauli type. When we combine circuits that separately measure stabilisers of different Pauli types, we additionally have to be careful not to entangle the ancilla qubits.

First, consider circuits that use controlled-Pauli gates. Denote by \linebreak $\mathcal{M}(Z; A; Q_1,Q_2,\dots, Q_\ell)$ the circuit we obtain from $\mathcal{M}(X; A; Q_1,Q_2,\dots, Q_\ell)$ (defined at the beginning of \Cref{sec:app-one_stab_explanation}) by replacing all CX gates by CZ gates. It is well-known that $\mathcal{M}(Z; A; Q_1,Q_2,\dots, Q_\ell)$ measures the stabiliser $\prod_{j=1}^\ell Z_{Q_j}$. Consider now the two circuits $\mathcal{M}(X; A; Q_1,Q_2,\dots, Q_\ell)$ and $\mathcal{M}(Z; A'; Q'_1,Q'_2,\dots, Q'_\ell)$, where $A$ and $A'$ are different and $A, A' \notin \{Q_1,\dots,Q_\ell\}\cup\{Q'_1,\dots,Q'_\ell\}$. Recall that when we combine these circuits so that they happen simultaneously, we get a circuit that measures $\prod_{j=1}^\ell X_{Q_j}$ and $\prod_{j=1}^\ell Z_{Q'_j}$ simultaneously and independently if and only if the following three conditions are satisfied:
\begin{itemize}
    \item[(a)] $Q_j\neq Q'_j$ for all $j=1,\dots,\ell$;
    \item[(b)] the set $\{(i,j)\colon Q_i = Q'_j, i<j\}$ has evenly many elements;
    \item[(c)] the set $\{(i,j)\colon Q_i = Q'_j, i>j\}$ has evenly many elements;
\end{itemize}
see e.g. \cite[Section 2]{TangledSchedules}. Condition (b)/(c) means that among the data qubits shared between the two stabilisers, evenly many are scheduled earlier/later in the circuit \linebreak $\mathcal{M}(X; A; Q_1,Q_2,\dots, Q_\ell)$ than in $\mathcal{M}(Z; A'; Q'_1,Q'_2,\dots, Q'_\ell)$. Furthermore, conditions (a)--(c) ensure that when we combine the circuits the ancilla qubits do not get entangled, and that the stabilisers $\prod_{j=1}^\ell X_{Q_j}$ and $\prod_{j=1}^\ell Z_{Q'_j}$ commute. We also point out that in the case where the Pauli types of the two stabilisers are the same, then only condition (a) needs to be satisfied, as was the case in \Cref{sec:app-one_stab_explanation}.

\begin{figure}[!htb]
    \centering
    \begin{subfigure}{0.30\textwidth}
        \centering
        \includegraphics[width=0.6\linewidth]{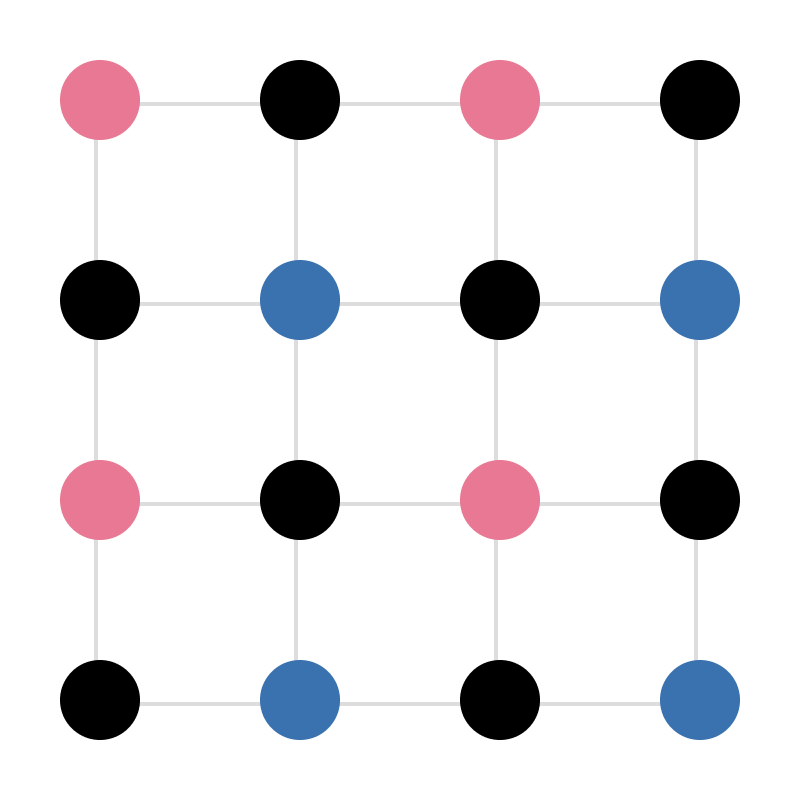}
        \caption{Layout 1}
        \label{fig:layout1}
    \end{subfigure}
    \begin{subfigure}{0.30\textwidth}
        \centering
        \includegraphics[width=0.6\linewidth]{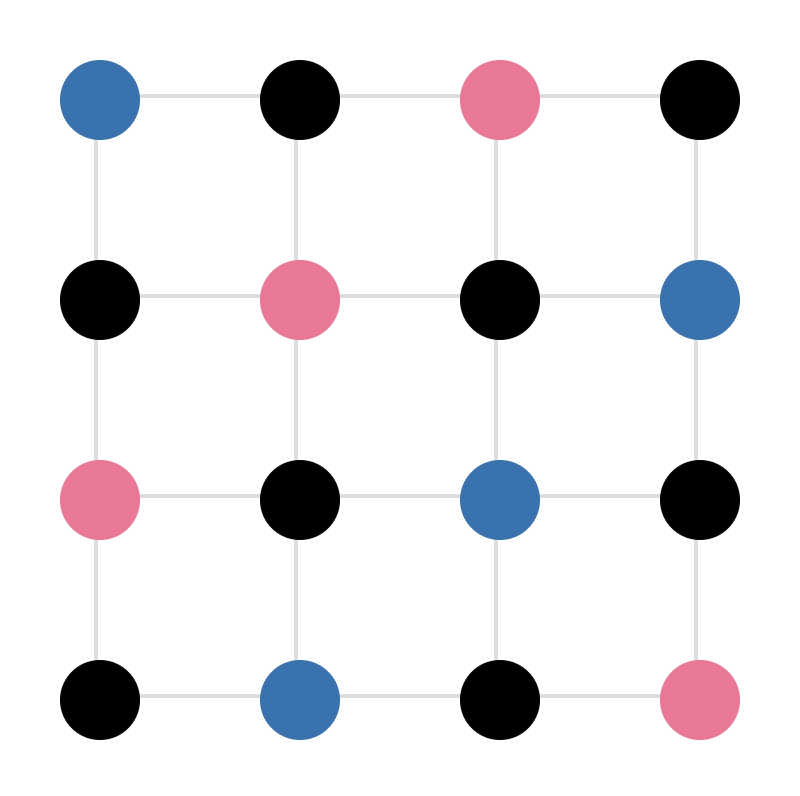}
        \caption{Layout 2}
        \label{fig:layout2}
    \end{subfigure}
    \begin{subfigure}{0.30\textwidth}
        \centering
        \includegraphics[width=0.6\linewidth]{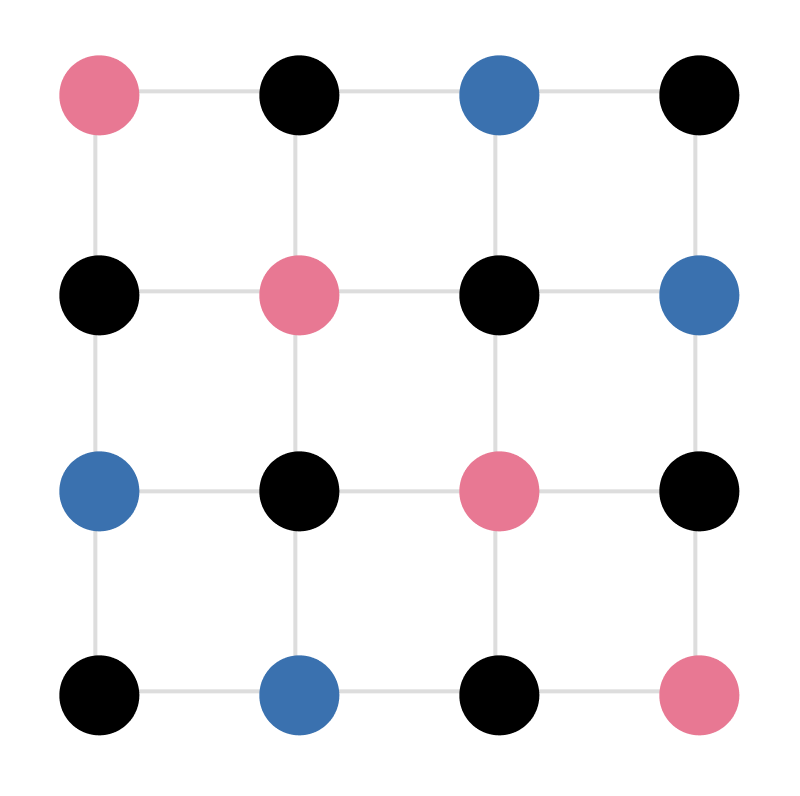}
        \caption{Layout 3}
        \label{fig:layout3}
    \end{subfigure}
    \caption{The three typical layouts $\Lambda$ for allocating stabilisers for a directional code, although other valid layouts do exist. Data/$X$-ancilla/$Z$-ancilla qubits are shown as black/red/blue dots.}
    \label{fig:Layouts}
\end{figure}

Next, we call a mapping $\Lambda\colon\Z^2_{\mathrm{anc}}\to\{X,Z\}$ a ``layout''. In the rest of the paper a layout will determine the location of the ancilla qubits for each stabiliser type. For instance, the layout considered in \Cref{sec:code-construct} where we have alternating rows of $X$- and $Z$-ancillas we call ``Layout 1'' (\Cref{fig:layout1}, see also \Cref{fig:N2E3N2_syndrome_extraction,fig:N2E3N2_on_parallelogram}), the layout where we have alternating diagonals travelling north-east we call ``Layout 2'' (\Cref{fig:layout2}), and the one where we have alternating diagonals travelling north-west we call ``Layout 3'' (\Cref{fig:layout3}). These three layouts were the most common during our investigation.

Now, we discuss circuits that use CPSWAP gates. Previously, we denoted by $A_j$ the position of the ancilla qubit of a stabiliser after the $j$th CPSWAP layer in the data qubits' frame of reference. We introduce the notation $A(j)$ for the position of the same qubit in the hardware's frame of reference. Consider a direction sequence $\vec{d_1},\vec{d_2},\dots, \vec{d_\ell}$, and a layout $\Lambda\colon\Z^2_{\mathrm{anc}}\to\{X,Z\}$. We define the circuit $\mathcal{C}(\Lambda; \vec{d_1},\vec{d_2},\dots, \vec{d_\ell})$ in the following way. 
\begin{itemize}
    \item[(0)] Reset all ancilla qubits $A\in\Z^2_{\mathrm{anc}}$ in the $|+\rangle$ state.
    \item[($j$)] ($1\leq j\leq \ell$) For each initial ancilla qubit position $A$, track where it is now due to the previous steps: $A(j-1)=A+\sum_{p=1}^{j-1}\vec{d_p}$. At each of these apply the $\CXSWAP(A(j-1),A(j-1)+\vec{d_j})$ gate if $\Lambda(A)=X$, otherwise the $\CZSWAP(A(j-1),A(j-1)+\vec{d_j})$ gate. Due to this layer, the ancilla qubit is now at position $A(j)=A(j-1)+\vec{d_j}=A+\sum_{p=1}^{j}\vec{d_p}$.
    \item[($\ell+1$)] Measure all ancilla qubits at their final locations $A(\ell)$ in the $X$ basis.
\end{itemize}
Note that with this notation the circuit from the previous subsection, that measures only $X$-stabilisers, corresponds to $\mathcal{C}(\Lambda_X; \vec{d_1},\vec{d_2},\dots, \vec{d_\ell})$ where $\Lambda_X$ is the mapping that has value $X$ everywhere. We saw that this circuit measures simultaneously and independently $D_1D_2\dots D_\ell$-stabilisers of $X$-type associated with all ancilla qubits. As another example, the circuit illustrated in \Cref{fig:N2E3N2_syndrome_extraction} is $\mathcal{C}(\Lambda; \vec{n},\vec{n},\vec{e},\vec{e},\vec{e},\vec{n},\vec{n})$ where $\Lambda$ is Layout 1. We noted in \Cref{sec:code-construct} and will prove in \Cref{cor:NaEbNa} that this circuit measures simultaneously and independently both $X$- and $Z$-stabilisers.

\begin{figure*}[!htb]
     \centering
     \begin{subfigure}{0.40\textwidth}
         \centering
         \includegraphics[width=\linewidth]{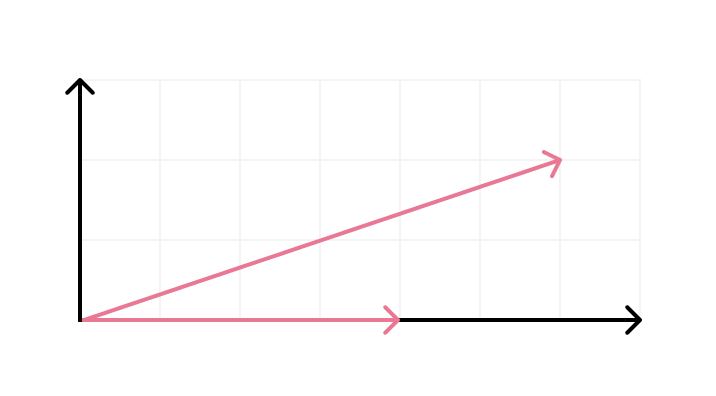}
         \caption{$\Delta_{odd}(\vec{n},\vec{e},\vec{e},\vec{e},\vec{n})$}
         \label{fig:NEEEN_DeltaVectors1}
     \end{subfigure}
     \begin{subfigure}{0.40\textwidth}
         \centering
         \includegraphics[width=\linewidth]{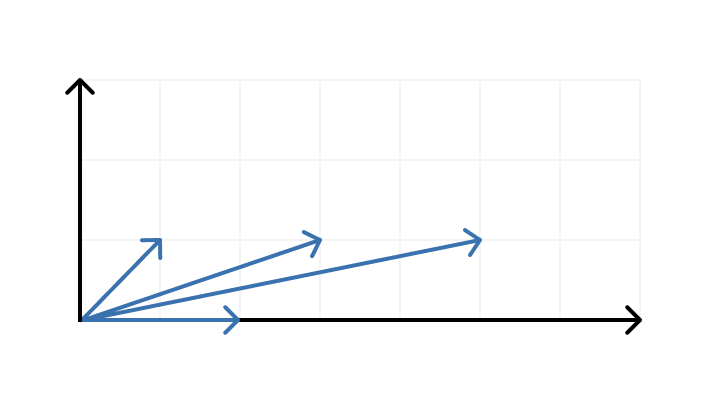}
         \caption{$\Delta(\vec{n},\vec{e},\vec{e},\vec{e},\vec{n})\setminus\Delta_{odd}(\vec{n},\vec{e},\vec{e},\vec{e},\vec{n})$}
         \label{fig:NEEEN_DeltaVectors2}
     \end{subfigure}
     \caption{Delta vectors of an $NE^3N$-stabiliser. (a) The vectors from $\Delta_{odd}(\vec{n},\vec{e},\vec{e},\vec{e},\vec{n})$ each shown in red. (b) The vectors from $\Delta(\vec{n},\vec{e},\vec{e},\vec{e},\vec{n})\setminus\Delta_{odd}(\vec{n},\vec{e},\vec{e},\vec{e},\vec{n})$ each shown in blue.}
     \label{fig:NEEEN_Delta_Vectors}
\end{figure*}

Next, our goal is to establish a set of ``if and only if'' conditions that the circuit $\mathcal{C}(\Lambda; \vec{d_1},\vec{d_2},\dots, \vec{d_\ell})$ needs to satisfy so that it measures the $D_1D_2\dots D_\ell$-stabiliser of $\Lambda(A)$-type, for all $A\in\Z^2_{\mathrm{anc}}$, simultaneously and independently. Since each CPSWAP layer moves the data qubit sublattice uniformly in one direction (see \Cref{fig:directional_layers_with_arrows}), we can use conditions (a)--(c) above in the data qubits' frame of reference for the CPSWAP case to validate circuits of the form $\mathcal{C}(\Lambda; \vec{d_1},\vec{d_2},\dots, \vec{d_\ell})$. Let us introduce some notation. We define the set of \emph{delta vectors} associated with the sequence of directions $\vec{d_1},\vec{d_2},\dots, \vec{d_\ell}$ as the set of vectors that point from an earlier scheduled data qubit $Q_i$ to a later scheduled one $Q_j$ (i.e. $i<j$). More precisely, using \Cref{eq:Q_j} this set can be written as
\begin{equation}\label{eq:delta_vector_definition}
    \Delta = \Delta(\vec{d_1},\vec{d_2},\dots, \vec{d_\ell}) := \left\{\overrightarrow{Q_iQ_j} = \vec{d_i} + 2\sum_{p=i+1}^{j-1} \vec{d_p} + \vec{d_j} \; \colon 1\leq i < j \leq \ell \right\}.
\end{equation}
Throughout the paper we shall always assume that $\vec{0}\notin\Delta$, or equivalently, that each pair of ancilla and data qubits interact at most once in the circuit $\mathcal{C}(\Lambda; \vec{d_1},\vec{d_2},\dots, \vec{d_\ell})$. The following subset of delta vectors will be useful in the proof of \Cref{thm:Thm1} to capture conditions (b)--(c):
\begin{equation}
    \Delta_{odd} = \Delta_{odd}(\vec{d_1},\vec{d_2},\dots, \vec{d_\ell}) := \left\{\vec{v} \in \Delta \; \colon \; \vec{v} = \overrightarrow{Q_iQ_j} \text{ for odd many } (i,j), 1\leq i < j \leq \ell \right\},
\end{equation}
In other words, $\Delta_{odd}$ is the set of delta vectors that appear odd many times. This set and $\Delta\setminus\Delta_{odd}$ are illustrated for $NE^3N$-stabilisers in \Cref{fig:NEEEN_Delta_Vectors}.

For a given set of vectors $S\subset\Z^2$, denote by $\SPAN_\Z (S)$ the set of all vectors that are linear combinations of vectors from $S$ with integer coefficients. Now, we are able to state and prove our theorem.

\begin{theorem}
\label{thm:Thm1}
Consider a layout $\Lambda\colon\Z^2_{\mathrm{anc}}\to\{X,Z\}$ and a sequence of directions $\vec{d_1},\vec{d_2},\dots, \vec{d_\ell}$, where $\vec{d_j} \in \{\vec{n},\vec{e},\vec{s},\vec{w}\}$, $\ell\in\Z$ and $\ell>0$. Assume $\vec{0} \notin \Delta$. Define the sub-lattice $\mathcal{L} \coloneq \SPAN_\Z (\Delta_{odd})$.
Then the circuit $\mathcal{C}(\Lambda;\vec{d_1},\vec{d_2},\dots, \vec{d_\ell})$ measures all the $D_1D_2\dots D_\ell$-stabilisers simultaneously and independently according to the layout $\Lambda$ if and only if 
\begin{equation}\label{eq:sublatt_colouring_constraint}
    \Lambda(A) = \Lambda(A') \text{ holds for all } A, A' \in \Z^2_{\mathrm{anc}}, \overrightarrow{AA'\;\;}\in\mathcal{L}.
\end{equation}
\end{theorem}

Before we prove our theorem, we explain its meaning through some special cases. According to the theorem, the only restriction on the layout $\Lambda$ is forced via the vectors of $\Delta_{odd}$, more precisely via the sub-lattice $\mathcal{L}$, but otherwise we are free to choose the type of the ancilla qubits. Consider the case of the $NE^2N$-stabilisers. Since we have $\Delta(\vec{n},\vec{e},\vec{e},\vec{n})_{odd} = \{2\vec{e},4\vec{e}+2\vec{n}\}$, we have $\mathcal{L} = \{2x\vec{n} + 2y\vec{e}\colon x,y\in\Z\}$. Therefore, a valid layout $\Lambda$ for the $NE^2N$ case needs to be constant on both sublattices $\mathcal{L}$ and $\mathcal{L}+\vec{e}+\vec{n}$. In particular, we may only get a quantum CSS code (after wrapping around a torus, see \Cref{sec:app-parallelograms}) if $\Lambda$ is ``Layout 1'', up to translation, which is what we had in \Cref{sec:app-motivating_example}. We note that the $NE^2N$-stabilisers require the full connectivity of the square-grid architecture (\Cref{fig:square_grid}).

As a further example, which this time gives a novel CSS code with weight-$5$ stabilisers, consider the $NE^3N$-stabilisers. We already established that $\Delta(\vec{n},\vec{e},\vec{e},\vec{e},\vec{n})_{odd} = \{4\vec{e},6\vec{e}+2\vec{n}\}$ (\Cref{fig:NEEEN_Delta_Vectors}, see also \Cref{fig:NE3N_stabilisers}), and thus $\mathcal{L} = \{2x(\vec{e}+\vec{n}) + 2y(\vec{e}-\vec{n})\colon x,y\in\Z\}$. Therefore, we have four sublattices for this case on which we are free to choose the value of the layout $\Lambda$. In particular, we may choose any of the layouts from \Cref{fig:Layouts}. We point out that remarkably the syndrome extraction circuit of $NE^3N$-stabilisers uses only the hexagonal-grid connectivity (\Cref{fig:hex_grid}). As another weight-$5$ example, consider the $NE^2SW$-stabilisers, like in \Cref{fig:NE2SW_stabilisers}. This is a particularly interesting case, as both $\vec{e}+\vec{n}$ and $-(\vec{e}+\vec{n})$ are in $\Delta_{odd}$. As such if we have differing ancilla types at $A$ and $A+\vec{e}+\vec{n}$, then even though the corresponding stabilisers commute, conditions (b)--(c) are not met. Therefore the same stabiliser type is enforced by \Cref{thm:Thm1} for these two. Also, since $\vec{e}-\vec{n}\in\Delta_{odd}$, we have $\mathcal{L} = \Z^2_{\mathrm{anc}}$, therefore the $NE^2SW$ case cannot give rise to a quantum CSS code.

\begin{figure*}[!htb]
    \centering
    \begin{subfigure}{0.48\textwidth}
        \centering
        \includegraphics[width=\linewidth]{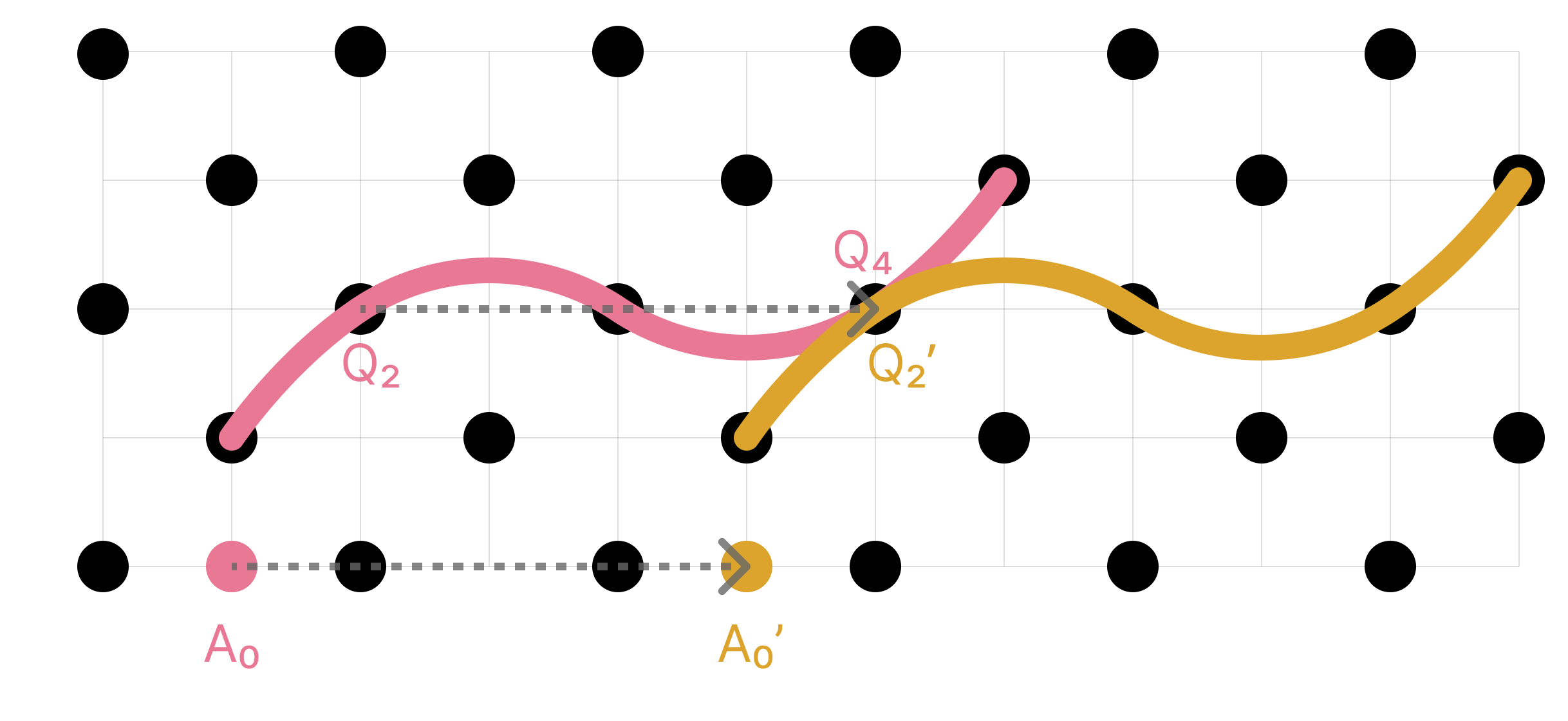}
        \caption{Same stabiliser type enforced.}
    \end{subfigure}
    \begin{subfigure}{0.48\textwidth}
        \centering
        \includegraphics[width=\linewidth]{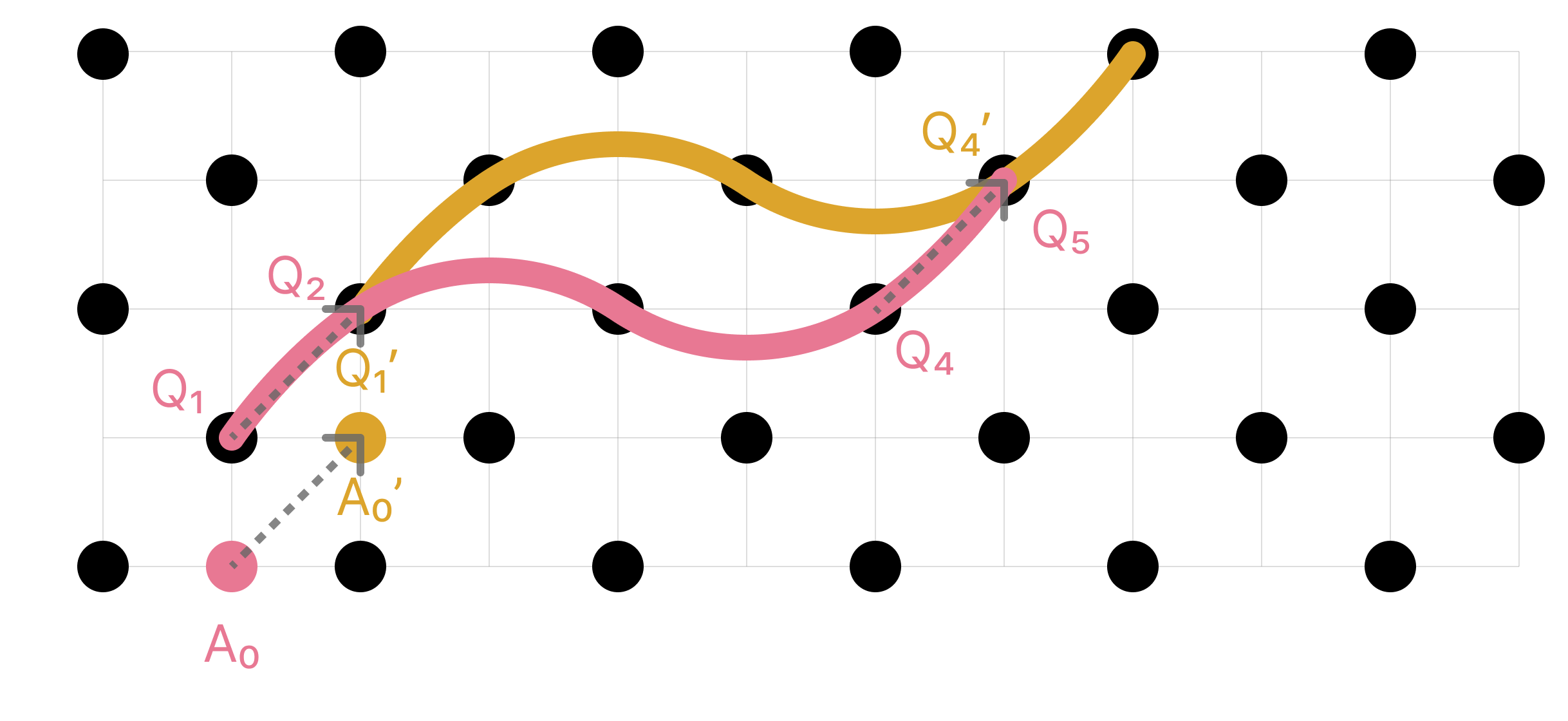}
        \caption{Differing stabiliser types allowed.}
    \end{subfigure}
    \caption{Two $NE^3N$-stabilisers in red and orange that are separated by the vector $\overrightarrow{A_0A'_0\;}$. Note that $\overrightarrow{A_0A'_0\;}=\overrightarrow{Q_iQ'_i\;}$ for all $i$. Black dots are data qubits, the red and orange dots are the two ancilla qubits associated with the two shown stabilisers. The rest of the stabilisers and ancilla qubits are not shown. In (a) $\overrightarrow{A_0A'_0\;} = 4\vec{e} \in\Delta_{odd}$, hence \Cref{thm:Thm1} forces them to be of the same Pauli type. In (b) $\pm\overrightarrow{A_0A'_0\;} = \pm(\vec{n}+\vec{e}) \notin\Delta_{odd}$, hence they can be of differing Pauli type, say the red one $X$ and the orange one $Z$.}
    \label{fig:NE3N_stabilisers}
\end{figure*}

\begin{figure*}[!htb]
    \centering
    \begin{subfigure}{0.35\textwidth}
        \centering
        \includegraphics[width=\linewidth]{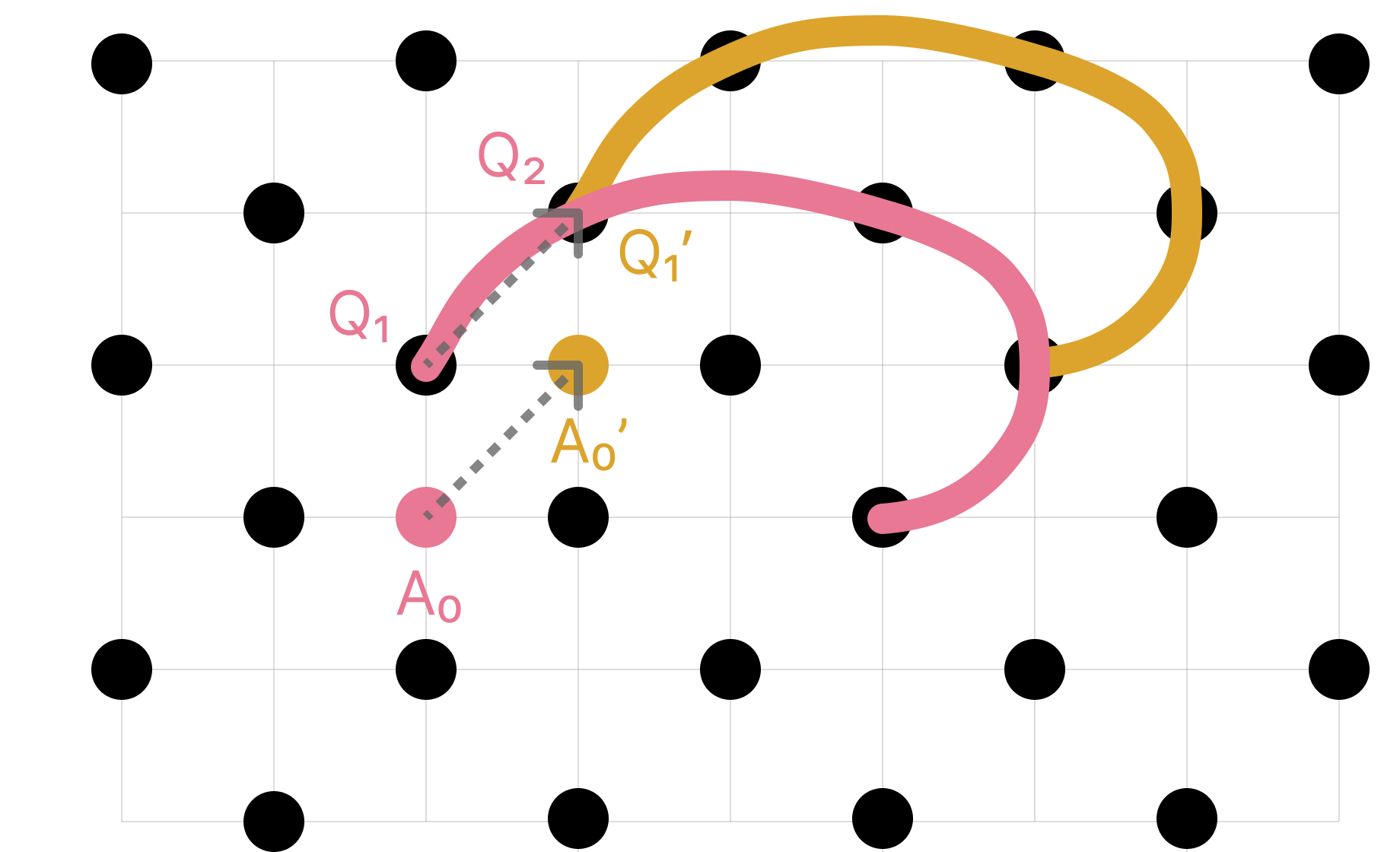}
        \caption{Same stabiliser type enforced.}
    \end{subfigure}
    \hspace{1cm}
    \begin{subfigure}{0.35\textwidth}
        \centering
        \includegraphics[width=\linewidth]{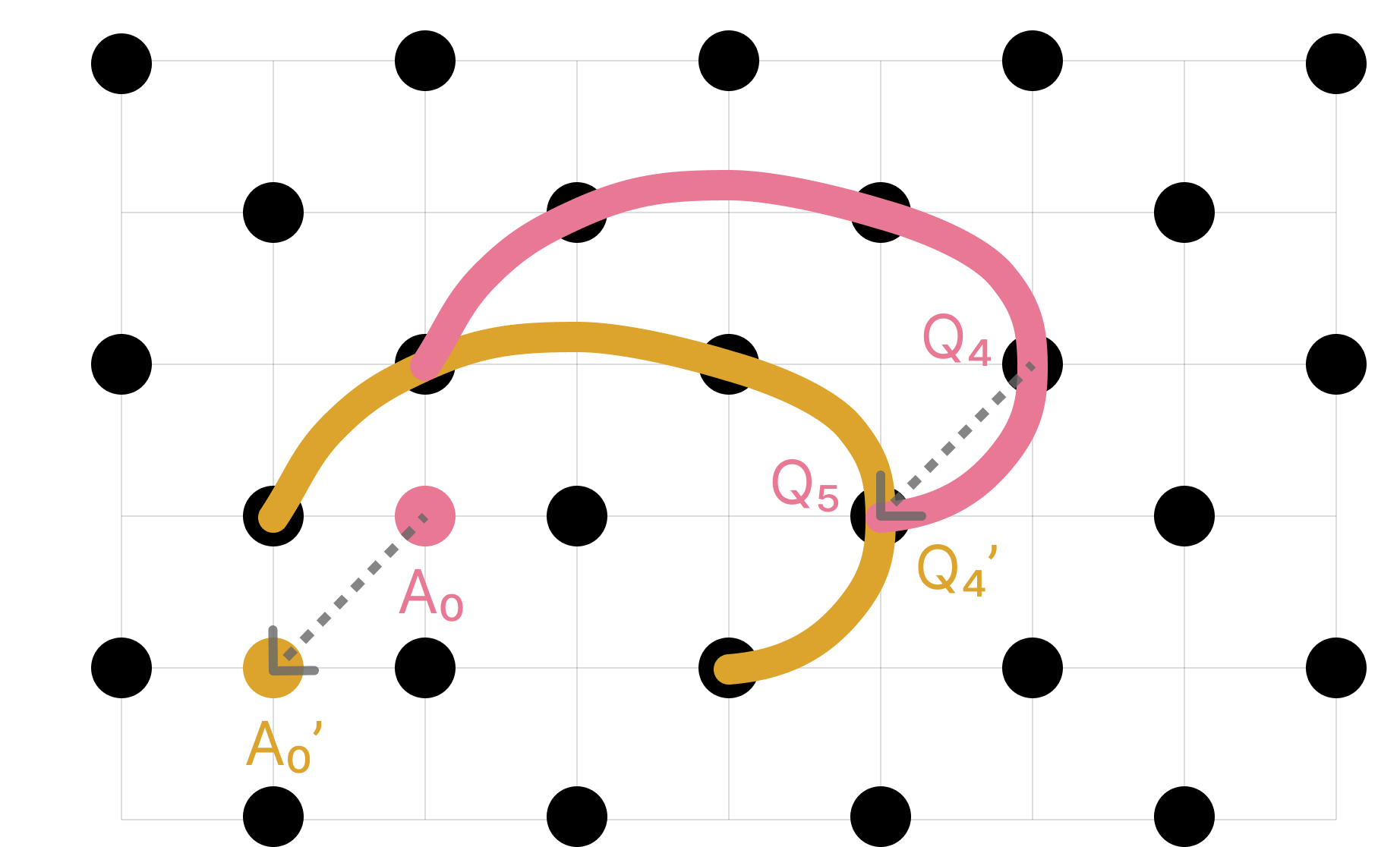}
        \caption{Same stabiliser type enforced.}
    \end{subfigure}
    \caption{Two $NE^2SW$-stabilisers in red and orange that are separated by the vector $\overrightarrow{A_0A'_0\;}$. In both (a) and (b) $\vec{e}+\vec{n},-(\vec{e}+\vec{n})\in\Delta_{odd}$. In (a) $\overrightarrow{A_0A'_0\;} = \vec{e}+\vec{n}$, and in (b) $\overrightarrow{A_0A'_0\;} = -(\vec{e}+\vec{n})$. Therefore, by \Cref{thm:Thm1}, the two stabilisers are forced to be of the same type in both cases. Even though the stabilisers would commute with differing Pauli types, the circuit $\mathcal{C}(\Lambda; \vec{n},\vec{e},\vec{e},\vec{s},\vec{w})$ would entangle the ancilla qubits $A_0$, $A'_0$.}
    \label{fig:NE2SW_stabilisers}
\end{figure*}

\begin{proof}[Proof of \Cref{thm:Thm1}]
We will consider everything in the data qubit sublattice's frame of reference. Note then that the CPSWAP-version of conditions (a)--(c), that were stated in the beginning of this subsection, remain the same, as the data qubits always move uniformly in one direction. So our goal from now on is to check against (a)--(c). Consider two ancilla qubits $A_0,A'_0\in\Z^2_{\mathrm{anc}}$, $A_0 \neq A'_0$, and the stabilisers they are associated with: $\prod_{j=1}^w \Lambda(A_0)_{Q_j}$ and $\prod_{j=1}^w \Lambda(A'_0)_{Q'_j}$. Note that
\begin{equation}\label{eq:shift_qubits}
    \overrightarrow{Q_iQ'_j\;} = \overrightarrow{Q_iQ_j} + \overrightarrow{Q_jQ'_j\;} = \overrightarrow{A_0A'_0\;} + \overrightarrow{Q_iQ_j} \quad \text{for all} \; i,j,
\end{equation}
see \Cref{eq:Q_j} and \Cref{fig:NE3N_stabilisers,fig:NE2SW_stabilisers}. Substituting $j=i$ in \Cref{eq:shift_qubits} gives 
\begin{equation}\label{eq:primed_shift_same}
    \overrightarrow{A_0A'_0\;}=\overrightarrow{Q_iQ'_i\;} \quad \text{for all} \; i,
\end{equation}
and hence $Q_i\neq Q'_i$. Therefore we conclude that condition (a) always holds, regardless of the layout $\Lambda$. 

Note that we only need to check conditions (b)--(c) when $\Lambda(A_0)\neq\Lambda(A'_0)$. From \Cref{eq:shift_qubits} we obtain
\begin{equation}
    \{(i,j)\colon Q_i = Q'_j, i<j\} = \{(i,j)\colon \overrightarrow{Q_iQ_j} = -\overrightarrow{A_0A'_0\;}, i<j\}.
\end{equation}
Therefore, (b) is satisfied if and only if $-\overrightarrow{A_0A'_0\;}\notin\Delta_{odd}$. By symmetry, we conclude that (c) holds if and only if $\overrightarrow{A_0A'_0\;}\notin\Delta_{odd}$.

From the above observations we see that the layout $\Lambda$ gives a valid circuit if and only if 
\begin{equation}
    \Lambda(A) = \Lambda(A') \text{ holds for all } A, A' \in \Z^2_{\mathrm{anc}}, \overrightarrow{AA'\;\;} \text{ or } -\overrightarrow{AA'\;\;} \in\Delta_{odd}.
\end{equation}
This means that shifting the stabiliser by a vector or its opposite from $\Delta_{odd}$ forces the same Pauli type on the two stabilisers. Clearly, if we shift again by another such vector, that third stabiliser is still forced to have the same Pauli type. Therefore, by induction it is straightforward to see that we have a valid circuit if and only if \Cref{eq:sublatt_colouring_constraint} holds, which concludes the proof.
\end{proof}

We prove the following consequence of \Cref{thm:Thm1}, where we note that in the case when $\alpha=\beta=1$, the data qubits of the obtained $NEN$-stabilisers lie on diagonal lines.

\begin{corollary}\label{cor:NaEbNa}
    For all $\alpha,\beta\geq 1$ integers, Layout 1 is a valid layout for $N^\alpha E^\beta N^\alpha$-stabilisers.
\end{corollary}

\begin{proof}
    Let $\Lambda$ be Layout 1. Note that there are two types of vectors of the form $\overrightarrow{AA'\;\;}$ with $A, A' \in \Z^2_{\mathrm{anc}}$:
    \begin{itemize}
        \item when $\Lambda(A) \neq \Lambda(A')$, then both coordinates are odd, i.e. $\overrightarrow{AA'\;\;} = (2x+1)\vec{e}+(2y+1)\vec{n}$ with some $x,y\in\Z$;
        \item when $\Lambda(A) = \Lambda(A')$, then both coordinates are even, i.e. $\overrightarrow{AA'\;\;} = 2x\vec{e}+2y\vec{n}$ with some $x,y\in\Z$.
    \end{itemize}
    We prove now that any vector from $\Delta_{odd}$ is of the second form. From \Cref{eq:delta_vector_definition} we can see that any vector from $\Delta$ is of the form
    \begin{equation}
        \overrightarrow{Q_iQ_j} = \vec{d_i} + 2\sum_{p=i+1}^{j-1} \vec{d_p} + \vec{d_j}
    \end{equation}
    Since $\vec{d_i}, \vec{d_j}\in\{\vec{n},\vec{e}\}$, the vector $\overrightarrow{Q_iQ_j}$ has odd coordinates if and only if $\vec{d_i} \neq \vec{d_j}$. Furthermore, based on the specific structure of the direction sequence, it is straightforward that this happens if and only if either $i \leq \alpha < j \leq \alpha+\beta$ or $\alpha < i \leq \alpha+\beta < j$. A simple calculation gives that, on the one hand, 
    \begin{center}
        if $i \leq \alpha < j \leq \alpha+\beta$, then $\alpha < (2\alpha+\beta+1-j) \leq \alpha+\beta < (2\alpha+\beta+1-i)$;
    \end{center} and, on the other hand, 
    \begin{center}
        if $\alpha < i \leq \alpha+\beta < j$ then $(2\alpha+\beta+1-j) \leq  \alpha < (2\alpha+\beta+1-i) \leq \alpha+\beta$.
    \end{center}
    But this also shows that whenever the delta vector $\overrightarrow{Q_iQ_j}$ has odd coordinates, then automatically we have $\{i,j\}\neq \{2\alpha+\beta+1-j,2\alpha+\beta+1-i\}$ and $\overrightarrow{Q_iQ_j}=\overrightarrow{Q_{2\alpha+\beta+1-j}Q_{2\alpha+\beta+1-i}}$, and as such $\overrightarrow{Q_iQ_j}\notin\Delta_{odd}$.

    From here, we see that indeed all elements of $\Delta_{odd}$ have even coordinates, and therefore, the same holds for all elements of the generated sub-lattice $\mathcal{L}$. This implies that ``Layout 1'' is indeed valid.
\end{proof}

\begin{table}[!htb]
\footnotesize
    \begin{tabular}[t]{|c|c|c|}        
    \hline
        \multicolumn{3}{|c|}{Weight 4} \\
        \hline
        Direction & Valid Layouts & Connectivity \\ [0.5ex]
        \hline
        $NE^2N$&1&square-grid\\ \hline

        \multicolumn{3}{|c|}{Weight 5} \\
        \hline
        Direction & Valid Layouts & Connectivity \\ [0.5ex]
        \hline
        $NE^3N$&1,2,3&hex-grid\\ \hline
        $NESEN$&1,2,3&hex-grid\\ \hline
        $N^2EN^2$&1,2,3&hex-grid\\ \hline

        \multicolumn{3}{|c|}{Weight 6} \\
        \hline
        Direction & Valid Layouts & Connectivity \\ [0.5ex]
        \hline
        $NE^4N$&1&square-grid\\ \hline
        $NEN^2EN$&1&square-grid \\ \hline   
        $NENWSW$&3&square-grid\\ \hline
        $NES^2EN$&1&square-grid \\ \hline  
        $N^2E^2N^2$&1&square-grid \\ \hline  
    \end{tabular}
    \hfill
    \begin{tabular}[t]{|c|c|c|}
    \hline
        \multicolumn{3}{|c|}{Weight 7} \\
        \hline
        Direction & Valid Layouts & Connectivity \\ [0.5ex]
        \hline
        $NE^5N$&1,2,3&hex-grid\\ \hline  
        $NE^2NE^2N$&1,2,3&square-grid\\ \hline  
        $NE^2SE^2N$&1,2,3&hex-grid\\ \hline  
        $NEN^3EN$&1,2,3&square-grid\\ \hline  
        $NENWNEN$&1,2,3&hex-grid\\ \hline  
        $NESESEN$&1,2,3&square-grid\\ \hline  
        $NES^3EN$&1,2,3&square-grid\\ \hline  
        $NES^2WNE$&1&square-grid\\ \hline  
        $NESWSEN$&1,2,3&square-grid\\ \hline  
        $NESW^2NE$&1&square-grid\\ \hline  
        $N^2E^3N^2$&1,2,3&square-grid\\ \hline  
        $N^2ENEN^2$&1,2,3&square-grid\\ \hline  
        $N^2ESEN^2$&1,2,3&square-grid\\ \hline  
        $N^3EN^3$&1,2,3&square-grid\\ \hline  
    \end{tabular}
    \caption{All direction sequences $\vec{d_1},\vec{d_2},\dots, \vec{d_\ell}$ with $4\leq \ell \leq 7$ satisfying \Cref{thm:Thm1}. Their valid layouts that give rise to balanced CSS codes (after wrapping around a torus) are also shown. Most of these require the full square-grid connectivity, but as indicated there are some that only need the sparser hexagonal-grid. We note that the $N^2E^4N^2$ case is not included here, as that has weight-$8$ stabilisers.}
    \label{tab:directions_and_layouts}
\end{table}

Throughout our research we focused on balanced CSS codes that are two-dimensional. For these, we only encountered the three layouts from \Cref{fig:Layouts}. In \Cref{tab:directions_and_layouts} we list all the direction sequences between length $4$ and $7$ that satisfy the conditions of \Cref{thm:Thm1} with at least one of these three layouts. We also included in \Cref{tab:directions_and_layouts} whether the circuit uses square- or hexagonal-grid connectivity. These direction sequences were found via an exhaustive search, where we assumed that the first direction is always $\vec{n}$, and the first direction after that which is different from it is $\vec{e}$. We did not include direction sequences of length $8$ in \Cref{tab:directions_and_layouts}, but we emphasise that the $N^2E^4N^2$ case was benchmarked in \Cref{sec:simulations} and \Cref{sec:app-additional_results}. If one increases the length of the direction sequence further, one will find further valid direction sequences. However, as we increase the length we have a price to pay: the depth of the syndrome extraction circuit is increased, and hence it will become more noisy and potentially introduce more so-called bad hook errors \cite{TopologicalQuantumMemory,PhysRevA.90.062320,TangledSchedules,geher_error-corrected_2023}. Furthermore, the circuits grow larger and thus obtaining numerical results with Tesseract becomes infeasible. As such, we decided to restrict our code search and simulations for directional codes that have stabilisers of weight at most $8$.


\subsection{Wrapping around parallelograms}
\label{sec:app-parallelograms}

\begin{figure}[!htb]
     \centering
     \begin{subfigure}{0.45\textwidth}
         \centering
         \includegraphics[width=\linewidth]{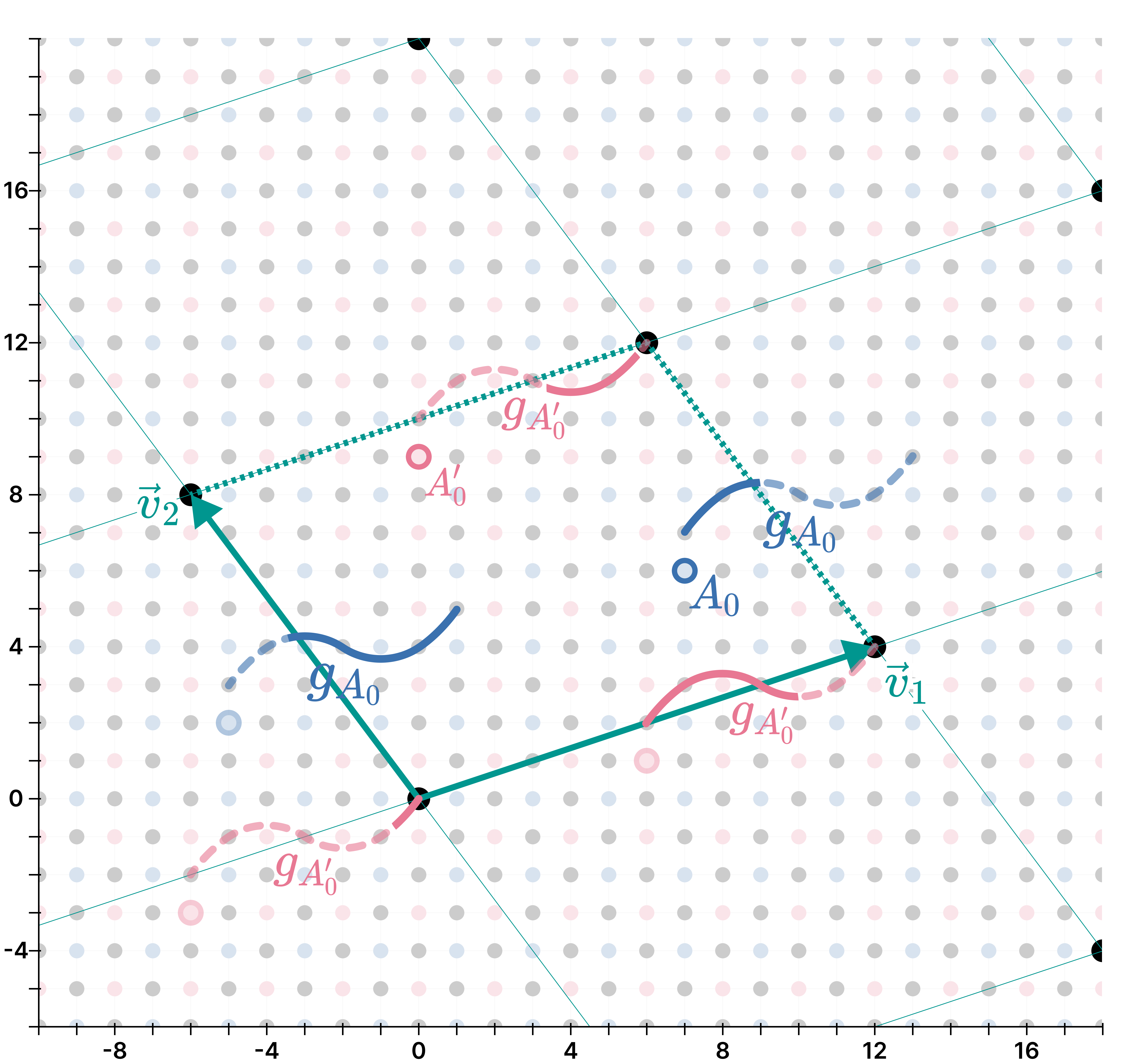}
         \caption{A parallelogram, its associated (twisted) torus, and an $X$- and $Z$-stabiliser.}
         \label{fig:parallelogram_lattice_sub}
     \end{subfigure}
     \begin{subfigure}{0.45\textwidth}
         \centering
         \includegraphics[width=\linewidth]{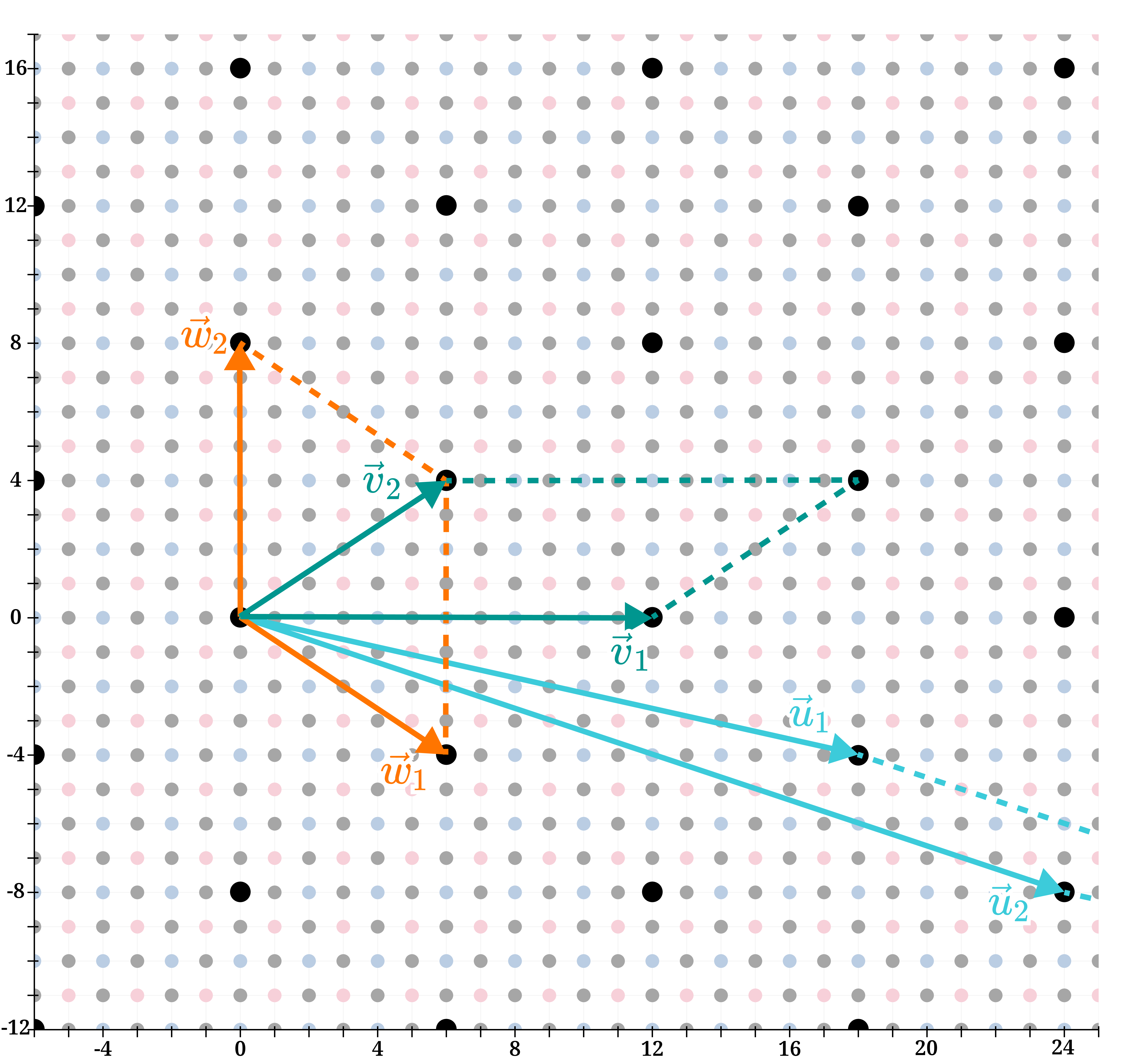}
         \caption{Three equivalent parallelograms that define the same (twisted) torus.}
         \label{fig:equivalent_parallelogram}
     \end{subfigure}
    \caption{
    We show examples of parallelograms on $\Z^2$, as well as examples of equivalent parallelograms that therefore give equivalent tori and code instances. (a) The parallelogram $\mathcal{P}(\vec{v}_1=(12, 4), \vec{v}_2 = (-6, 8))$ in green and its associated sublattice $\mathcal{K}$ as large black points. This parallelogram with Layout 1 defines an $NE^3N$-code with four logical qubits. Two examples of $NE^3N$-stabilisers defined across the torus boundary are shown. Following the directions ordering, $g_{A_0}$ in blue is the $Z$-stabiliser associated with the starting ancilla qubit $A_0=(7, 6)$, and is supported on data qubits $(7, 7)$, $(8, 8)$, $(-2, 4) \equiv (10, 8)\;(\modulo\mathcal{K})$, $(0, 4)$ and $(1, 5)$. Similarly, $g_{A_0'}$ in red is the $X$-stabiliser associated with the starting ancilla qubit $A_0'=(0,9)$, and is supported on data qubits $(6, 2) \equiv (0,10)\;(\modulo\mathcal{K})$, $(7, 3)$, $(9, 3)$, $(5, 11) \equiv (11, 3)\;(\modulo\mathcal{K})$, and  $(0, 0) \equiv (6, 12)\;(\modulo\mathcal{K})$.
    (b) An example of three equivalent parallelograms: $\mathcal{P}(\vec{v}_1=(12, 0), \vec{v}_2=(6, 4))$ in green, $\mathcal{P}(\vec{w}_1=(6, -4), \vec{w}_2=(0, 8))$ in orange, and $\mathcal{P}(\vec{u}_1=(18, -4), \vec{u}_2=(24, -8))$ in blue. Their associated (shared) sublattice $\mathcal{K}$ is also shown as large black points.
    }
    \label{fig:example_parallelogram_lattice}
\end{figure}

So far we worked on an infinite plane and constructed a syndrome extraction circuit that measured an infinite number of stabilisers. In order to define a CSS code, we wrap the infinite plane, with the qubits and the stabilisers, around a parallelogram so that we obtain a CSS code on a regular or twisted torus. We defined $\mathcal{P} = \mathcal{P}(\vec{v}_1, \vec{v}_2)$, $\mathcal{K}$ and the torus $\Z^2/\mathcal{K}$ obtained by wrapping around $\mathcal{P}$ in \Cref{sec:code-construct}. Note that using the notation introduced just before stating \Cref{thm:Thm1}, we may write $\mathcal{K} = \SPAN_\Z(\{\vec{v}_1, \vec{v}_2\})$. Whenever we consider a point on the torus $\Z^2/\mathcal{K}$, we will use a representative point from $\Z^2$ to denote it. Accordingly, two points $A, B \in \Z^2$ are the same on the torus if and only if $\overrightarrow{AB}\in\mathcal{K}$. In case we would like to emphasise this, we will use the notation $A\equiv B \;(\modulo\mathcal{K})$. \Cref{fig:parallelogram_lattice_sub} illustrates an example of how stabilisers defined on a parallelogram wrap around across its boundary, see also \Cref{fig:N2E3N2_on_parallelogram}. Now we prove a proposition that allows us to define a CSS code on a torus whose syndrome extraction circuit requires a square- or hexagonal-grid connectivity.

\begin{proposition}\label{prop:wrap_CSS}
    Consider a sequence of directions $\vec{d_1},\vec{d_2},\dots, \vec{d_\ell}$, a valid layout on the infinite plane $\Lambda\colon \Z^2_{\mathrm{anc}} \to \{X,Z\}$, and a parallelogram $\mathcal{P} = \mathcal{P}(\vec{v}_1, \vec{v}_2)$. Assume that 
    \begin{itemize}
        \item[(i)] $A\not\equiv Q\;(\modulo\mathcal{K})$ for all $A\in\Z^2_{\mathrm{anc}}$ and $Q\in\Z^2_{\mathrm{data}}$,
        \item[(ii)] $A\not\equiv A'\;(\modulo\mathcal{K})$ for all $A, A'\in\Z^2_{\mathrm{anc}}$ and $\Lambda(A)\neq\Lambda(A')$,
        \item[(iii)] $\Delta(\vec{d_1},\vec{d_2},\dots, \vec{d_\ell})\cap\mathcal{K}=\emptyset$,
        \item[(iv)] for any two delta vectors $\vec{u},\vec{w}\in\Delta$ we have $\{\vec{u}-\vec{w},\vec{u}+\vec{w}\}\cap\mathcal{K}\subseteq\{\vec{0}\}$.
    \end{itemize} 
    Let $\mathcal{P}_{\mathrm{data}}:=\Z^2_{\mathrm{data}}\cap\mathcal{P}$, and for each $A_0\in\mathcal{P}_{\mathrm{anc}}:=\Z^2_{\mathrm{anc}}\cap\mathcal{P}$ define the Pauli operator $g_{A_0}:=\prod_{j=1}^\ell \Lambda(A_0)_{Q_j}$ where $Q_j\in\mathcal{P}_{\mathrm{data}}$ is defined by \Cref{eq:Q_j} modulo $\mathcal{K}$. Then 
    \begin{equation}\label{eq:stabiliser_group}
        \mathcal{S}=\mathcal{S}(\mathcal{P}; \Lambda; \vec{d_1},\vec{d_2},\dots, \vec{d_\ell}) :=\langle g_{A_0}\colon \; A_0\in \mathcal{P}_{\mathrm{anc}} \rangle
    \end{equation}
    defines a CSS code on data qubits $\mathcal{P}_{\mathrm{data}}$. Furthermore, the circuit $\mathcal{C}(\Lambda; \vec{d_1},\vec{d_2},\dots, \vec{d_\ell})$ when considered on the torus achieves measuring all the stabilisers simultaneously and independently, and uses the set of qubits $\mathcal{P}_{\mathrm{data}}\cup\mathcal{P}_{\mathrm{anc}}$.
\end{proposition}

\begin{proof}
    Conditions (i)--(ii) ensure that the operator $g_{A_0}$ is well-defined. Furthermore, if in the definition of $g_{A_0}$ we replaced $A_0$ and $Q_j$ with other representations from $\Z^2$, then this defines the same operator on the torus. Condition (iii) implies that the data qubits of any stabiliser $g_{A_0}$ are different from each other on the torus. We will show that (iv) further implies that the circuit measures all the operators $g_{A_0}$ simultaneously and independently, which implicitly implies that $\mathcal{S}$ is Abelian. For this, consider two operators $g_{A_0}=\prod_{j=1}^\ell \Lambda(A_0)_{Q_j}$ and $g_{A_0'}=\prod_{j=1}^\ell \Lambda(A_0')_{Q_j'}$ such that $A_0\not\equiv A_0'\;(\modulo\mathcal{K})$, and we will show that they satisfy conditions (a)--(c) from the beginning of \Cref{sec:app-schedule_conflicts}. Note that \Cref{eq:primed_shift_same} implies $Q_i\not\equiv Q_i'\;(\modulo\mathcal{K})$ for all $i$, and as such (a) is satisfied. Assume now that $\Lambda(A_0)\neq \Lambda(A_0')$ and $Q_i\equiv Q_j'\;(\modulo\mathcal{K})$ for some $i\neq j$. By considering a different representation of $A_0$, we may assume that $Q_i=Q_j'$ as points on the plane. Note that for any indices $p,q$ we have $\overrightarrow{Q_pQ_q'\;} = \overrightarrow{Q_pQ_i} + \overrightarrow{Q_j'Q_q'\;}$. Since each vector on the right hand side is either a delta vector or its opposite is, we obtain from (iv) that $Q_p\equiv Q_q' \; (\modulo\mathcal{K})$ happens if and only if $Q_p = Q_q'$. Therefore, if the circuit measures all the operators simultaneously and independently on the plane, it also does so on the torus, which concludes the proof.
\end{proof}

We refer to such a CSS code on a regular or twisted torus as a $D_1D_2\dots D_\ell$-code.

We define two parallelograms $\mathcal{P}(\vec{v}_1, \vec{v}_2)$ and $\mathcal{P}(\vec{w}_1, \vec{w}_2)$ to be equivalent if they define the same torus. This happens if and only if the sublattices $\SPAN_\Z(\{\vec{v}_1, \vec{v}_2\})$ and $\SPAN_\Z(\{\vec{w}_1, \vec{w}_2\})$ are the same. \Cref{fig:NEEN_to_TC} and \Cref{fig:equivalent_parallelogram} show examples. We may write $\vec{w}_i = \sum_{j=1}^2\gamma_{ij}\vec{v}_j$ with some $\gamma_{ij}\in\mathbb{R}$, and define the matrix $\Gamma := [\gamma_{ij}]_{i,j=1}^2\in\mathbb{R}^{2\times 2}$. Note that $\vec{v}_i = \sum_{j=1}^2\beta_{ij}\vec{w}_j$ where $\Gamma^{-1} = [\beta_{ij}]_{i,j=1}^2$. We prove the following about the equivalence of parallelograms.

\begin{proposition}\label{prop:equiv_par}
    The parallelograms $\mathcal{P}(\vec{v}_1, \vec{v}_2)$ and $\mathcal{P}(\vec{w}_1, \vec{w}_2)$ are equivalent if and only if $\Gamma$ has integer elements and its determinant is either $-1$ or $+1$.
\end{proposition}

\begin{proof}
    Clearly, $\vec{w}_i \in \SPAN_\Z(\{\vec{v}_1, \vec{v}_2\})$ holds if and only if both elements in the $i$th row of $\Gamma$ are integers. Therefore, $\SPAN_\Z(\{\vec{w}_1, \vec{w}_2\}) \subseteq \SPAN_\Z(\{\vec{v}_1, \vec{v}_2\})$ if and only if $\Gamma\in\Z^{2\times 2}$. Note that because the two vectors in each pair are linearly independent, it follows that $\Gamma$ is invertible in $\mathbb{R}^{2\times 2}$. Thus, the reverse inclusion $\SPAN_\Z(\{\vec{v}_1, \vec{v}_2\}) \subseteq \SPAN_\Z(\{\vec{w}_1, \vec{w}_2\})$ holds if and only if 
    $\Gamma^{-1}\in\Z^{2\times 2}$. Therefore, we conclude that the two sublattices are the same if and only if $\Gamma, \Gamma^{-1}\in\Z^{2\times 2}$. 
    
    Now, on the one hand, both $\Gamma$ and $\Gamma^{-1}$ being integer matrices implies that both $\det\Gamma$ and $\frac{1}{\det\Gamma}$ are integers, forcing $\det\Gamma=\pm 1$. On the other hand, if $\Gamma\in\Z^{2\times 2}$ and $\det\Gamma=\pm 1$, then $\Gamma^{-1} = \frac{1}{\det\Gamma}\left[\begin{array}{cc}
         \gamma_{22} & -\gamma_{12}  \\
         -\gamma_{21} & \gamma_{11}
    \end{array}\right]$, and as such it has integer elements, which concludes the proof.
\end{proof}

This proposition also ensures that if we wrap the circuit $\mathcal{C}(\Lambda; \vec{d_1},\vec{d_2},\dots, \vec{d_\ell})$ around two equivalent parallelograms, then, given the conditions of \Cref{prop:wrap_CSS} hold, we obtain the same circuit and associated CSS code.

\section{Logical operator structure of rotated \texorpdfstring{$N^\alpha E^\beta N^\alpha$}{N$\alpha$E$\beta$N$\alpha$}-codes}
\label{sec:app-logicals}

The goal of this appendix is to describe the logical operator structure of rotated $N^\alpha E^\beta N^\alpha$ directional codes for $\alpha\geq1$, $\beta\geq 2$. Recall from \Cref{sec:code-construct} that the stabilisers are allocated according to Layout 1, and that for each $d$ divisible by four, the code is defined on the twisted torus specified by the diamond-shaped parallelogram 
\begin{equation}\label{eq:rotated_diamond}
    \mathcal{P}\left(\vec{v}_1=\left(-\tfrac{1}{2}\beta d, \alpha d\right), \vec{v}_2=\left(\tfrac{1}{2} \beta d, \alpha d\right)\right).
\end{equation}
We conjecture the code distance of all rotated $N^\alpha E^\beta N^\alpha$-codes to be exactly $d$, which was verified for all codes that were benchmarked in the paper. We will not describe the logical operators for the rectangular and filler codes, but note that they can be described in a similar way.

The logical operator description is divided into three subsections. First, in \Cref{sec:app-construct_block_cycle_operators} we construct ``cycle'' Pauli-$X$ operators of which we will have two types: horizontal and diagonal. These cycle operators consist of regularly spaced ``block'' operators around the torus, and as we will show they commute with all stabilisers and all have weight exactly $d$. Secondly, in \Cref{sec:app-cycle_operator_equiv} we investigate when two horizontal/diagonal cycle Pauli-$X$ operators are equivalent up to stabilisers, and we also identify sets of cycle operators whose product is a stabiliser of the code. With this we prove that there are at most $2(2\alpha-1)(\beta-1)$ many independent cycle Pauli-$X$ operators. Lastly, in \Cref{sec:app-logical_operator_construction} we show that these independent cycle operators are logical $X$-operators. In order to do that, we define the Pauli-$Z$ cycle operators as shifted versions of the Pauli-$X$ cycle operators whose Pauli types we also change from $X$ to $Z$. Then we prove that the anti-commutation matrix associated with these $2(2\alpha-1)(\beta-1)$ pairs of cycle Pauli-$X$- and $Z$-operators is invertible. This will enable us to express the logical $Z$-operators as products of cycle Pauli-$Z$ operators, and prove the lower bound $k \geq 2(2\alpha-1)(\beta-1)$ for the number of logical qubits. Even though this on its own does not exclude the possibility that some rotated code instances may possess additional logical qubits, we numerically never observed this to be the case. In particular, for all the benchmarked directional codes from \Cref{sec:simulations} and \Cref{sec:app-additional_results} we confirmed that $k=2(2\alpha-1)(\beta-1)$ always holds, and we conjecture this to be the case in general.

\begin{figure}[!htb]
    \centering
    \includegraphics[width=0.2\linewidth]{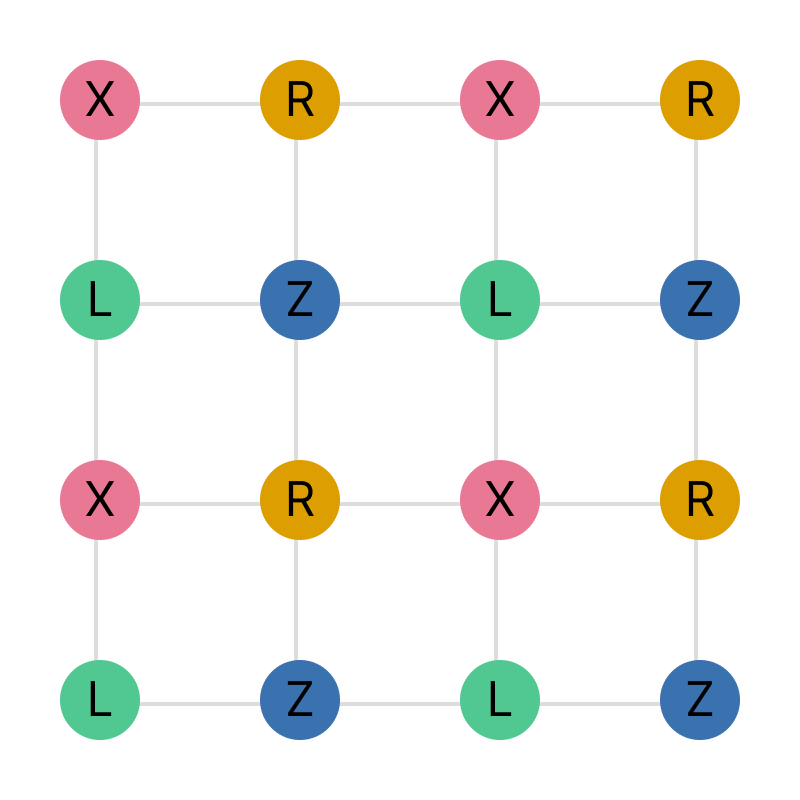}
    \caption{The four types of labels we use for qubits: $X$ (red) and $Z$ (blue) for ancilla qubits, and $L$ (green) and $R$ (orange) for data qubits.}
    \label{fig:XZLR_qubits}
\end{figure}

We introduce qubit labelling on $\Z^2$ with Layout 1 that we will be using throughout this appendix. We partition the set of ancilla qubits $\Z^2_{\mathrm{anc}}$ into the following two subsets:
\begin{equation}\label{eq:X_aq}
    \Z^2_X:= \{(x,y)\in\Z^2 : x\equiv 0,\; y\equiv 1 \; (\modulo 2)\}
\end{equation}
and
\begin{equation}\label{eq:Z_aq}
    \Z^2_Z := \{(x,y)\in\Z^2 : x\equiv 1,\; y\equiv 0 \; (\modulo 2)\}.
\end{equation}
The elements of $\Z^2_X$ and $\Z^2_Z$ are exactly the $X$- and $Z$-type ancilla qubits.
We also partition the set of data qubits $\Z^2_{\mathrm{data}}$ into two subsets:
\begin{equation}\label{eq:L_dq}
    \Z^2_L := \{(x,y)\in\Z^2 : x\equiv y\equiv 0 \; (\modulo 2)\}
\end{equation}
and
\begin{equation}\label{eq:R_dq}
    \Z^2_R := \{(x,y)\in\Z^2 : x\equiv y\equiv 1 \; (\modulo 2)\}.
\end{equation}
We will refer to elements of $\Z^2_L$ as $L$-data qubits, which always lie on rows of $Z$-ancilla qubits, or equivalently, on columns of $X$-ancilla qubits. We will term elements of $\Z^2_R$ as $R$-data qubits, which always lie on rows of $X$-ancilla qubits, or equivalently, on columns of $Z$-ancilla qubits. We depict this qubit labelling in \Cref{fig:XZLR_qubits}, and we note that it is the same as what was used in \cite{Bravyi2024}. Finally, we emphasise again that whenever we consider a point on the torus, we always use a representative point from $\Z^2$ to denote it, as in \Cref{sec:app-parallelograms}.

\subsection{Construction of horizontal and diagonal cycle operators}\label{sec:app-construct_block_cycle_operators}

The goal of this subsection is to construct two types of ``cycle'' Pauli-$X$ operators, horizontal and diagonal, that both loop around the torus and commute with all stabilisers of the code. In later sections, these cycle operators are then shown to be logical $X$-operators. 

We start by describing the anti-commutation relation between a $Z$-stabiliser and a weight-$1$ Pauli-$X$ operator. In order to do that we briefly consider the general case, namely, a directional sequence $\vec{d}_1,\vec{d}_2,...,\vec{d}_\ell$ that satisfies the conditions of \Cref{thm:Thm1} and \Cref{prop:wrap_CSS} with $\Lambda$ being Layout 1. We associate each $D_1D_2\dots D_\ell$-stabiliser with its ancilla qubit (at its initial location). More precisely, with a $Z$-ancilla qubit $A_0\in\Z^2_Z$ we associate the $Z$-stabiliser $g_{A_0}$ supported on the data qubits $Q_1,Q_2,\dots,Q_\ell\in\Z^2_{\mathrm{data}}$ defined by \Cref{eq:Q_j} (similarly for $X$-ancilla qubits). Obviously, a weight-$1$ Pauli-$X$ operator supported on the data qubit $Q\in\Z^2_{\mathrm{data}}$ anti-commutes with $g_{A_0}$ if and only if $Q\in\{Q_1,Q_2,\dots,Q_\ell\}$. Also, \Cref{eq:Q_j} and \Cref{eq:X_aq,eq:Z_aq,eq:L_dq,eq:R_dq} readily imply that if $Q=Q_j$, then $Q$ is an $L$-data qubit if $\vec{d_j}\in\{\vec{e},\vec{w}\}$, and otherwise an $R$-data qubit. Therefore, the set of $L$- and $R$-data qubits that support weight-$1$ Pauli-$X$ operators which anti-commute with $g_{A_0}$ are exactly
\begin{equation}\label{eq:general_L}
    \left\{ A_0 + 2\sum_{p=1}^{j-1} \vec{d_p} + \vec{d_j} : 1 \leq j \leq \ell,\; \vec{d}_j \in \{\vec{e},\vec{w}\}\right\}
\end{equation}
and
\begin{equation}\label{eq:general_R}
    \left\{ A_0 + 2\sum_{p=1}^{j-1} \vec{d_p} + \vec{d_j} : 1 \leq j \leq \ell,\; \vec{d}_j \in \{\vec{n},\vec{s}\}\right\},
\end{equation}
respectively. 

\begin{figure*}[!htbp]
     \centering
     \begin{subfigure}[t]{0.45\textwidth}
         \centering
         \vspace{0pt}
         \includegraphics[width=0.8125\linewidth]{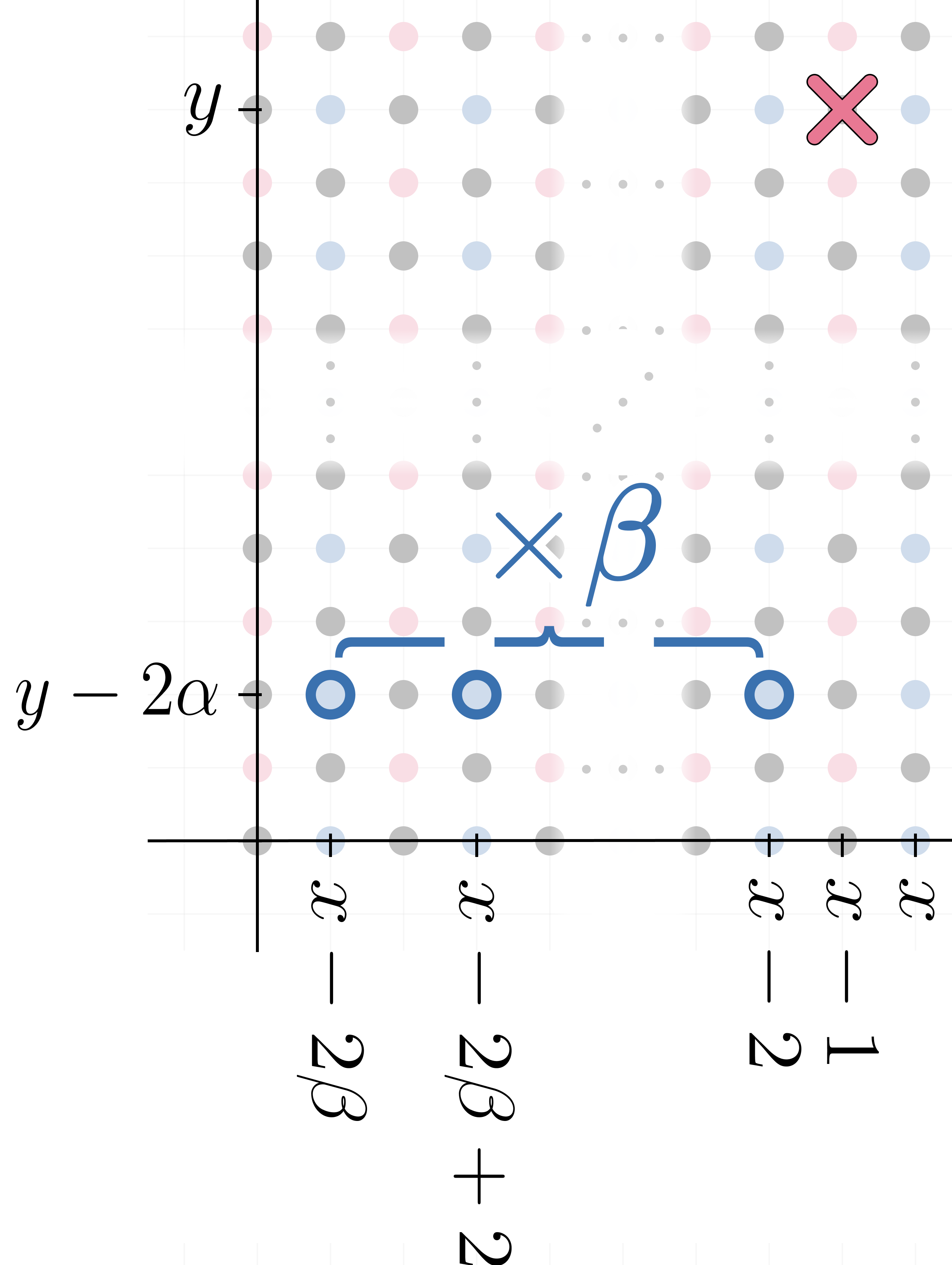}
         \caption{A weight-$1$ Pauli-$X$ operator on the $L$-data qubit $(x-1,y)$, and the $Z$-ancilla qubits corresponding to $N^\alpha E^\beta N^\alpha$-stabilisers anti-commuting with it, see \Cref{eq:Z_anti-comm_L}.}
         \label{fig:NaEbNa_x_error_l_qubit}
     \end{subfigure}
     \hspace{0.25cm}
     \begin{subfigure}[t]{0.45\textwidth}
         \centering
         \vspace{0pt}
         \includegraphics[width=1.0\linewidth]{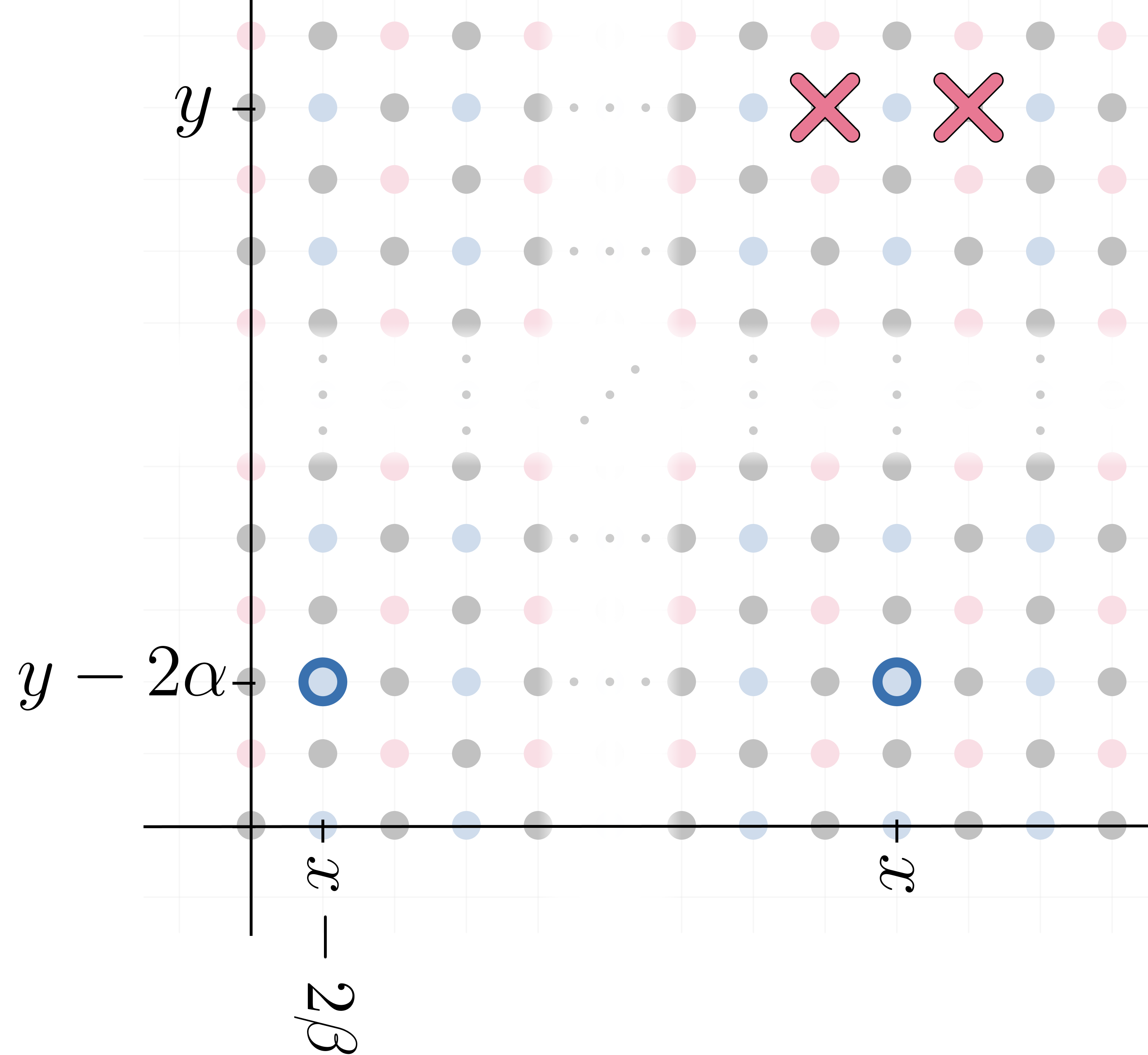}
         \vspace{0.65cm}
         \caption{The weight-$2$ Pauli-$X$ block operator $B^X_{\mathrm{hor}}(x,y)$, which is the building block of the horizontal cycle Pauli-$X$ operators. It anti-commutes with two $Z$-type $N^\alpha E^\beta N^\alpha$-stabilisers, see \Cref{eq:anti-comm_hor_block}.}
         \label{fig:NaEbNa_xx_error_l_qubits}
     \end{subfigure}\\[5mm]
     \begin{subfigure}[t]{0.45\textwidth}
         \centering
         \includegraphics[width=1.0\linewidth]{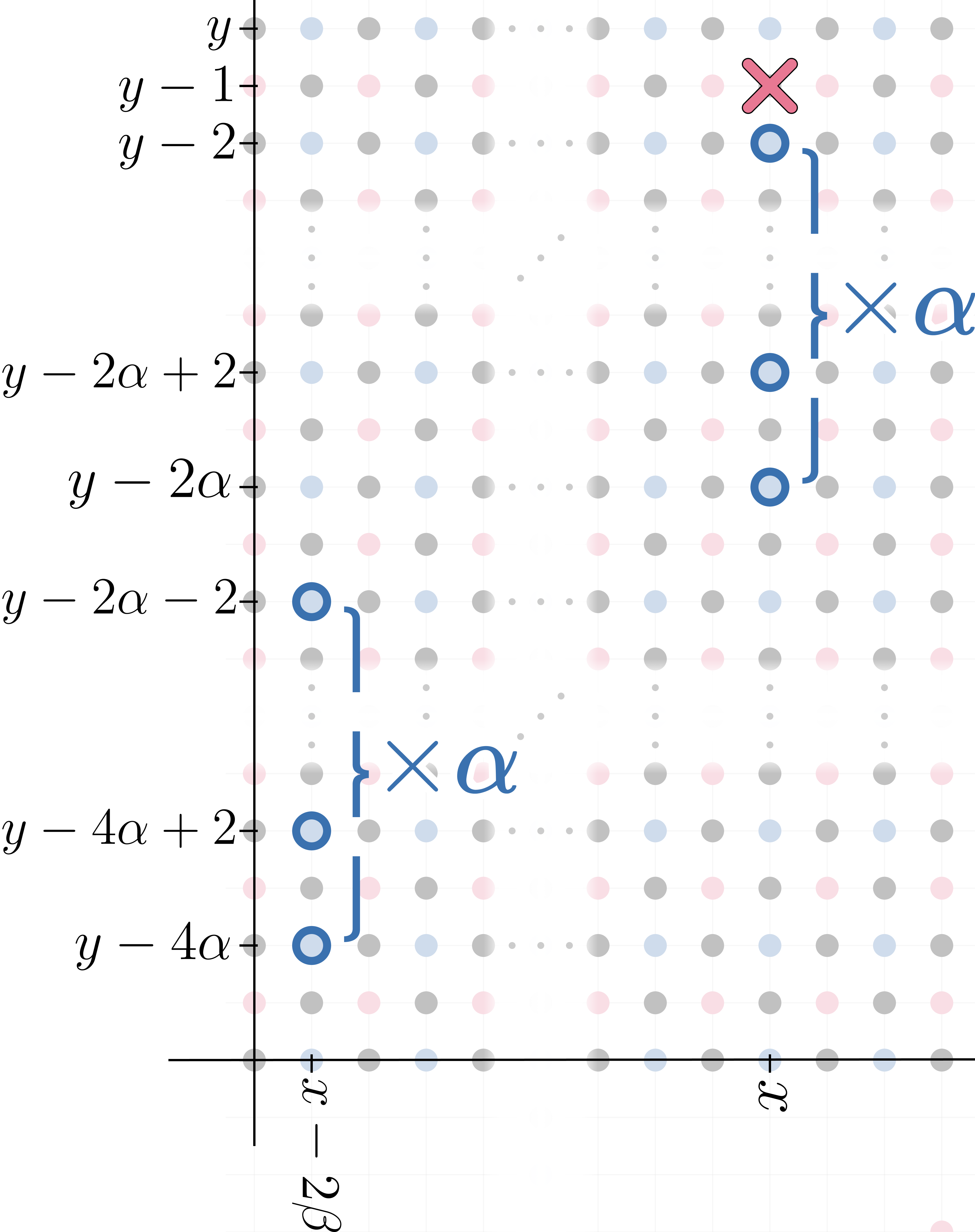}
         \caption{A weight-$1$ Pauli-$X$ operator on the $R$-data qubit $(x,y-1)$, and the $Z$-ancilla qubits corresponding to $N^\alpha E^\beta N^\alpha$-stabilisers anti-commuting with it, see \Cref{eq:Z_anti-comm_R}.}
         \label{fig:NaEbNa_x_error_r_qubit}
     \end{subfigure}
     \hspace{0.25cm}
     \begin{subfigure}[t]{0.45\textwidth}
         \centering
         \includegraphics[width=0.88\linewidth]{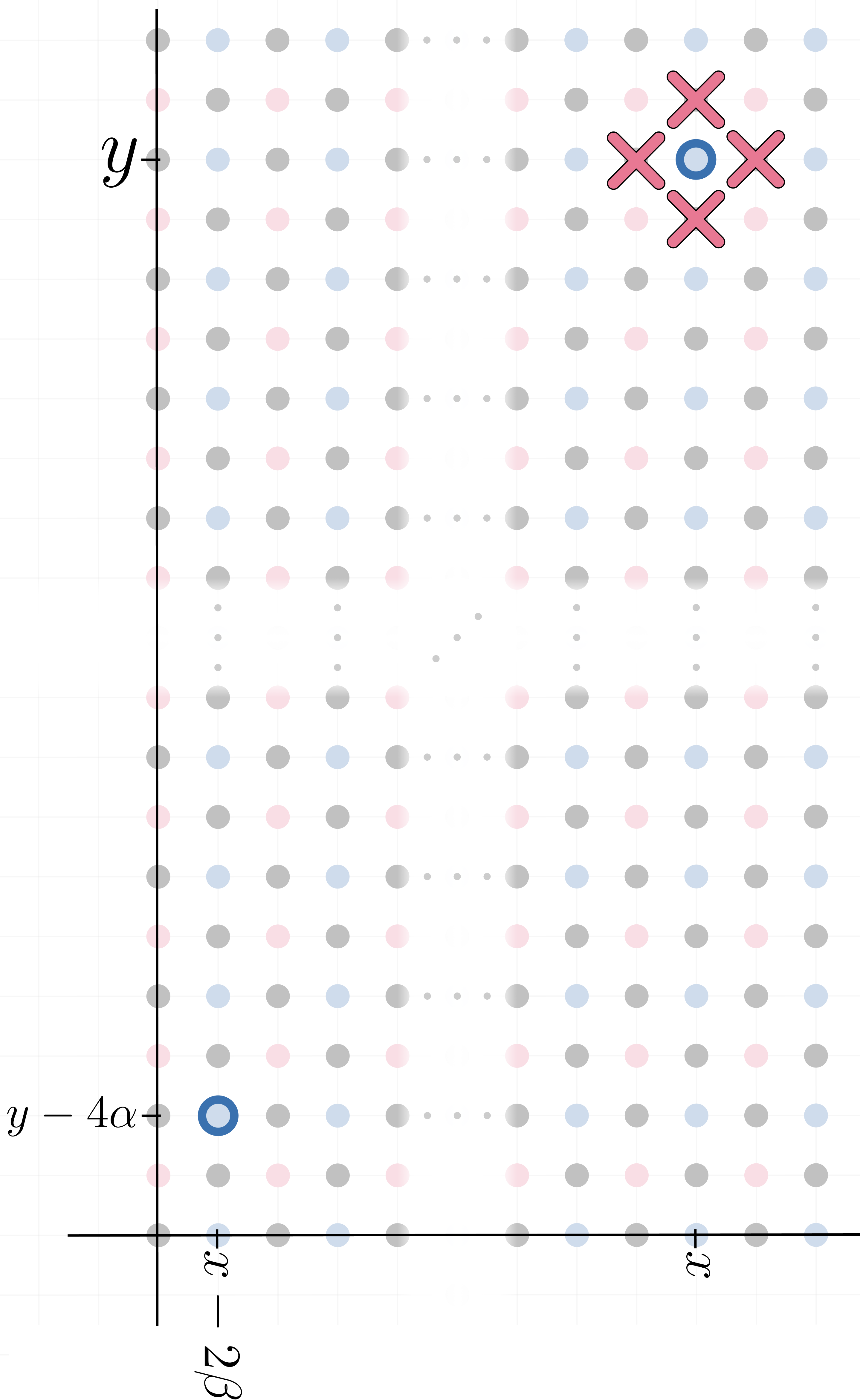}
         \caption{The weight-$4$ Pauli-$X$ block operator $B^X_{\mathrm{diag}}(x,y)$, which is the building block of the diagonal cycle Pauli-$X$ operators. It anti-commutes with two $Z$-type $N^\alpha E^\beta N^\alpha$-stabilisers, see \Cref{eq:anti-comm_diag_block}.}
         \label{fig:NaEbNa_xxxx_error}
     \end{subfigure}
     \caption{Each panel shows a Pauli-$X$ operator (red crosses) and the $Z$-ancilla qubits (highlighted blue dots) that correspond to $N^\alpha E^\beta N^\alpha$-stabilisers which anti-commute with it.}
\end{figure*}

From now on until the end of the Appendix, we consider $N^\alpha E^\beta N^\alpha$-codes. We also reserve the notation $(x,y)$ to denote a $Z$-ancilla qubit. As such, $(x\pm 1,y)$ is always an $L$-data qubit, $(x, y\pm 1)$ is always an $R$-data qubit, and $(x\pm 1, y\pm 1)$ is always an $X$-ancilla qubit. From \Cref{eq:general_L} we infer that the set of $L$-data qubits which support weight-$1$ Pauli-$X$ operators that anti-commute with the $Z$-type $N^\alpha E^\beta N^\alpha$-stabiliser whose ancilla qubit is at $(x,y)\in\Z^2_Z$ is
\begin{equation}\label{eq:L_anti-comm_Z}
    \left\{ (x+2j-1,y+2\alpha) : 1 \leq j \leq \beta \right\}.
\end{equation}
This set has $\beta$ elements. Similarly, for the $R$-data qubits we obtain from \Cref{eq:general_R}
\begin{equation}\label{eq:R_anti-comm_Z}
    \left\{ (x,y+2j-1), (x+2\beta,y+2j-1+2\alpha) : 1 \leq j \leq \alpha \right\},
\end{equation}
which has 2$\alpha$ elements. From \Cref{eq:L_anti-comm_Z} we obtain that the weight-$1$ Pauli-$X$ operator supported on the $L$-data qubit $(x-1,y)\in\Z^2_L$ anti-commutes with exactly those $Z$-type $N^\alpha E^\beta N^\alpha$-stabilisers whose ancilla qubits are elements of the set
\begin{equation}\label{eq:Z_anti-comm_L}
    \left\{ (x-2j,y-2\alpha) : 1 \leq j \leq \beta \right\}.
\end{equation}
This is depicted in \Cref{fig:NaEbNa_x_error_l_qubit}. Similarly, \Cref{eq:R_anti-comm_Z} implies that the weight-$1$ Pauli-$X$ operator supported on the $R$ data-qubit $(x,y-1)\in\Z^2_R$ anti-commutes with exactly the $Z$-type $N^\alpha E^\beta N^\alpha$-stabilisers whose ancilla qubits are elements of the set
\begin{equation}\label{eq:Z_anti-comm_R}
    \left\{ (x,y-2j), (x-2\beta,y-2j-2\alpha) : 1 \leq j \leq \alpha \right\},
\end{equation}
see \Cref{fig:NaEbNa_x_error_r_qubit} for an illustration. 

We now construct higher weight Pauli-$X$ operators that anti-commute with exactly two $Z$-type $N^\alpha E^\beta N^\alpha$-stabilisers. We will refer to these operators as ``block'' operators, and they will form the building blocks of the cycle operators we will later construct. First, consider the weight-$2$ Pauli-$X$ operator supported on two horizontally next-nearest-neighbour $L$-data qubits, which we parametrise by the $Z$-ancilla qubit that lies in between them:
\begin{equation}\label{eq:x_horizontal_block}
        B^X_{\mathrm{hor}}(x,y) := X_{(x - 1, y)} X_{(x+1, y)} 
        \qquad (x,y)\in\Z^2_Z.
\end{equation}
We call this a horizontal block Pauli-$X$ operator. It is straightforward to see from \Cref{eq:Z_anti-comm_L} that the only two $Z$-type $N^\alpha E^\beta N^\alpha$-stabilisers that anti-commute with $B^X_{\mathrm{hor}}(x,y)$ have ancilla qubits located at
\begin{equation}\label{eq:anti-comm_hor_block}
    (x,y-2\alpha) \text{ and } (x-2\beta,y-2\alpha),
\end{equation}
see \Cref{fig:NaEbNa_xx_error_l_qubits} for an illustration. Second, consider the weight-$4$ Pauli-$X$ operator supported on the four nearest-neighbours of a $Z$ ancilla qubit:
\begin{equation}\label{eq:x_diag_block}
    B^X_{\mathrm{diag}}(x,y):=X_{(x - 1, y)} X_{(x + 1, y)} X_{(x, y - 1)} X_{(x, y + 1)}
    \qquad (x,y)\in\Z^2_Z.
\end{equation}
We call this a diagonal block Pauli-$X$ operator. It follows from \Cref{eq:Z_anti-comm_R,eq:anti-comm_hor_block} and \Cref{fig:NaEbNa_x_error_r_qubit,fig:NaEbNa_xx_error_l_qubits} that the only two $Z$-type $N^\alpha E^\beta N^\alpha$-stabilisers that anti-commute with $B^X_{\mathrm{diag}}(x,y)$ have ancilla qubits located at
\begin{equation}\label{eq:anti-comm_diag_block}
    (x,y) \text{ and } (x-2\beta,y-4\alpha),
\end{equation}
see \Cref{fig:NaEbNa_xxxx_error} for an illustration.

\begin{figure}[htbp]
    \centering
    \begin{subfigure}[b]{0.56\textwidth}
        \centering
        \begin{subfigure}[b]{\textwidth}
            \centering
            \includegraphics[width=\textwidth]{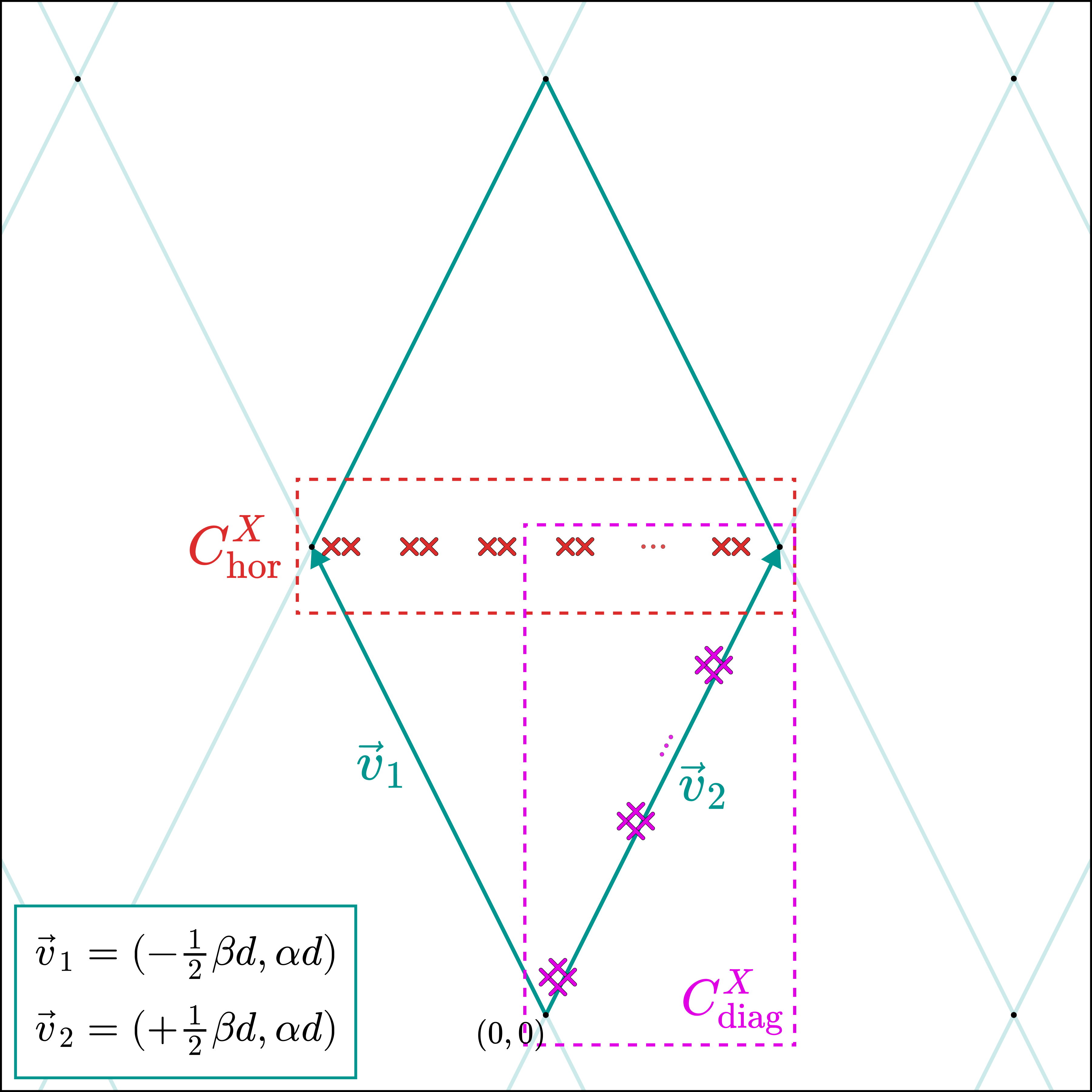}
            \caption{Zoomed-out view of the torus and the two types of cycle operators.}
         \label{fig:zoomed_out_x_cycle_operators}
        \end{subfigure}
        
        \vspace{0.5cm} 
        
        \begin{subfigure}[b]{\textwidth}
            \centering
            \includegraphics[width=\textwidth]{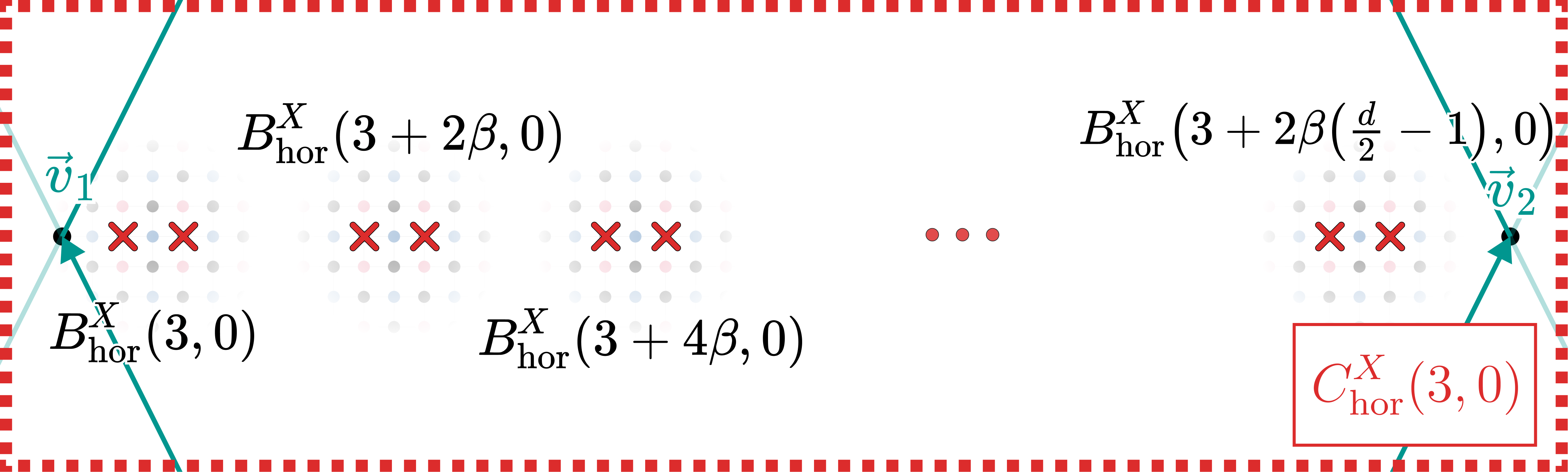}
         \caption{The operator $C^X_{\mathrm{hor}}(3,0)$.}
         \label{fig:zoomed_in_x_horizontal_cycle_operator}
        \end{subfigure}
    \end{subfigure}
    \hfill 
    \begin{subfigure}[b]{0.40\textwidth}
        \centering
        \includegraphics[width=\textwidth]{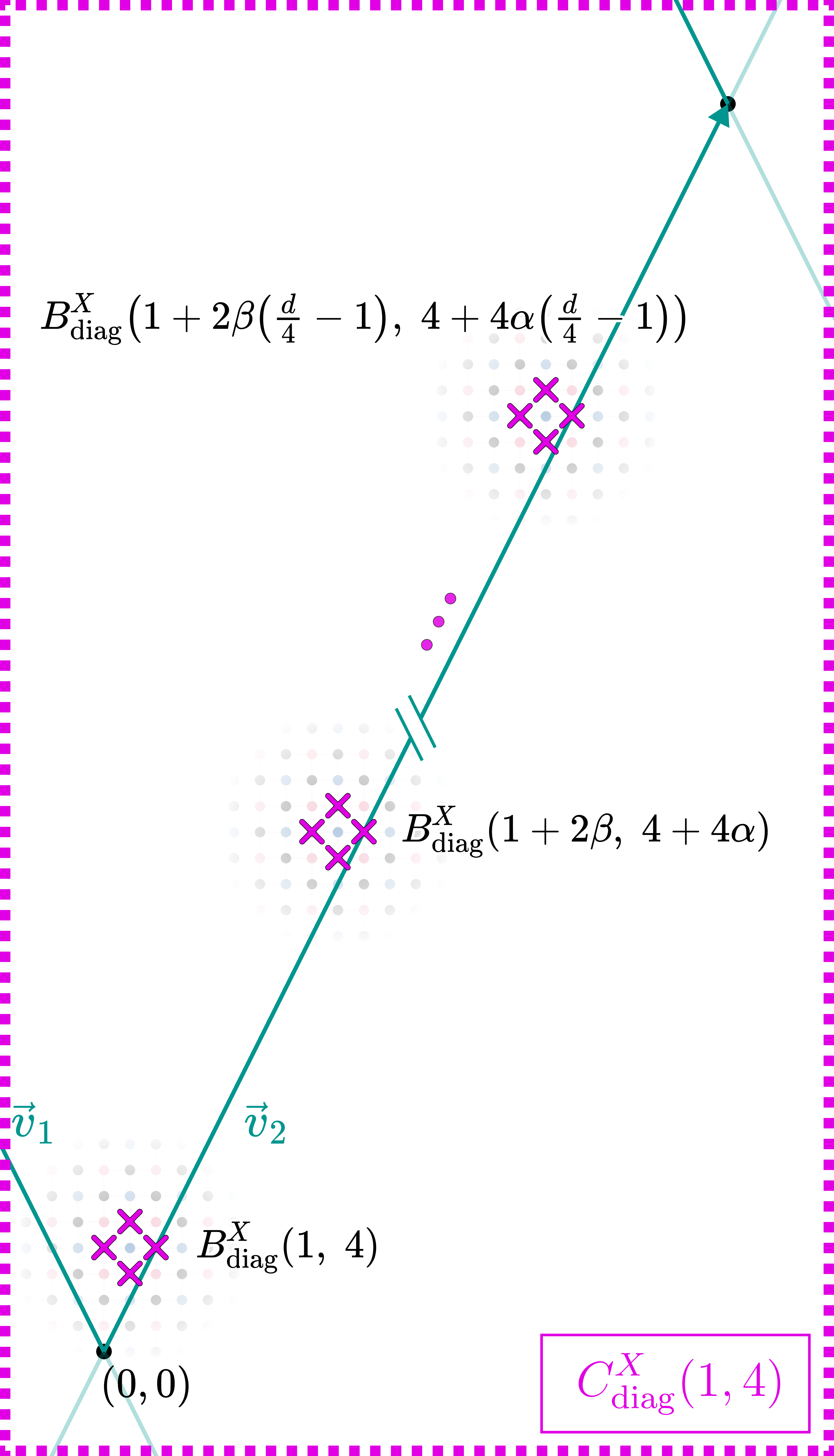}
         \caption{The operator $C^X_{\mathrm{diag}}(1,4)$.}
         \label{fig:zoomed_in_x_diagonal_cycle_operator}
    \end{subfigure}
    \caption{Horizontal and diagonal cycle operators wrapped around the diamond-shaped parallelogram (\Cref{eq:rotated_diamond}) on which the rotated $N^\alpha E^\beta N^\alpha$-code is defined. A zoomed-out view of the torus with both types of cycle operators is depicted in (a). The horizontal cycle operator is inside the orange dashed rectangle, while the diagonal cycle operator is inside the purple dashed rectangle. A zoomed-in view of the horizontal cycle operator $C^X_{\mathrm{hor}}(3,0)$ from (a) is shown in (b). A zoomed in view of the diagonal cycle operator $C^X_{\mathrm{diag}}(1,4)$ from (a) is displayed in (c).}
     \label{fig:x_cycle_operators}
\end{figure}

Now, we use the block operators to define the cycle operators. By \Cref{prop:equiv_par}, the parallelogram from \Cref{eq:rotated_diamond} is equivalent to 
\begin{equation}\label{eq:rotated_hermite_repr}
    \mathcal{P}\left(\vec{v}_1=\left(-\frac{1}{2}\beta d, \alpha d\right), \vec{v}_2-\vec{v}_1=\left(\beta d, 0\right)\right).
\end{equation}
As such, two points $A,B\in\Z^2$ with $\overrightarrow{AB}=\gamma\vec{e}$ are the same on the torus if and only if $\gamma$ is divisible by $\beta d$. Therefore, for all $Z$-ancilla qubits $(x,y)\in\Z^2_Z$ the following operator is well-defined and has weight exactly $d$:
\begin{equation}\label{eq:x_hor_operator}
    C^X_{\mathrm{hor}}(x, y) := \prod_{m=0}^{\frac{d}{2}-1}B^X_{\mathrm{hor}}(x+2\beta m,y)
    \quad (x,y)\in\Z^2_Z.
\end{equation}
We call $C^X_{\mathrm{hor}}(x, y)$ the horizontal Pauli-$X$ cycle operator associated with $(x,y)$, see \Cref{fig:zoomed_out_x_cycle_operators,fig:zoomed_in_x_horizontal_cycle_operator} for an illustration of $C^X_{\mathrm{hor}}(3,0)$. Notice that by \Cref{eq:anti-comm_hor_block} and \Cref{fig:NaEbNa_xx_error_l_qubits}, the horizontal cycle operators commute with all stabilisers of the code. 

Next, the definition of the parallelogram in \Cref{eq:rotated_diamond} implies that two points $A,B\in\Z^2$ with $\overrightarrow{AB}=\gamma(\beta\vec{e}+2\alpha\vec{n})$ are the same on the torus if and only if $\gamma$ is divisible by $\frac{d}{2}$. Therefore, for all $Z$-ancilla qubits $(x,y)\in\Z^2_Z$ the following operator is well-defined and has weight exactly $d$:
\begin{equation}\label{eq:x_diag_operator}
    C^X_{\mathrm{diag}}(x, y) := \prod_{m=0}^{\frac{d}{4}-1}B^X_{\mathrm{diag}}(x+2\beta m,y+4\alpha m)
    \qquad (x,y)\in\Z^2_Z.
\end{equation}
We call $C^X_{\mathrm{diag}}(x, y)$ the diagonal Pauli-$X$ cycle operator associated with $(x,y)$, see \Cref{fig:zoomed_out_x_cycle_operators,fig:zoomed_in_x_diagonal_cycle_operator} for an illustration of $C^X_{\mathrm{diag}}(1,4)$. Notice that by \Cref{eq:anti-comm_diag_block} and \Cref{fig:NaEbNa_xxxx_error}, the diagonal cycle operators commute with all stabilisers of the code.

With this we identified a natural set of Pauli-$X$ operators that, due to their cycle structure, are natural candidates for $X$-logical operators. In the next subsection we investigate the relationship between sets of cycle operators up to stabilisers.

\subsection{Dependence and equivalence of cycle operators up to stabilisers}\label{sec:app-cycle_operator_equiv}

In this subsection, we show that some products of cycle operators are stabilisers of the code, thereby demonstrating dependency between them. We start with the horizontal cycle operators. First, directly from \Cref{eq:rotated_hermite_repr,eq:x_hor_operator} we obtain that shifting a horizontal cycle Pauli-$X$ operator by $2\beta\vec{e}$ gives back the same operator, i.e.
\begin{equation}\label{eq:hor_2beta_east_shift}
    C^X_{\mathrm{hor}}(x, y) = C^X_{\mathrm{hor}}(x+2\beta, y) \;\;\text{for all}\;\; (x,y)\in\Z^2_Z.
\end{equation}
Second, observe that the product of all the $\beta$ horizontal cycle Pauli-$X$ operators that are supported on the same horizontal line is the identity operator, i.e.
\begin{equation}\label{eq:beta_hor_prod_on_hor_line}
    \prod_{p=0}^{\beta-1}{C^X_{\mathrm{hor}}(x + 2p, y)}=I \;\;\text{for all}\;\; (x,y)\in\Z^2_Z.
\end{equation}
This is straightforward from the fact that the Pauli-$X$ terms of $C^X_{\mathrm{hor}}(x+2p+2, y)$ and $C^X_{\mathrm{hor}}(x+2p, y)$ cancel at $(x+2p+1+2\beta m, y)$, for all $m=0,\dots,\frac{d}{2}-1$. 

\begin{figure*}[!htbp]
     \centering
     \begin{subfigure}[t]{0.45\textwidth}
         \centering
         \includegraphics[height=8cm]{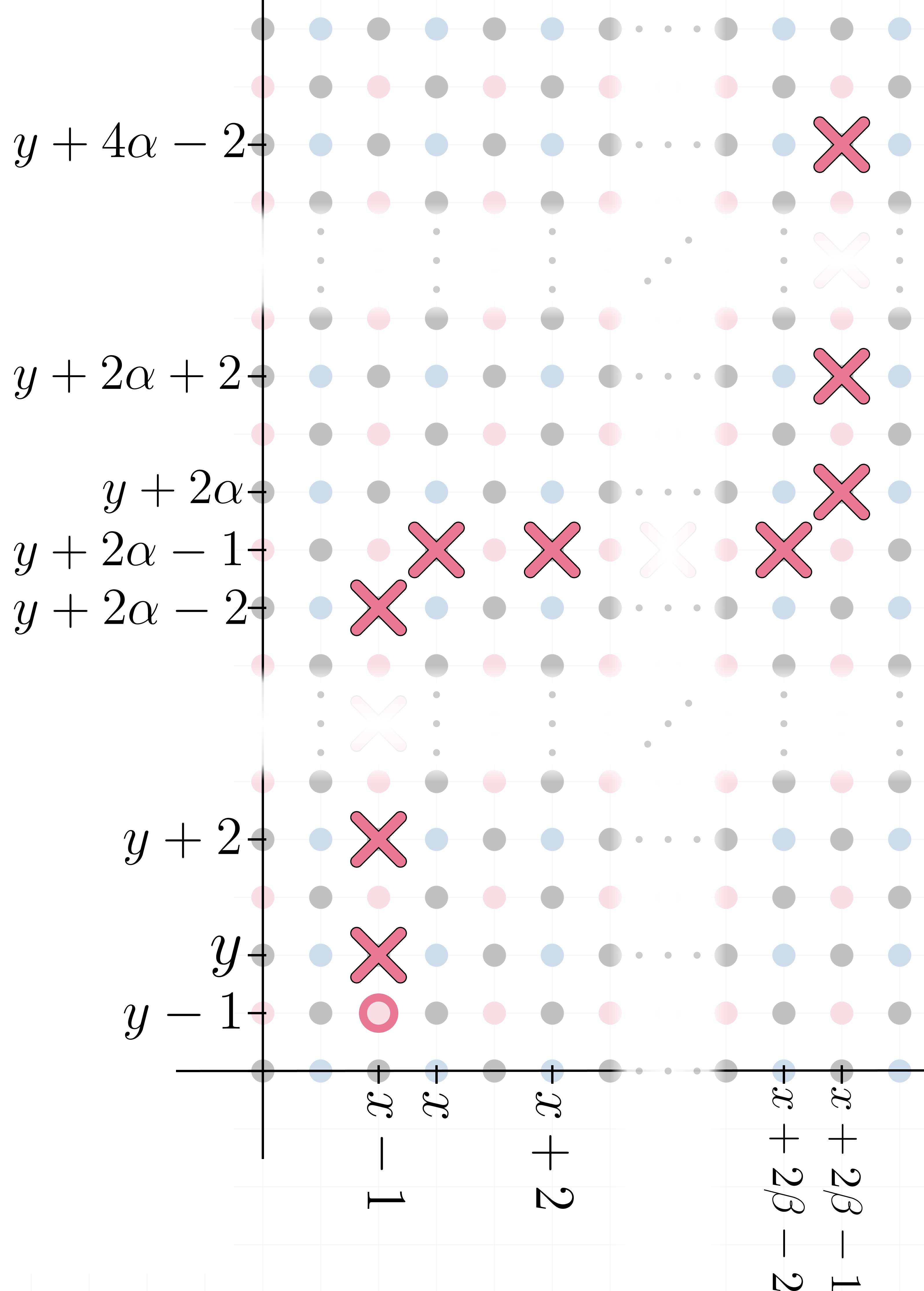}
         \caption{An $X$-stabiliser with associated ancilla qubit $(x-1,y-1)\in\Z^2_X$.}
         \label{fig:NaEbNa_stabiliser}
     \end{subfigure}
     \hspace{0.25cm}
     \begin{subfigure}[t]{0.45\textwidth}
         \centering
         \includegraphics[height=8cm]{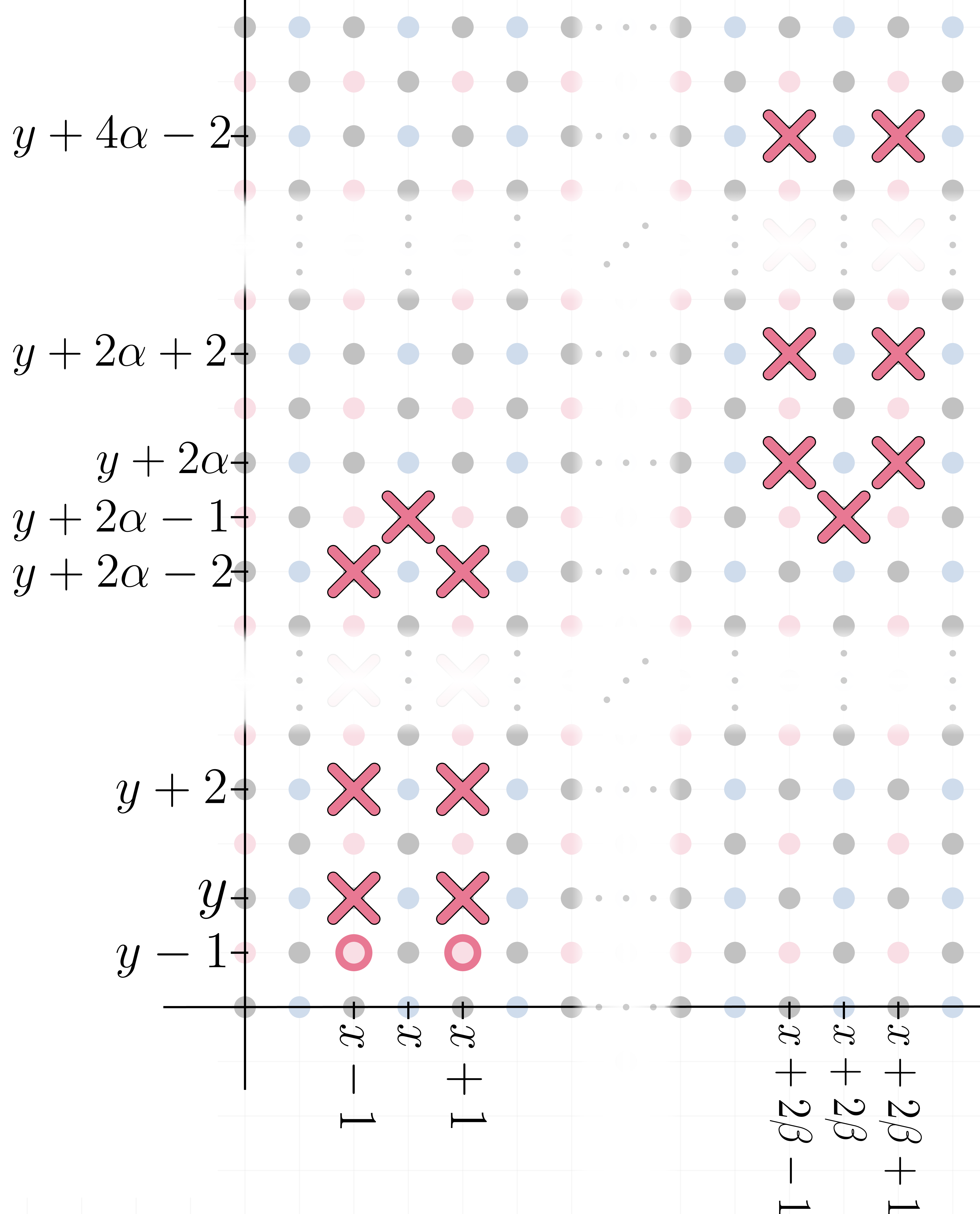}
         \caption{The product of two $X$-stabilisers with associated ancilla qubits $(x-1,y-1),(x+1,y-1)\in\Z^2_X$.}
         \label{fig:two_horizontal_NaEbNa_stabilisers}
     \end{subfigure}
     \caption{An $X$-type $N^\alpha E^\beta N^\alpha$-stabiliser is shown in (a), and the product of two such $X$-stabilisers is depicted in (b). Pauli $X$ terms of these operators are shown as red crosses, and the ancilla qubits are highlighted with red circles.}
     \label{fig:X_NaEbNa_stabilisers_hor_dep}
\end{figure*}

Next, consider the two $X$-stabilisers with associated ancilla qubits $(x-1,y-1), (x+1,y-1)\in\Z^2_X$. The former is depicted in \Cref{fig:NaEbNa_stabiliser}. The product of the two $N^\alpha E^\beta N^\alpha$-stabilisers can be expressed in the following way:
\begin{equation}\label{eq:two_horizontal_NaEbNa_stabilisers}
    X_{(x,y+2\alpha-1)}X_{(x+2\beta,y+2\alpha-1)}\prod_{q=0}^{\alpha-1} B^X_{\mathrm{hor}}(x,y+2q)B^X_{\mathrm{hor}}(x+2\beta,y+2\alpha+2q),
\end{equation}
see \Cref{fig:two_horizontal_NaEbNa_stabilisers}. Notice that if we shift this operator with the vector $2\beta\vec{e}$, then the $X_{(x+2\beta,y+2\alpha-1)}$ term is cancelled in the product of the original and the shifted operator. Therefore, if we apply the same shift consecutively $\frac{d}{2}-1$ times and take the product of the obtained $\frac{d}{2}$ operators, we obtain a stabiliser of the code that is of the form
\begin{equation}\label{eq:2alpha_hor_prod_vert_cons_shifts}
    \prod_{q=0}^{\alpha-1}\prod_{m=0}^{\frac{d}{2}-1} B^X_{\mathrm{hor}}\left(x+2\beta m,y+2q\right)B^X_{\mathrm{hor}}\left(x+2\beta (m+1),y+2\alpha+2q\right)
    = \prod_{q=0}^{2\alpha-1}{C^X_{\mathrm{hor}}(x, y+2q)},
\end{equation}
for all $(x,y)\in\Z^2_Z$, see \Cref{fig:2alpha_hor_prod_vert_cons_shifts}. This demonstrates that if we shift a horizontal cycle Pauli-$X$ operator with the vector $2\vec{n}$ consecutively $2\alpha-1$ times, then the product of these and the original operator is a stabiliser of the code. This fact readily implies that the horizontal cycle Pauli-$X$ operators $C^X_{\mathrm{hor}}(x, y)$ and $C^X_{\mathrm{hor}}(x, y+4\alpha)$ are equivalent up to stabilisers. Indeed, we simply need to multiply the stabiliser from \Cref{eq:2alpha_hor_prod_vert_cons_shifts} with the stabiliser we obtain by translating it by the vector $2\vec{n}$, see \Cref{fig:hor_4alpha_north_shift}. With this we showed that up to stabiliser equivalence we have at most $(2\alpha-1)(\beta-1)$ independent horizontal cycle Pauli-$X$ operators, namely, 
\begin{equation}\label{eq:candidate_hor_X_logicals}
    C^X_{\mathrm{hor}}(1+2p, 2q) \qquad (0\leq p \leq \beta-2,\; 0\leq q \leq 2\alpha-2).
\end{equation}

\begin{figure*}[!htbp]
     \centering
     \begin{subfigure}[t]{0.45\textwidth}
         \centering
         \includegraphics[height=5cm]{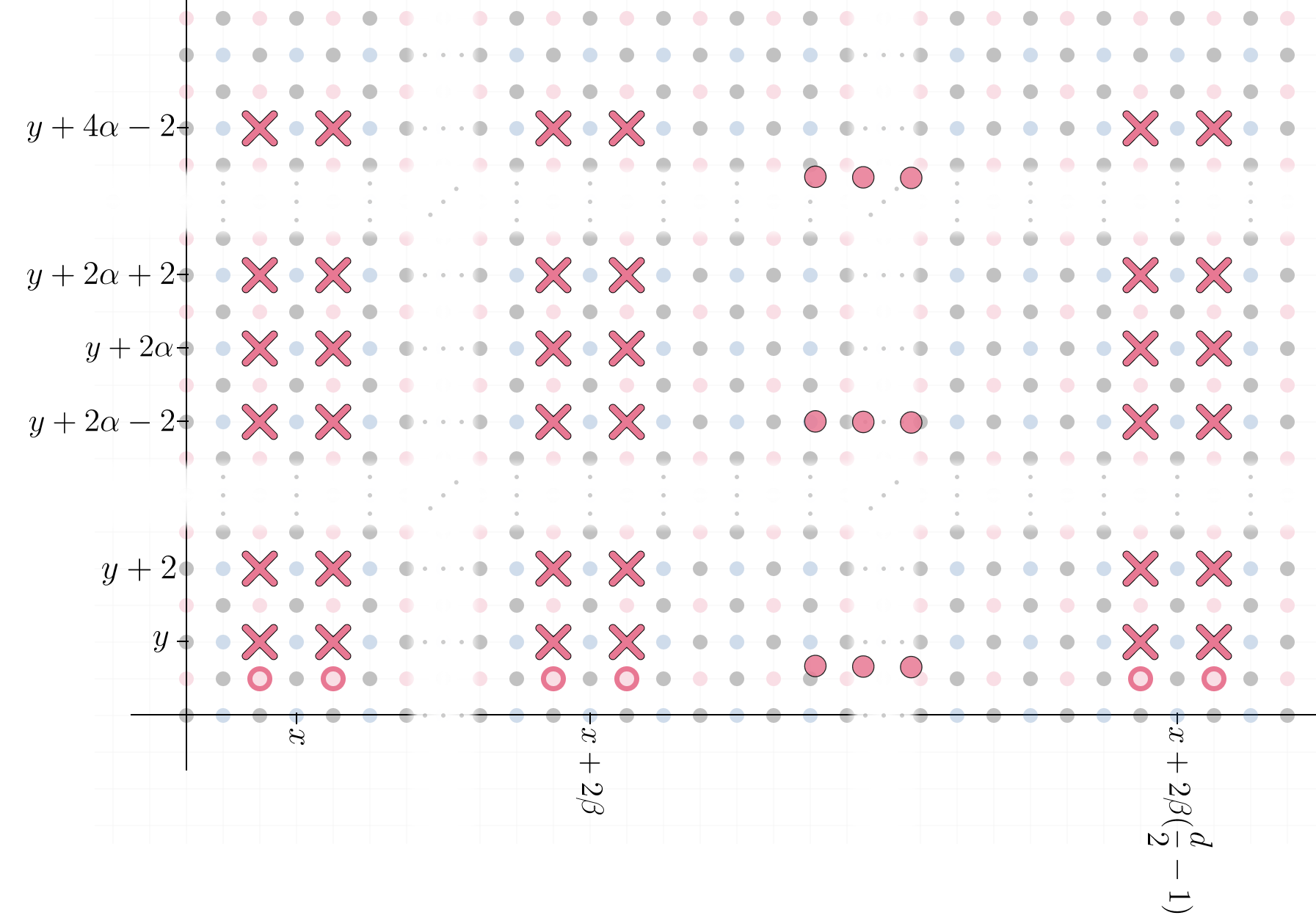}
         \caption{The operator from \Cref{eq:2alpha_hor_prod_vert_cons_shifts} that is the product of $X$-stabilisers whose ancilla qubits are highlighted. This operator is also a product of $2\alpha$ horizontal cycle Pauli-$X$ operators.}
         \label{fig:2alpha_hor_prod_vert_cons_shifts}
     \end{subfigure}
     \hspace{0.25cm}
     \begin{subfigure}[t]{0.45\textwidth}
         \centering
         \includegraphics[height=5cm]{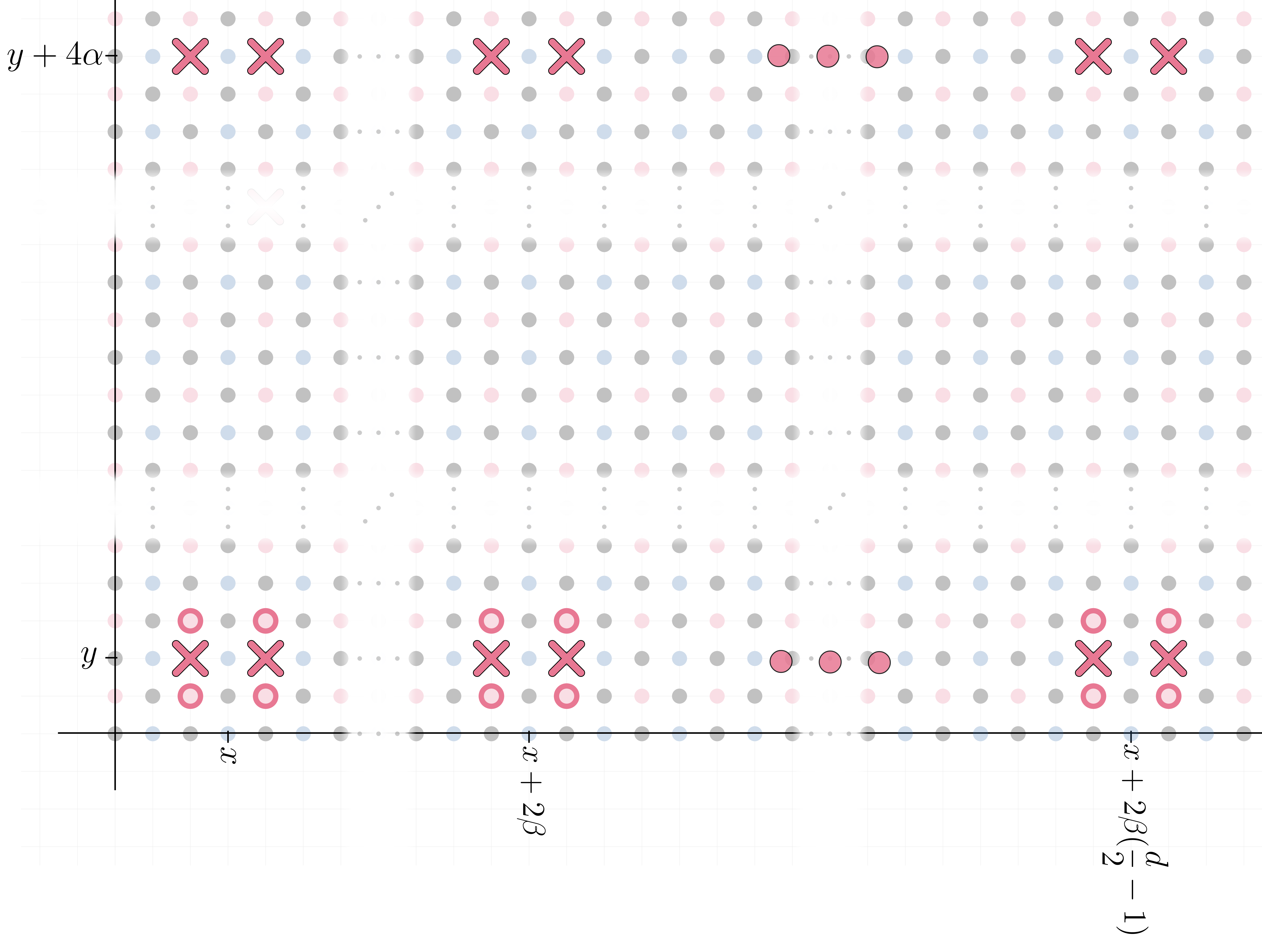}
         \caption{Two horizontal cycle Pauli-$X$ operators that are equivalent up to stabilisers. The top one is the translated version of the bottom one with the vector $4\alpha\vec{n}$. The ancilla qubits of the $X$-stabilisers whose product is equal to the product of these two cycle operators are highlighted.}
         \label{fig:hor_4alpha_north_shift}
     \end{subfigure}
     \caption{Dependencies between horizontal cycle Pauli-$X$ operators up to stabilisers.}
\end{figure*}

\begin{figure*}[!htbp]
     \centering
     \begin{subfigure}[t]{0.30\textwidth}
         \centering
         \includegraphics[width=1.0\linewidth]{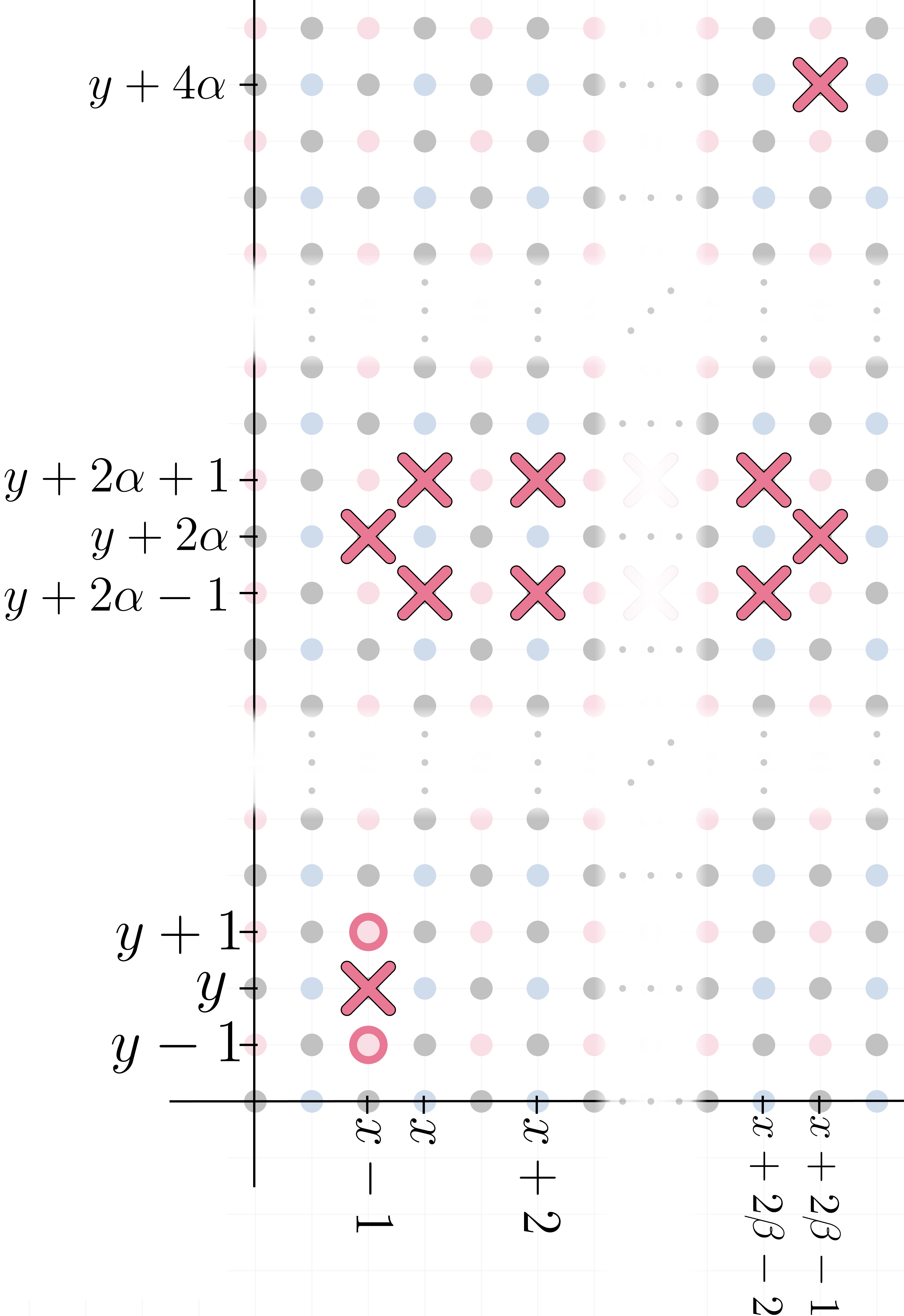}
         \caption{The product of two $X$-type $N^\alpha E^\beta N^\alpha$-stabilisers with associated ancilla qubits $(x-1,y-1), (x-1,y+1)\in\Z^2_X$.}
         \label{fig:two_vertical_NaEbNa_stabilisers}
     \end{subfigure}
     \hspace{0.25cm}
     \begin{subfigure}[t]{0.60\textwidth}
         \centering
         \includegraphics[width=1.0\linewidth]{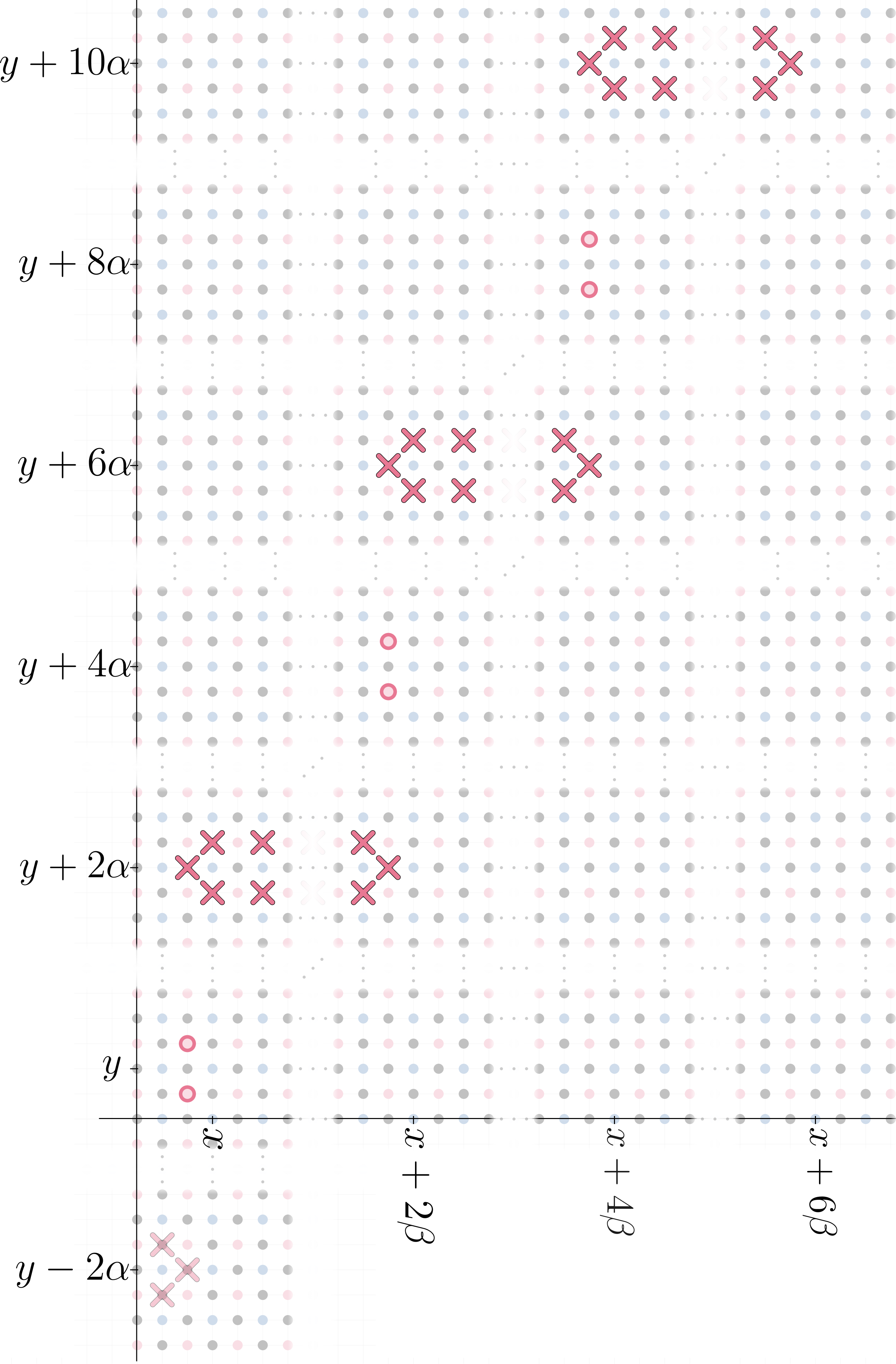}
         \caption{The operator from \Cref{eq:beta_diag_prod_diag_cons_shifts} with $d=12$ that is the product of $X$-type $N^\alpha E^\beta N^\alpha$-stabilisers whose ancilla qubits are highlighted. This operator is also a product of $\beta$ diagonal cycle Pauli-$X$ operators. Note that since $d=12$, the points $(x,y-2\alpha)$ and $(x+6\beta,y+10\alpha)$ are the same on the torus. This is illustrated by the three faintly coloured red crosses at the bottom-left corner, that are the same three red crosses depicted at the top-right corner.}
         \label{fig:beta_diag_prod_diag_cons_shifts}
     \end{subfigure}
     \caption{Products of $X$-type $N^\alpha E^\beta N^\alpha$-stabilisers. In (a) a product of two stabilisers is shown. With (a) we establish a dependency relation between diagonal cycle Pauli-$X$ operators as shown in (b). Pauli $X$ terms are shown as red crosses, and the ancilla qubits are highlighted with red circles.}
     \label{fig:X_NaEbNa_stabilisers_diag_dep}
\end{figure*}

We can establish similar dependencies between diagonal cycle Pauli-$X$ operators. First, directly from the definition in \Cref{eq:x_diag_operator} we obtain
\begin{equation}\label{eq:diag_2beta_east_4alpha_north_shift}
    C^X_{\mathrm{diag}}(x, y) = C^X_{\mathrm{diag}}(x+2\beta, y+4\alpha) \;\;\text{for all}\;\; (x,y)\in\Z^2_Z.
\end{equation}
Second, consider the product of two $X$-type $N^\alpha E^\beta N^\alpha$-stabilisers with associated ancilla qubits $(x-1,y-1), (x-1,y+1) \in \Z^2_X$, which is depicted in \Cref{fig:two_vertical_NaEbNa_stabilisers} and can be expressed in the following way:
\begin{equation}\label{eq:two_vertical_NaEbNa_stabilisers}
    X_{(x-1,y)}X_{(x+2\beta-1,y+4\alpha)}\prod_{p=0}^{\beta-1}B^X_{\mathrm{diag}}(x+2p,y+2\alpha).
\end{equation}
Notice that if we shift this operator with the vector $2\beta\vec{e}+4\alpha\vec{n}$, then in the product of the original and the shifted operator the $X_{(x+2\beta-1,y+4\alpha)}$ term is cancelled. Therefore, if we apply the same shift consecutively $\frac{d}{4}-1$ times and take the product of the obtained $\frac{d}{4}$ operators, we obtain a stabiliser that is of the form
\begin{equation}\label{eq:beta_diag_prod_diag_cons_shifts}
    \prod_{p=0}^{\beta-1}\prod_{m=0}^{\frac{d}{4}-1}B^X_{\mathrm{diag}}(x+2p+2\beta m,y+2\alpha+4\alpha m)
    = \prod_{p=0}^{\beta-1}C^X_{\mathrm{diag}}(x+2p,y+2\alpha),
\end{equation}
see \Cref{fig:beta_diag_prod_diag_cons_shifts}. Next, consider shifting the stabiliser from \Cref{eq:two_horizontal_NaEbNa_stabilisers} (see also \Cref{fig:two_horizontal_NaEbNa_stabilisers}) with the vector $2\beta\vec{e}+4\alpha\vec{n}$ consecutively $\frac{d}{4}-1$ times. The product of these stabilisers is illustrated in \Cref{fig:2alpha_diag_prod_vertical_cons_shifts} and can be expressed as 
\begin{equation}\label{eq:2alpha_diag_prod_vertical_cons_shifts}
    \prod_{q=-\alpha}^{\alpha-1}C^X_{\mathrm{diag}}(x,y+2q).
\end{equation}
\begin{figure*}[!htbp]
     \centering
     \includegraphics[width=0.6\linewidth]{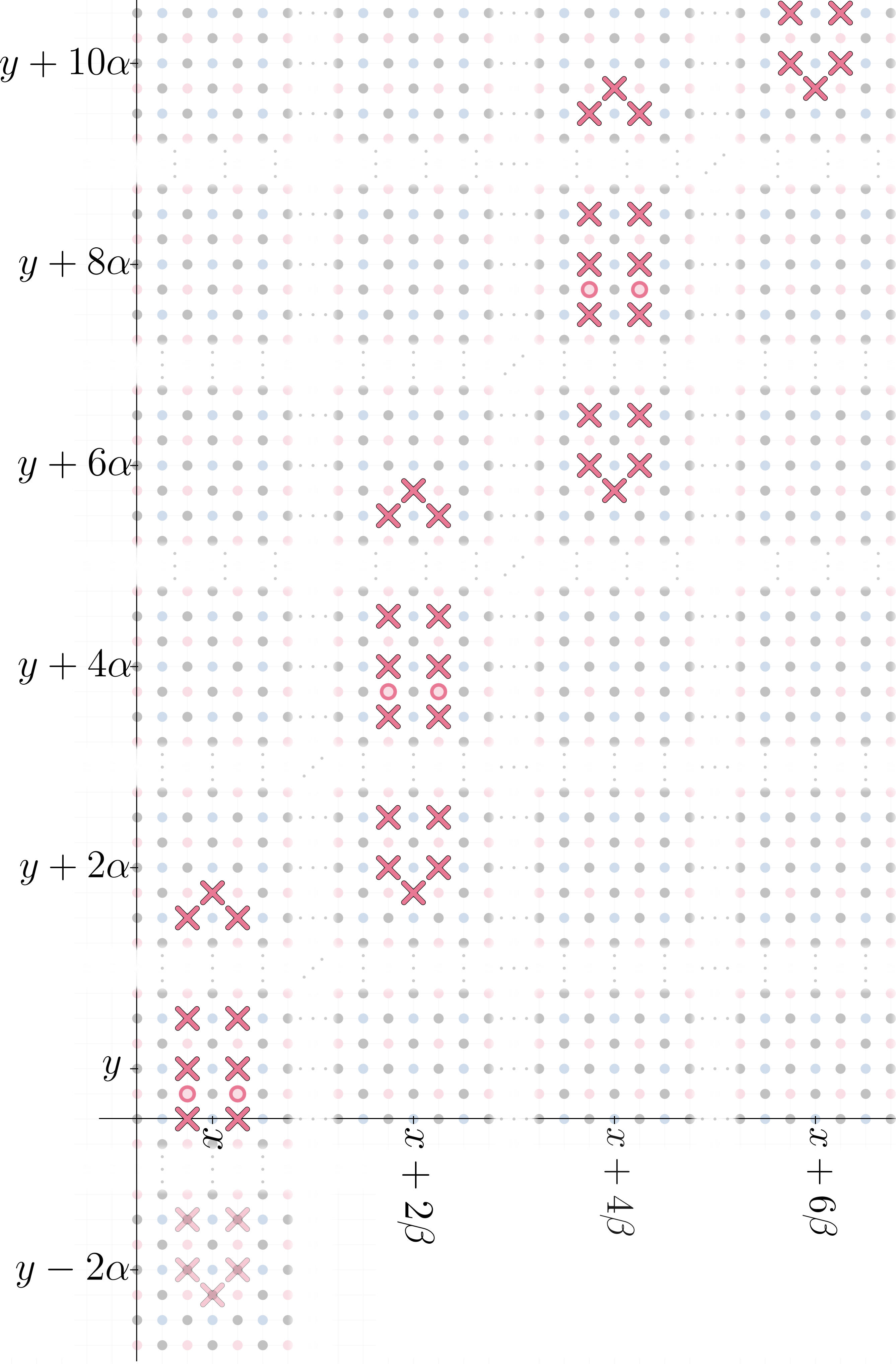}
     \caption{The operator from \Cref{eq:2alpha_diag_prod_vertical_cons_shifts} with $d=12$ that is the product of the $X$-stabilisers whose ancilla qubits are highlighted. This operator is also a product of $2\alpha$ diagonal cycle Pauli-$X$ operators. Note that since $d=12$, the faintly coloured five red crosses at the bottom-left corner are the same as the five red crosses depicted at the top-right corner.}
     \label{fig:2alpha_diag_prod_vertical_cons_shifts}
\end{figure*}
Furthermore, we point out that either of \Cref{eq:beta_diag_prod_diag_cons_shifts,eq:2alpha_diag_prod_vertical_cons_shifts} readily implies that the three diagonal cycle Pauli-$X$ operators $C^X_{\mathrm{diag}}(x, y)$, $C^X_{\mathrm{diag}}(x+2\beta, y)$ and $C^X_{\mathrm{diag}}(x, y+4\alpha)$ are all equivalent up to stabilisers of the code, see \Cref{fig:diag_2beta_east_shift}.
With this we showed that up to stabiliser equivalence we have at most $(2\alpha-1)(\beta-1)$ independent diagonal cycle Pauli-$X$ operators, namely, 
\begin{equation}\label{eq:candidate_diag_X_logicals}
    C^X_{\mathrm{diag}}(1+2p, 2q) \quad (0\leq p \leq \beta-2, 0\leq q \leq 2\alpha-2).
\end{equation}
In the next subsection we show that the cycle Pauli-$X$ operators from \Cref{eq:candidate_hor_X_logicals,eq:candidate_diag_X_logicals} are all logical $X$-operators.

\begin{figure*}
    \centering
    \includegraphics[width=0.6\linewidth]{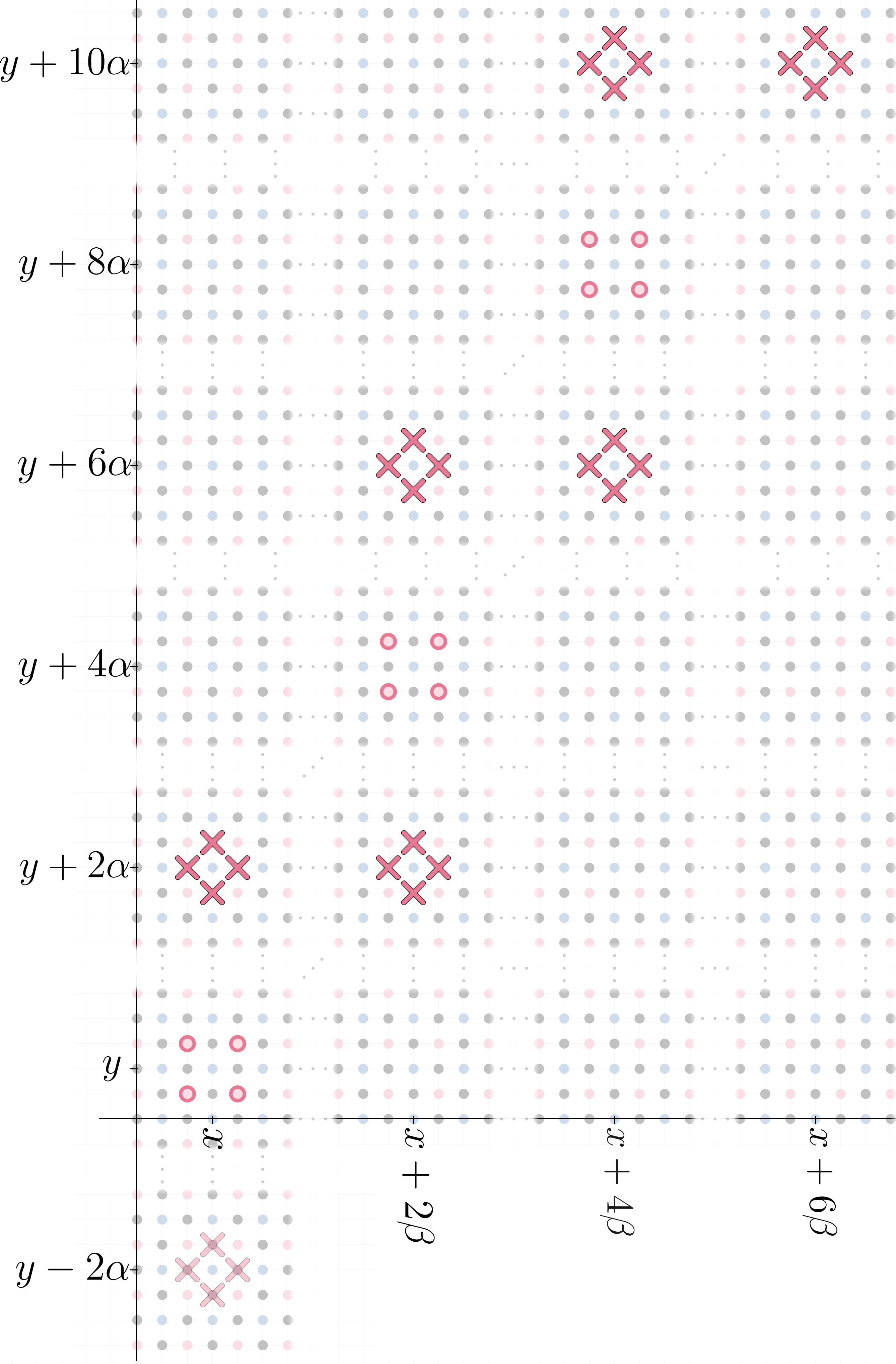}
    \caption{Two diagonal cycle Pauli-$X$ operators that are equivalent up to stabilisers in the $d=12$ case. One is a shifted version of the other with the vector $2\beta\vec{e}$, or alternatively, with the vector $-4\alpha\vec{n}$. The ancilla qubits of the $X$-stabilisers whose product is equal to the product of these two cycle operators are highlighted. Note that since $d=12$, the faintly coloured four red crosses at the bottom-left corner are the same as the four red crosses depicted at the top-right corner.}
    \label{fig:diag_2beta_east_shift}
\end{figure*}


\subsection{Construction of logical Pauli operators}\label{sec:app-logical_operator_construction}

In this subsection, we construct the logical $X$- and $Z$-operators of rotated $N^\alpha E^\beta N^\alpha$-codes from the cycle operators. First, for $k$ pairs of Pauli-$X$ and -$Z$ operators $\tilde{X}_1,\tilde{Z}_1;\allowbreak \tilde{X}_2,\tilde{Z}_2;\dots;\tilde{X}_k,\tilde{Z}_k$ we define the associated anti-commutation matrix $\Omega = [\omega_{i,j}]_{i,j=1}^k$ by
\begin{equation}\label{eq:general_anti-comm_mx}
    \omega_{i,j} = \left\{
    \begin{matrix}
        0, & \text{if } \tilde{X}_i \text{ and } \tilde{Z}_j\text{ commute,}\\
        1, & \text{otherwise.}
    \end{matrix}
    \right.
\end{equation}
where all elements are from the field $\Z_2$. Recall that these pairs of Pauli-$X$ and -$Z$ operators are logical operators of a CSS code if they all commute with the stabilisers and their anti-commutation matrix $\Omega$ is the $k\times k$ identity matrix.

We now construct the codes' logical operators. Recall that both horizontal and diagonal Pauli-$X$ cycle operators were parametrised by $Z$-ancilla qubits $(x,y)$, and therefore that $(x+1,y+1)$ is always an $X$-ancilla qubit. Using the notation $\sigma \in \{\mathrm{hor},\mathrm{diag}\}$, we define the Pauli-$Z$ cycle operator $C^Z_{\sigma}(x+1, y+1)$ by taking the Pauli-$X$ cycle operator $C^X_{\sigma}(x, y)$ and replacing each of its Pauli-$X$ terms $X_{(u,v)}$ with the Pauli-$Z$ operator $Z_{(u+1,v+1)}$. As such, Pauli-$Z$ cycle operators are shifted versions of the Pauli-$X$ cycle operators with the vector $\vec{e}+\vec{n}$. Due to the symmetric structure of the $X$- and $Z$-stabilisers, it is straightforward to see that Pauli-$Z$ cycle operators commute with all stabilisers of the code.

Consider the index-set 
\begin{equation}
    \mathcal{I}:=\{\mathrm{hor},\mathrm{diag}\}\times\{0,1,...,\beta-2\}\times\{0,1,...,2\alpha-2\},
\end{equation}
and associate with each element $(\sigma,p,q)\in\mathcal{I}$ the Pauli-$X$ cycle operator $C^X_{\sigma}(1+2p, 2q)$ and the Pauli-$Z$ cycle operator $C^Z_{\sigma}(2+2p, 1+2q)$. The anti-commutation matrix 
\begin{equation}
    \Omega = [\omega_{(\sigma,p,q),(\sigma',p',q')}]_{(\sigma,p,q),(\sigma',p',q')\in\mathcal{I}}
\end{equation}
associated with these operators is then the $(2(2\alpha-1)(\beta-1))\times(2(2\alpha-1)(\beta-1))$ matrix defined by
\begin{equation}
    \omega_{(\sigma,p,q),(\sigma',p',q')} = \left\{
    \begin{matrix}
        0, & \text{if } C^X_{\sigma}(1+2p, 2q) \text{ and } C^Z_{\sigma'}(2+2p', 1+2q') \text{ commute,}\\
        1, & \text{otherwise.}
    \end{matrix}
    \right.
\end{equation}
This is the same matrix defined in \Cref{eq:general_anti-comm_mx}, except with elements indexed by $\mathcal{I}$. We will additionally denote by $\Omega_{\sigma,\sigma'}$ the $(2\alpha-1)(\beta-1)\times(2\alpha-1)(\beta-1)$ block of $\Omega$ where $\sigma$ and $\sigma'$ are fixed. In what follows, we prove that $\Omega$ is invertible, and then show that as a result the logical $X$- and $Z$-operators can be expressed as products of cycle operators.

\begin{proposition}\label{prop:cycle_ops_anti-comm_invertible}
    The anti-commutation matrix $\Omega$ associated with the cycle operators is invertible.
\end{proposition}

\begin{proof}
    First, since horizontal Pauli-$X$ cycle operators are supported on $L$-data qubits and horizontal Pauli-$Z$ cycle operators are supported on $R$-data qubits, they always commute and therefore 
    \begin{equation}
        \Omega_{\mathrm{hor}, \mathrm{hor}}=0.
    \end{equation}
    Second, recall that diagonal cycle Pauli-$X$ operators are composed of $B^X_{\mathrm{diag}}$ operators that are each supported on the four nearest-neighbours of a $Z$-ancilla qubit. The diagonal cycle Pauli-$Z$ operators have a similar structure, except that the block operators are each supported on the four nearest-neighbours of an $X$-ancilla qubit. Hence the intersection of the supports of an $X$- and a $Z$-type diagonal block operator contains either zero or two qubits, and as such they always commute. Therefore,
    \begin{equation}
        \Omega_{\mathrm{diag}, \mathrm{diag}}=0.
    \end{equation}
    This implies that $\Omega$ is invertible if and only if the two off-diagonal blocks $\Omega_{\mathrm{hor}, \mathrm{diag}}$ and $\Omega_{\mathrm{diag}, \mathrm{hor}}$ are both invertible.
    
    Next, in order to prove that $\Omega_{\mathrm{hor},\mathrm{diag}}$ is invertible, we investigate when a horizontal Pauli-$X$ cycle operator $C^X_{\mathrm{hor}}(1+2p, 2q)$ commutes with a diagonal Pauli-$Z$ cycle operator $C^Z_{\mathrm{diag}}(2+2p', 1+2q')$ for $0\leq p, p' \leq \beta-2,\; 0\leq q, q' \leq 2\alpha-2$. These operators can be expressed as
    \begin{equation}\label{eq:x_hor_for_anticommute}
        C^X_{\mathrm{hor}}(1+2p, 2q) = X_{(2p, 2q)} X_{(2+2p, 2q)} \cdot \left(\prod_{m=1}^{\frac{d}{2}-1} B^X_{\mathrm{hor}}(1+2p+2\beta m,2q)\right)
    \end{equation}
    and
    \begin{equation}\label{eq:x_diag_for_anticommute}
        \begin{split}
            C^Z_{\mathrm{diag}}(2+2p', 1+2q') = Z_{(2+2p',2q')} & Z_{(2+2p',2+2q')} Z_{(1+2p',1+2q')}Z_{(3+2p',1+2q')} \\ & \cdot \Bigg(\prod_{m=1}^{\frac{d}{4}-1}B^Z_{\mathrm{diag}}(2+2p'+2\beta m,1+2q'+4\alpha m)\Bigg).
        \end{split}
    \end{equation}
    Clearly, the two cycle operators anti-commute if and only if \linebreak $X_{(2p, 2q)} X_{(2+2p, 2q)}$ and $Z_{(2+2p',2q')}Z_{(2+2p',2+2q')}$ do. Indeed, the operator in \Cref{eq:x_hor_for_anticommute} is supported on $L$-data qubits whose $y$-coordinates satisfy $0\leq 2q \leq 4\alpha-4$, whereas in \Cref{eq:x_diag_for_anticommute} this does not hold for any of the Pauli-$Z$ terms except for the first two. Furthermore, the $x$-coordinates of these two Pauli-$Z$ terms satisfy $2\leq 2+2p'\leq 2\beta-2$, which can only hold for the first two Pauli-$X$ terms of the operator in \Cref{eq:x_hor_for_anticommute}. Therefore, we arrive at
    \begin{equation}\label{eq:offdiagonal_anticommutation}
    \begin{aligned}
    \omega_{(\mathrm{hor},p,q),(\mathrm{diag},p',q')}
    &=
    \begin{cases}
    1, & \text{if } p-p',\, q-q' \in \{0,1\},\\
    0, & \text{otherwise},
    \end{cases}
    \\
    &=
    \begin{cases}
    1, & \text{if } p-p' \in \{0,1\},\\
    0, & \text{otherwise},
    \end{cases}
    \times
    \begin{cases}
    1, & \text{if } q-q' \in \{0,1\},\\
    0, & \text{otherwise}.
    \end{cases}
    \end{aligned}
    \end{equation}
    Recall the definition of the Kronecker delta symbol: $\delta_{i,j} = 1$ if $i= j$, otherwise $\delta_{i,j} = 0$. With this notation the $n\times n$ identity matrix is expressed as $I_{n} = [\delta_{i,j}]_{i,j=0}^{n-1}$. We also define the $n\times n$ shift matrix as $L_{n} = [\delta_{i,j+1}]_{i,j=0}^{n-1}$ which has zero elements everywhere except just below the diagonal where its elements are all $1$. Notice that by \Cref{eq:offdiagonal_anticommutation} we have
    \begin{equation}
        \Omega_{\mathrm{hor},\mathrm{diag}}= (I_{\beta-1}+L_{\beta-1})\otimes(I_{2\alpha-1}+L_{2\alpha-1}).
    \end{equation}
    Since $I_{n}+L_{n}$ has determinant $1$ and hence is invertible, this tensor product is also invertible. One can prove by an analogous argument that $\Omega_{\mathrm{diag},\mathrm{hor}}$ is also invertible, which concludes that the anti-commutation matrix $\Omega$ is indeed invertible.
\end{proof}

Next, we prove that the logical operators can be expressed as products of the cycle operators. For this it is more convenient to consider a more general scenario.

\begin{proposition}\label{prop:anti-comm_mx_invertible}
    Assume we have $k$ pairs of Pauli-$X$ and -$Z$ operators $\tilde{X}_1,\tilde{Z}_1;\allowbreak \tilde{X}_2,\tilde{Z}_2;\allowbreak\dots;\tilde{X}_k,\tilde{Z}_k$ that are all supported on the data qubits of a CSS code and all commute with its stabilisers. We further assume that their associated anti-commutation matrix $\Omega = [\omega_{i,j}]_{i,j=1}^k$ is invertible. Then it is possible to define $k$ logical operators $\hat{X}_1,\hat{Z}_1;\allowbreak \hat{X}_2,\hat{Z}_2;\dots;\allowbreak\hat{X}_k,\hat{Z}_k$ of the CSS code with the following form:
    \begin{itemize}
        \item $\hat{X}_j = \tilde{X}_j$ holds for all $j$,
        \item each $\hat{Z}_j$ is a product of some $\tilde{Z}_m$ Pauli-$Z$ operators.
    \end{itemize}
\end{proposition}

\begin{proof}
    We define $\hat{X}_j := \tilde{X}_j$ for all $j$. For each $j$ we consider an index-set $\Gamma_j \subset\{1,2,\dots,k\}$ and define $\hat{Z}_j := \prod_{m\in\Gamma_j} \tilde{Z}_m$. It is straightforward that for all $i,j$ we have
    \begin{equation}\label{eq:ij_anti-comm}
    \begin{aligned}
        \hat{X}_i \hat{Z}_j 
        &= \hat{X}_i \prod_{m\in\Gamma_j} \tilde{Z}_m \\
        &= (-1)^{\sum_{m\in\Gamma_j}\omega_{i,m}} \prod_{m\in\Gamma_j} \tilde{Z}_m \hat{X}_i \\
        &= (-1)^{\sum_{m=1}^k\omega_{i,m} \gamma_{m,j}} \hat{Z}_j \hat{X}_i
    \end{aligned}
    \end{equation}
    where $\gamma_{m,j} = 1$ if $m\in\Gamma_j$, otherwise $\gamma_{m,j} = 0$. Note the sum can be taken modulo $2$ in \Cref{eq:ij_anti-comm}. As such, if we define the matrix $\Gamma = [\gamma_{m,j}]_{m,j=1}^k$, then the exponent on the last line is exactly the $(i,j)$ element of the matrix-product $\Omega\Gamma$. Therefore, choosing $\Gamma=\Omega^{-1}$ gives $\hat{X}_i \hat{Z}_j = (-1)^{\delta_{ij}} \hat{Z}_j \hat{X}_i$, which means the operator pairs $(\hat{X}_j, \hat{Z}_j)$ together define $k$ logical qubits. This concludes the proof.
\end{proof}

Since the cycle operators satisfy the conditions of \Cref{prop:anti-comm_mx_invertible}, this proves that the rotated $N^\alpha E^\beta N^\alpha$-codes have at least $2(2\alpha-1)(\beta-1)$ logical qubits, and that the logical operators can be expressed as products of cycle operators. Since the cycle operators have all weight exactly $d$, we also conclude that the code distance is at most $d$. As was mentioned in the main text, we verified the code distance to be exactly $d$ using integer programming up to $d=12$.


\section{Simulation details and additional results}
\label{sec:app-sim_details_additional_results}
In this appendix, we provide further details on the simulation methodology and present additional numerical results for the $NE^3N$-, $N^2E^2N^2$- and $N^2E^4N^2$-directional codes that were not discussed in detail in \Cref{sec:simulations}. Additionally, we compare the TC's logical error rate using the Tesseract decoder's short and long beam settings. All circuits for the directional codes that were benchmarked, their parity check matrices, and the collected numerical data are available in \cite{our_circuits}.

\subsection{Simulation details}
\label{sec:app-simulation_details}
In this subsection, we provide further details on our simulation methodology. All the investigated codes were benchmarked using an $X$-memory circuit with $d$ QEC rounds, where $d$ is the code distance. More precisely, for each code we prepared the data qubits in the $|+\rangle$ state, then performed $d$ rounds of syndrome extraction, and finally measured out the data qubits in the $X$-basis, forming logical observables from the logical $X$-operators. Note that it is well-known that the $X$- and $Z$-memory performances are approximately the same for the RTC, TC, RPSC and BB codes. As for the directional codes, the same holds due to the following symmetry. Take the stabilisers and their scheduling, together with the logical operators of the code, translate them by the vector $(1,1)$, and change the Pauli-$X$ (-$Z$) terms to Pauli-$Z$ (-$X$) terms. This way we get back the same CSS code, and importantly, the scheduling of each stabiliser remains unchanged.

The benchmarked directional code circuits were compiled to the following gates: $Z$-basis reset (R), $Z$-basis measurement (M), the iSWAP entangling gate, the Hadamard gate (H), the phase gate (S) and the $\frac{\pi}{2}$-$X$-rotation ($\sqrt{X}$). These circuits can all be executed under degree-$4$ (\Cref{fig:square_grid}) and some of them even under degree-$3$ connectivity (\Cref{fig:hex_grid}), see \Cref{tab:directions_and_layouts}. As for the RTC, TC, RPSC and BB code circuits, we used the same gates, except instead of iSWAP we used the CZ gate. The RTC, TC and RPSC circuits were constructed in the standard way, i.e. under square-grid connectivity, although we note that circuits for these codes can also be constructed on the hexagonal-grid with CZ gates, and with iSWAP gates both on square- and hex-grid. These alternative circuits have similar QEC performance, see \cite{MBG}.

Noise was added to the compiled circuits according to the standard superconducting-inspired circuit-level Pauli noise model from \cite{si1000-1, si1000-2}, called ``SI-$1000$''. This noise model is parametrised by $p$, the physical error rate, which we varied as $p=10^{-3+\frac{j}{5}}$ for $0\leq j\leq 5$. We point out that noise was added for all gates, including the state preparations at the beginning and the measurements at the end of the circuits. We note that the SI-$1000$ noise model originally does not include the iSWAP gate, so we added the same noise after iSWAP gates as after CZ gates. Additionally, it is important to keep in mind that the SI-$1000$ noise model may be optimistic for the BB codes. These codes indeed require two additional and non-local connections at each qubit location, which makes the physical device noisier \cite{alec-dynamic-demonstration,qldpc-demonstration} but is not penalised by the SI-$1000$ noise model (see \cite{tesseractdecoder} for a proposed noise model that does so).

We used ``stim'' \cite{stim} to sample our noisy circuits and the Tesseract decoder \cite{tesseractdecoder} with the short beam setting to decode them. For each data point, we first sampled until either $10^6$ shots were performed or $10^4$ decoding failures were reached. After this we sampled:
\begin{itemize}
    \item an additional $10^7$ shots for the distance $8$ rectangular $N^2E^2N^2$-code at $p=10^{-3}$,
    \item an additional $10^8$ shots for the RTC, TC and RPSC at $p=10^{-3}$.
\end{itemize}
The sampling and decoding were handled by the ``sinter'' decoding sampler, a stim extension which efficiently distributes the tasks across multiple CPU cores. All of our simulations, including the initial explorations and parallelogram searches, used a total of $\approx 45$ CPU-core-years. 

All plots we present in the paper show the sampled logical error rate. More precisely, after sampling and decoding a quantum memory circuit with $d$ QEC rounds for a total number of $N$ shots, the ``sampled logical failure probability'' is $P_L(d) = \frac{f}{N}$ where $f$ is the number of samples where the decoder failed. We then calculate the Clopper--Pearson interval $(P_{min}^{\alpha},P_{max}^{\alpha})$ defined by the following equations:
\begin{equation}\label{eq:P_min}
    P^\alpha_{min} = B\left(\frac{\alpha}{2}; f, N - f + 1\right)
\end{equation}
and
\begin{equation}\label{eq:P_max}
    P^\alpha_{max} = B\left(1 - \frac{\alpha}{2}; f + 1, N - f\right),
\end{equation}
where $0<\alpha<1$ and $B(p;v,w)$ is the $p$th quantile from a beta distribution with shape parameters $v$ and $w$. Specifically, we substitute $\alpha=0.05$, which provides an interval where the (real) logical failure probability lies with $95\%$ confidence based on its sampled value $P_L(d)$, which we use as the error bars around $P_L(d)$. Next, we calculate the ``sampled logical error rate'' (per round) using the following standard formula
\begin{equation}\label{eq:LER}
    p_L = 1-\sqrt[d]{1-P_L(d)} \approx \frac{P_L(d)}{d},
\end{equation}
and we obtain error bars around each by applying the same formula to $P_{min}^{0.05}$ and $P_{max}^{0.05}$. In each of our plots a marker shows the value of $p_L$. The error bars are displayed as shaded areas for plots where $p$ varies on the $x$-axis, and otherwise as vertical lines. Note that in some cases we consider $m$ copies of a code, e.g. in \Cref{fig:N2E3N2_vs_RTC} we considered $m=6$ copies of RTCs in order to match the number of logical operators of the $N^2E^3N^2$-codes. In such a case, the sampled logical failure probability $P_L(d,m)$ of $m$ copies of the code was obtained using the formula $P_L(d,m) = 1-(1-P_L(d))^m$.

We finish this section by explaining how we obtained the values in \Cref{tab:slope_ratios} and \Cref{fig:N2E3N2_fixed_p} (see also \Cref{fig:NE3N_fixed_p,fig:N2E2N2_fixed_p,fig:N2E4N2_fixed_p}). Note that for the RTC, TC, and RPSC it is well-known that the logical error rates decrease exponentially with the distance. Equivalently, they decrease exponentially with $\sqrt{n_{ph}}$, where $n_{ph}$ is the total number of physical qubits used by the code. Based on the calculated distances for directional codes, we conjecture this to be the same for the rectangular and rotated families. Therefore, on the square-root--logarithmic scaled plots for a fixed value of $p$, such as \Cref{fig:N2E3N2_fixed_p}, we may fit a line, and extrapolate to larger distances. Assume the following line fits for a family of directional codes with $k$ logical qubits, and $\frac{k}{2}$ copies of the RTC, respectively:
\begin{equation}\label{eq:line-fit_DC}
    \log p_L^{dir} = s_{dir}\sqrt{n_{ph}^{dir}}+c_{dir}
\end{equation}
and
\begin{equation}\label{eq:line-fit_RTC}
    \log p_L^{rtc} = s_{rtc}\sqrt{n_{ph}^{rtc}}+c_{rtc}.
\end{equation}
Then, using these formulae, solving the equation $p_L^{dir} = p_L^{rtc}$ gives that 
\begin{equation}
    n_{ph}^{dir} \approx \frac{(s_{rtc})^2}{(s_{dir})^2} n_{ph}^{rtc}.
\end{equation}
Therefore, to achieve the same logical error rate as the RTC, the given family of directional codes requires $\approx \frac{(s_{rtc})^2}{(s_{dir})^2}$-times the number of physical qubits. These are the numbers that are shown in \Cref{tab:slope_ratios}. The sub- and super-scripts of each number describe a confidence interval of at least $98\%$, which we obtain in the following way. For each set of three $p_L$'s, on which we intend to fit a line, we sample around the corresponding $P_L(d)$ values using the probability distribution defined by \Cref{eq:P_min,eq:P_max}. We then transform these numbers using \Cref{eq:LER} and calculate the slope of the line fitted on them. We repeat this $100,000$-times, giving $100,000$ values for the slope of each code family. Sorting the values in ascending order and taking the $500$th and $99,500$th elements provides a $99\%$ confidence interval for each slope's (real) value. Denote the two endpoints of this interval by $s_{rtc}^{min}$ and $s_{rtc}^{max}$ for the RTC, and by $s_{dir}^{min}$ and $s_{dir}^{max}$ for the given family of directional codes. Note that the value of the slopes are negative, and as such, the squared ratios of the (real) slopes lie in the following interval
\begin{equation}
    \left(\frac{(s_{rtc}^{max})^2}{(s_{dir}^{min})^2},
    \frac{(s_{rtc}^{min})^2}{(s_{dir}^{max})^2}\right)
\end{equation}
with at least $98\%$ (which is $< (99\%)^2$) confidence.

Next, we present the numerical results for the directional codes that were only briefly considered in \Cref{sec:simulations}.

\subsection{Additional numerical results}
\label{sec:app-additional_results}

\begin{table}[ht]
\centering
\begin{tabular}{|Sc|Sc|Sc|Sc|Sc|Sc|Sc|} 
 \hline
 Code name & $[\![n,k,d]\!]$ & $d_{circ}$ & $\vec{v}_1$ & $\vec{v}_2$ & $\frac{k d^2}{n}$\\ [0.5ex]
 \hline
 \hline
 $NE^3N$ rectangular & $[\![36,4,4]\!]$ & $4$ & $(18,0)$ & $(0,4)$ & $1+\frac{7}{9}$ \\
 \hline
 $NE^3N$ filler & $[\![72,4,6]\!]$ & $5$ & $(18,0)$ & $(0,8)$ & $2$ \\
 \hline
 $NE^3N$ rectangular & $[\![144,4,8]\!]$ & $8$ & $(36,0)$ & $(0,8)$ & $1+\frac{7}{9}$ \\
 \hline
 \hline
 $N^2E^2N^2$ rotated & $[\![32,6,4]\!]$ & $4$ & $(-8,0)$ & $(-4,8)$ & $3$ \\
 \hline
 $N^2E^2N^2$ rectangular & $[\![48,6,4]\!]$ & $4$ & $(12,0)$ & $(0,8)$ & $2$ \\
 \hline 
 $N^2E^2N^2$ filler & $[\![96,6,6]\!]$ & $6$ & $(12,0)$ & $(0,16)$ & $2+\frac{1}{4}$ \\
 \hline 
 $N^2E^2N^2$ rotated & $[\![128,6,8]\!]$ & $8$ & $(-16,0)$ & $(-8,16)$ & $3$ \\
 \hline
 $N^2E^2N^2$ rectangular & $[\![192,6,8]\!]$ & $8$ & $(24,0)$ & $(0,16)$ & $2$ \\
 \hline
 \hline
 $N^2E^3N^2$ rotated & $[\![48,12,4]\!]$ & $4$ & $(12,0)$ & $(6,8)$ & $4$ \\
 \hline
 $N^2E^3N^2$ rectangular & $[\![72,12,4]\!]$ & $4$ & $(18,0)$ & $(0, 8)$ & $2+\frac{2}{3}$ \\
 \hline
 $N^2E^3N^2$ filler & $[\![144,12,6]\!]$ & $6$ & $(18,0)$ & $(0,16)$ & $3$ \\
 \hline
 $N^2E^3N^2$ rotated & $[\![192,12,8]\!]$ & $8$ & $(24,0)$ & $(12,16)$ & $4$ \\
 \hline
 $N^2E^3N^2$ rectangular & $[\![288,12,8]\!]$ & $8$ & $(36,0)$ & $(0,16)$ & $2+\frac{2}{3}$ \\
 \hline
 \hline
 $N^2E^4N^2$ rotated & $[\![64,18,4]\!]$ & $4$ & $(-8,8)$ & $(8,8)$ & $4+\frac{1}{2}$ \\
 \hline
 $N^2E^4N^2$ filler & $[\![192,18,6]\!]$ & $6$ & $(24,0)$ & $(0,16)$ & $3+\frac{3}{8}$ \\
 \hline
 $N^2E^4N^2$ rotated & $[\![256,18,8]\!]$ & $\leq 8$ & $(-16,16)$ & $(16,16)$ & $4+\frac{1}{2}$ \\
 \hline
\end{tabular}
\caption{A list of all benchmarked directional codes. Also listed are their code parameters $[\![n,k,d]\!]$, their circuit-level distance $d_{circ}$, their code efficiency $\frac{k d^2}{n}$, and the vectors $\vec{v}_1$ and $\vec{v}_2$ that specify the (possibly twisted) torus on which they are defined. All of these are compared against the RTC, TC and RPSC in \Cref{tab:slope_ratios}, and the $N^2E^3N^2$-codes were examined in detail in \Cref{sec:simulations}. In this section we present the additional plots for the rest of the codes. All distances were calculated using integer programming, except for the circuit-level distance of the distance $8$ rotated $N^2E^4N^2$-code where we only give an estimation. Among calculated distances, the $NE^3N$ filler is the only code instance in the table for which $d_{circ}$ differs from $d$. Note that for the rotated directional codes the vectors $\vec{v}_1$ and $\vec{v}_2$ given here differ from the ones specified in their definition in \Cref{sec:code-construct}. They always define an equivalent parallelogram and hence the same twisted torus and CSS code, see \Cref{prop:equiv_par}. The vectors shown here are the ones that we used for constructing the circuits and representing the numerical data with in \cite{our_circuits}.}
\label{tab:code_list}
\end{table}

Recall that in \Cref{sec:code-construct,sec:simulations} we mainly focused on $N^2E^3N^2$-codes, and presented detailed plots only for them. The other directional codes, namely, the $NE^3N$-codes encoding four logical qubits, the $N^2E^2N^2$-codes encoding six logical qubits, and the $N^2E^4N^2$-codes encoding eighteen logical qubits, were only tangentially discussed. In this section we present the rest of the detailed plots for these additional directional codes. In \Cref{tab:code_list} we listed all directional codes that we benchmarked, and recall that \Cref{tab:slope_ratios} compares all of them to the RTC, TC and RPSC at $p=10^{-3}$ physical error rate. We note that for all directional codes in \Cref{tab:code_list} we calculated the code distance and also the circuit-level distance using integer programming, except for the largest $N^2E^4N^2$-code. We did this using quantum memory circuits with one QEC round. We point out that in directional code circuits we reverse every other QEC round, therefore they fall under the so-called ``two-round morphing circuits'' category from \cite{Mac-morphing-2}. As such \cite[Proposition 1]{Mac-morphing-2} guarantees that the circuit-level distance obtained this way equals the circuit-level distance we would obtain with $d$ QEC rounds.

\begin{figure}[!htb]
     \centering
     \begin{subfigure}{0.37\textwidth}
         \centering
         \includegraphics[width=1.0\linewidth]{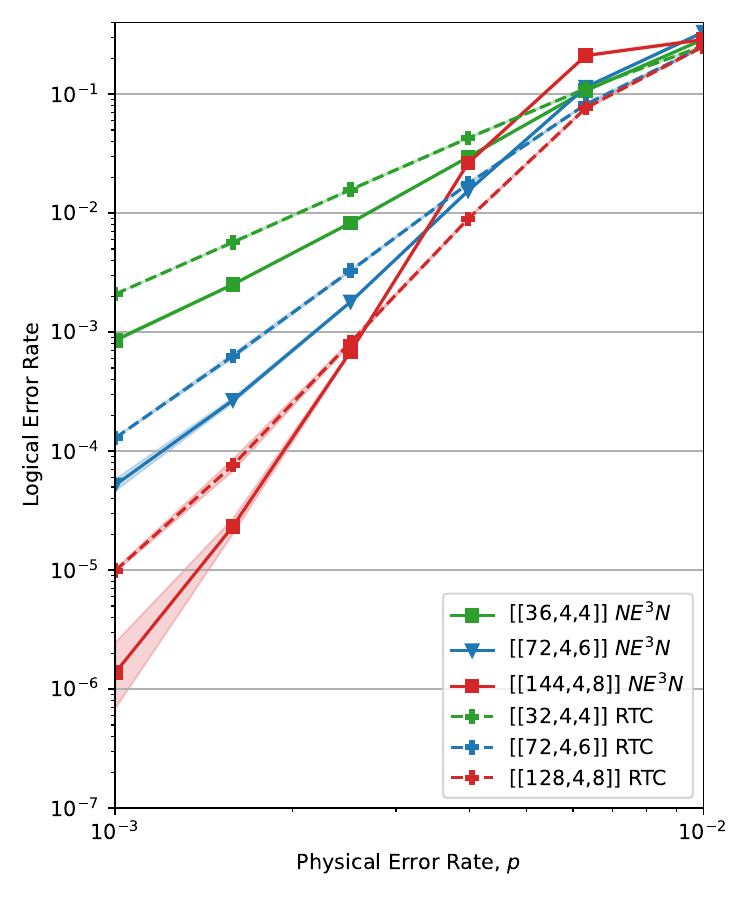}
         \caption{rectangular $NE^3N$ vs RTC}
         \label{fig:NE3N_rectangular_vs_RTC}
     \end{subfigure}
     \begin{subfigure}{0.59\textwidth}
         \centering
         \includegraphics[width=1.0\linewidth]{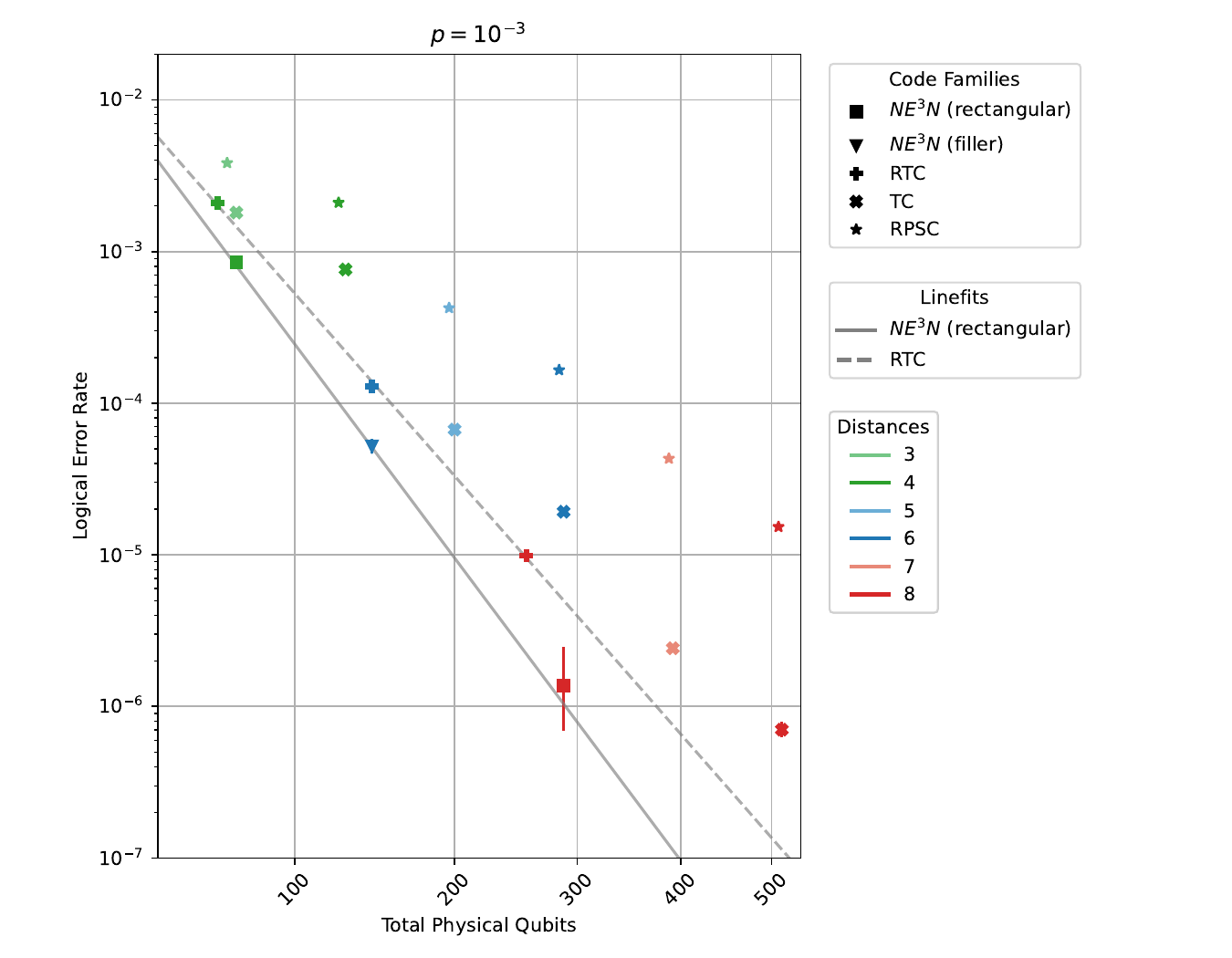}
         \caption{$NE^3N$ vs three types of surface codes at $p=10^{-3}$}
         \label{fig:NE3N_fixed_p}
     \end{subfigure}
     \caption{(a) Comparison of the rectangular ($d=4,8$) and the filler ($d=6$) $NE^3N$-codes against two copies of the RTC, giving four logical qubits for both. (b) Comparison of the three $NE^3N$-codes against the RTC, TC and RPSC where we fix $p=10^{-3}$.}
     \label{fig:NE3N_Results}
\end{figure}

In \Cref{fig:NE3N_rectangular_vs_RTC,fig:N2E2N2_rotated_vs_RTC,fig:N2E2N2_rectangular_vs_RTC,fig:N2E4N2_rotated_vs_RTC} we compare the relevant directional code families against the RTC, in a similar way as was done in \Cref{fig:N2E3N2_rotated_vs_RTC,fig:N2E3N2_rectangular_vs_RTC}. Furthermore, in \Cref{fig:NE3N_fixed_p,fig:N2E2N2_fixed_p,fig:N2E4N2_fixed_p} we make a similar comparison to that in \Cref{fig:N2E3N2_fixed_p}. We use the same convention for markers, line-styles, and colouring. The number of copies for the RTC is always taken so that it matches the number of logical qubits of the directional codes.

\begin{figure}[!htb]
     \centering
     \begin{subfigure}{0.45\textwidth}
         \centering
         \includegraphics[width=1.0\linewidth]{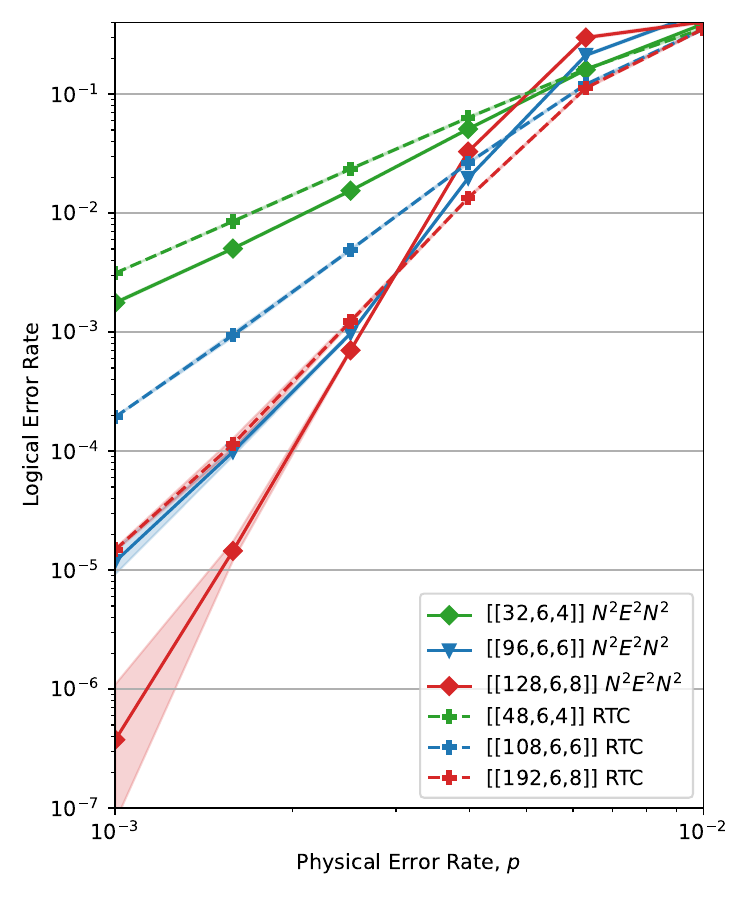}
         \caption{rotated $N^2E^2N^2$ vs RTC}
         \label{fig:N2E2N2_rotated_vs_RTC}
     \end{subfigure}
     \begin{subfigure}{0.45\textwidth}
         \centering
         \includegraphics[width=1.0\linewidth]{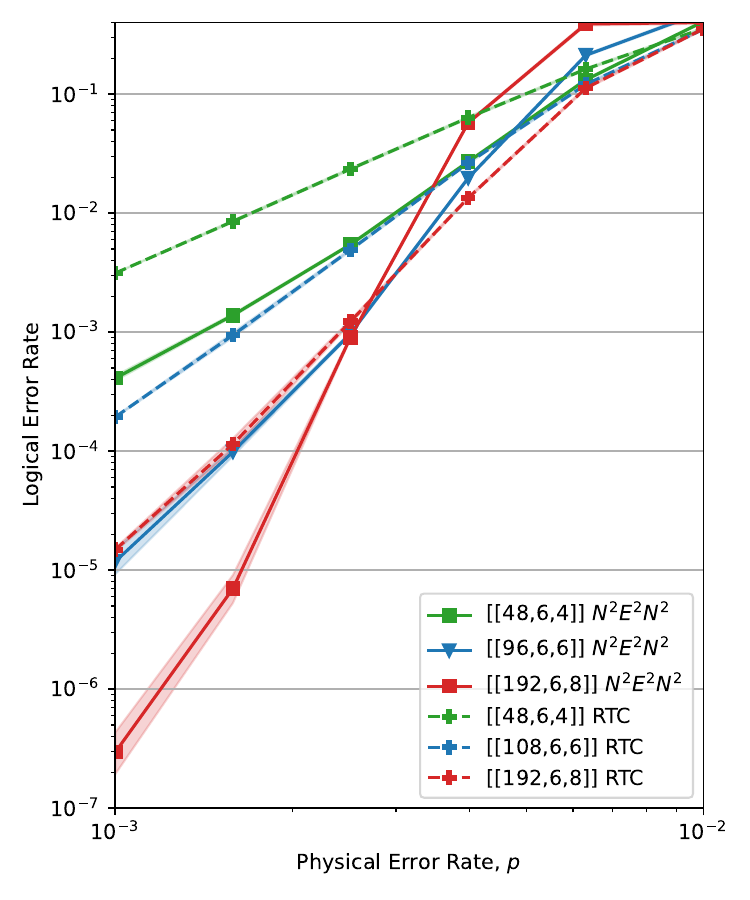}
         \caption{rectangular $N^2E^2N^2$ vs RTC}
         \label{fig:N2E2N2_rectangular_vs_RTC}
     \end{subfigure}
     \caption{Comparison of $N^2E^2N^2$-codes against three copies of the RTC, giving six logical qubits for both. (a) The rotated ($d=4,8$) and the filler ($d=6$) $N^2E^2N^2$-codes compared against the RTC. (b) The rectangular ($d=4,8$) and the filler ($d=6$) $N^2E^2N^2$-codes compared against the RTC.}
     \label{fig:N2E2N2_vs_RTC}
\end{figure}

\begin{figure}[!htb]
    \centering
    \includegraphics[width=0.65\linewidth]{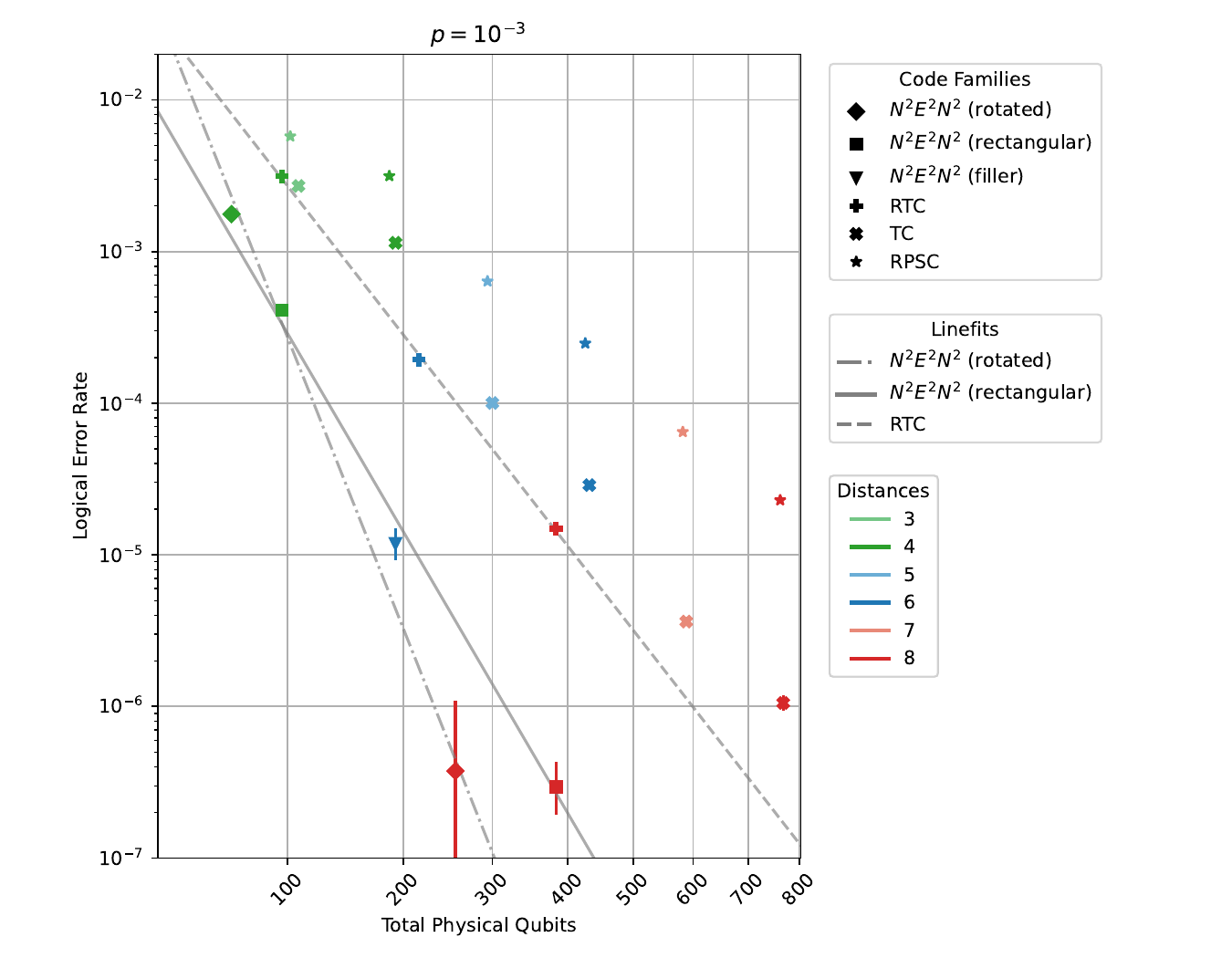}
    \caption{Comparison of all five $N^2E^2N^2$-codes against the RTC, TC and RPSC where we fix $p=10^{-3}$.}
    \label{fig:N2E2N2_fixed_p}
\end{figure}

\begin{figure}[!htb]
     \centering
     \begin{subfigure}{0.37\textwidth}
         \centering
         \includegraphics[width=1.0\linewidth]{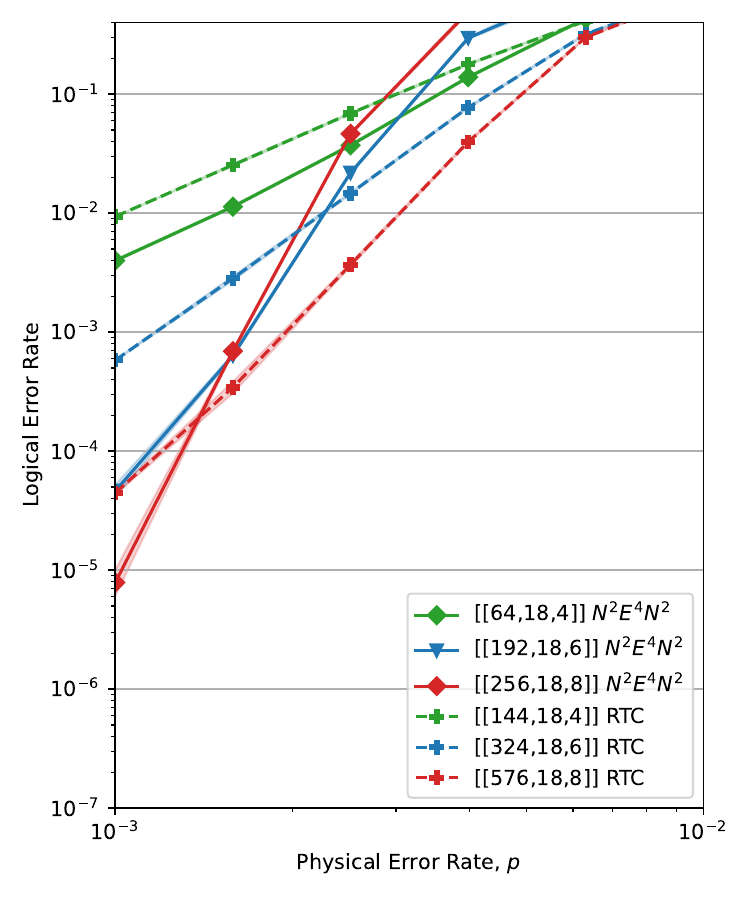}
         \caption{rotated $N^2E^4N^2$ vs RTC}
         \label{fig:N2E4N2_rotated_vs_RTC}
     \end{subfigure}
     \begin{subfigure}{0.59\textwidth}
         \centering
         \includegraphics[width=1.0\linewidth]{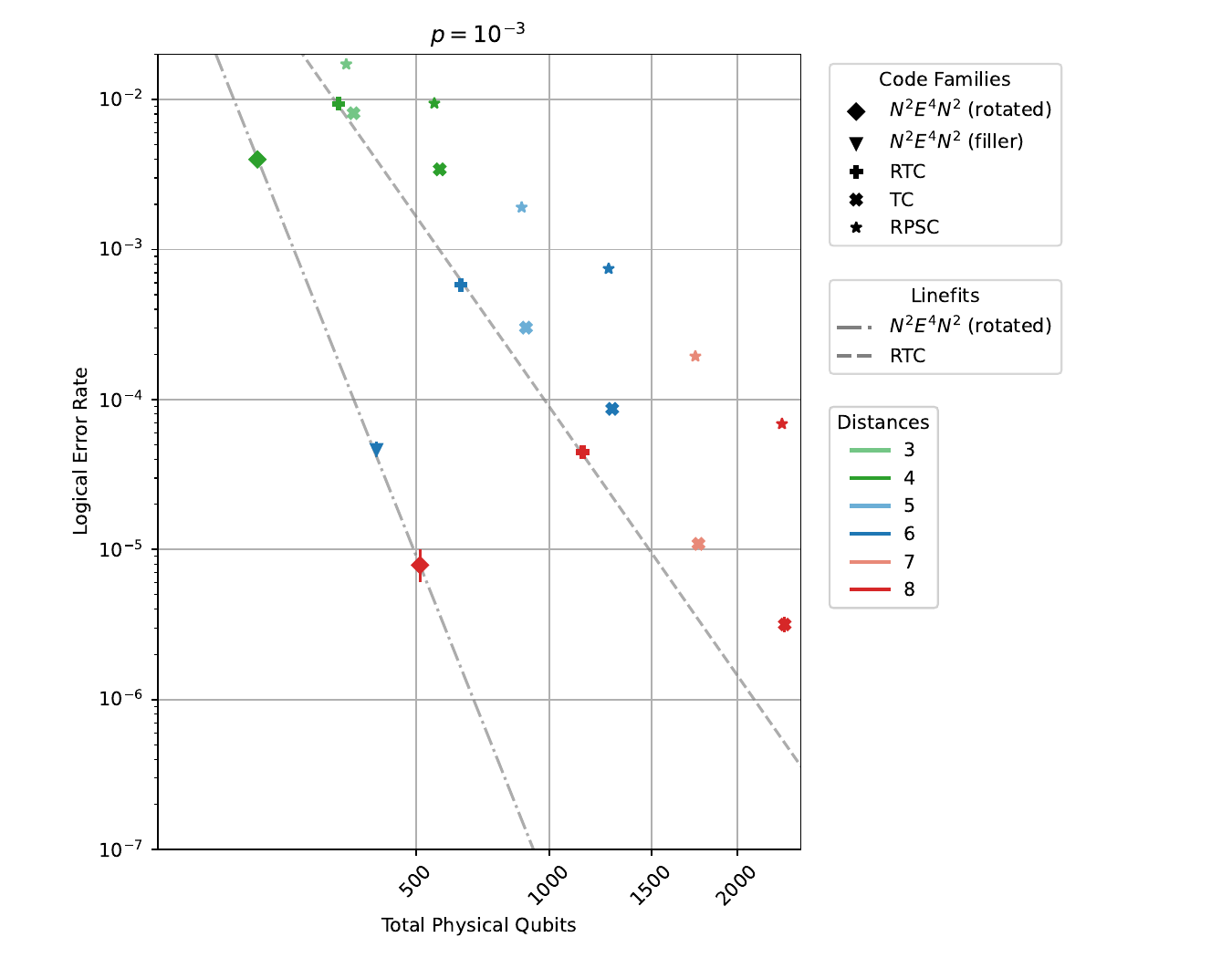}
         \caption{$N^2E^4N^2$ vs three types of surface codes at $p=10^{-3}$}
         \label{fig:N2E4N2_fixed_p}
     \end{subfigure}
     \caption{(a) Comparison of the rotated ($d=4,8$) and the filler ($d=6$) $N^2E^4N^2$-codes against nine copies of the RTC, giving eighteen logical qubits for both. (b) Comparison of the three $N^2E^4N^2$-codes against the RTC, TC and RPSC where we fix $p=10^{-3}$.}
     \label{fig:N2E4N2_Results}
\end{figure}

\subsection{The impact of the beam setting in the Tesseract decoder for the Toric Code}
\label{sec:app-short_beam_big_circuits}
\begin{figure}[!htb]
    \centering
    \includegraphics[width=0.75\linewidth]{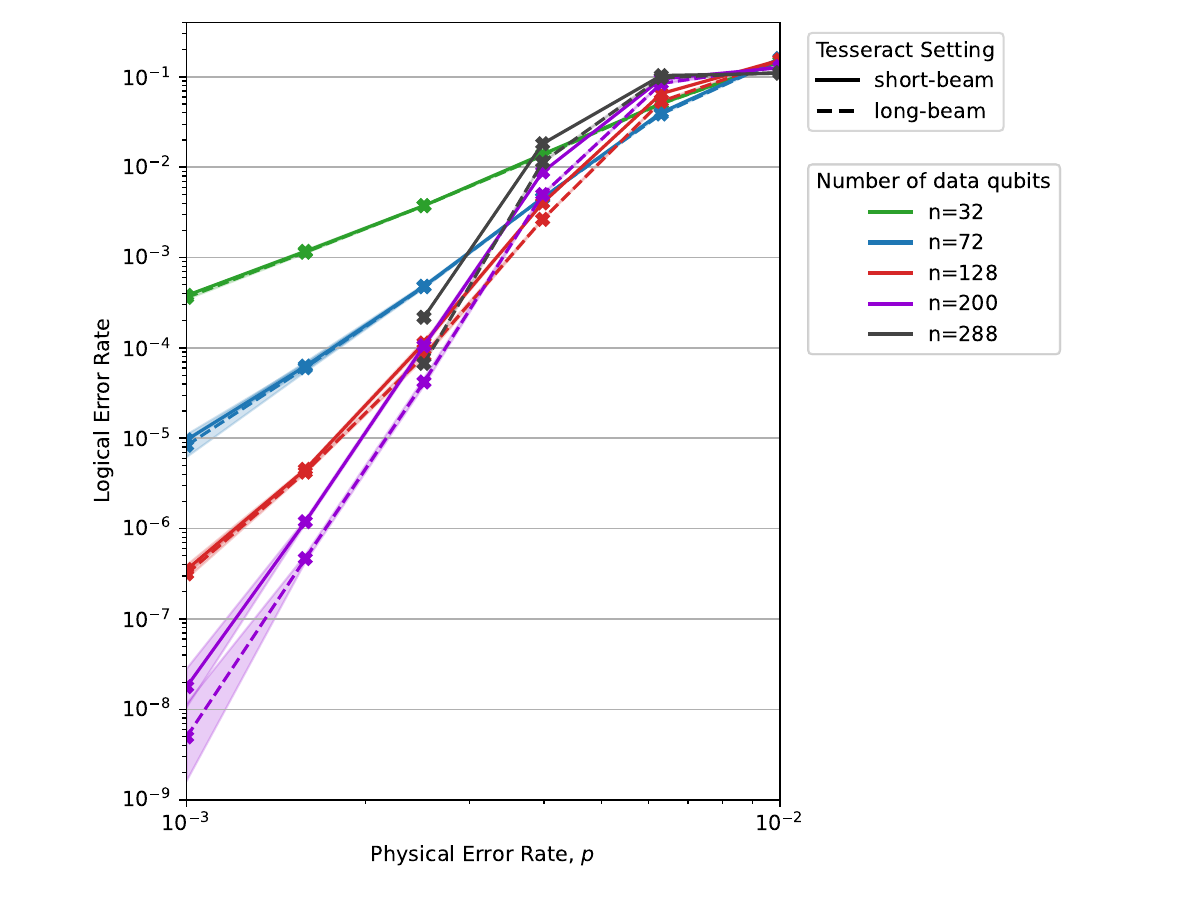}
    \caption{Comparison of the TC's QEC performance with Tesseract's short and long beam settings.}
    \label{fig:TC_short_vs_long_beam}
\end{figure}

As was mentioned in \Cref{sec:simulations}, the gap between the larger subsequent distance $N^2E^3N^2$ directional codes is smaller in \Cref{fig:N2E3N2_vs_RTC} than between that of the lower distance codes. A similar behaviour can be observed for the other directional codes. One possible reason for this is the use of the ``short beam'' setting for the Tesseract decoder \cite{tesseractdecoder}. To demonstrate this, we compare in \Cref{fig:TC_short_vs_long_beam} the performance of the ``short beam'' and ``long beam'' settings of the Tesseract decoder using the exact same TC circuits. More precisely, for $d=4,6,8,10,12$, or equivalently, $n=32,72,128,200,288$, we benchmarked a quantum memory circuit with $d$ QEC rounds. It is clearly visible that at lower $n$ the performance is almost identical. However, as $n$ increases to $200$ and beyond, a visible performance difference emerges where the long beam setting outperforms the short beam setting. Note that for the $n=288$ case, we only simulated for $p\geq 10^{-2.6}$ due to restrictions on compute usage. Based on this, we suspect that in general as circuits get larger, the performance difference between the two beam settings also becomes more pronounced. This effect may additionally impact performance comparisons across code families. In particular, since the size of the circuit for a directional code of a given distance tends to be much larger than that of a single copy of the RTC, TC and RPSC of the same distance, the comparison values from \Cref{tab:slope_ratios} may be improved by using a different beam setting of Tesseract, or another decoder.

\end{document}